\documentclass[10pt,twocolumn]{IEEEtran}

\usepackage[cmex10]{amsmath}
\usepackage{mathrsfs}
\usepackage{cite}
\usepackage{hyperref}
\usepackage[caption=false,font=footnotesize,subrefformat=parens]{subfig}
\usepackage{stfloats}

\usepackage{amssymb,amsthm}
\usepackage{tikz}
\usetikzlibrary{calc,intersections,automata,arrows}

\def \graphicswidth {21pc}

\newtheorem{theorem}{Theorem}

\newtheorem{proposition}{Proposition}
\newtheorem{lemma}{Lemma}

\theoremstyle{definition}
\newtheorem{definition}{Definition}

\theoremstyle{remark}
\newtheorem*{remark}{Remark}

\newcommand{\Cub}{C_{\mathrm{ub}}}

\DeclareMathOperator*{\argmax}{\arg\!\max}

\newcounter{MYtempeqncnt}

\begin{document}
\allowdisplaybreaks

\title{
Capacity of Remotely Powered Communication
}

\author{Dor Shaviv, Ayfer \"{O}zg\"{u}r, and Haim H. Permuter
\thanks{
D.~Shaviv and A.~\"Ozg\"ur were supported by a Robert Bosch Stanford Graduate Fellowship, by the National Science Foundation (NSF) under grant CCF-1618278 and the Center for Science of Information (CSoI), an NSF Science and Technology Center, under grant agreement CCF-0939370.
H.~H.~Permuter was supported by the European Research Council under the European Union's Seventh Framework Programme (FP7/2007-2013) / ERC grant agreement n°337752 and by Israeli Science foundation (ISF).
This work was presented in part at the 2016 IEEE International Symposium on Information Theory (ISIT)~\cite{RPCISIT2016}.}
\thanks{D.~Shaviv and A.~\"{O}zg\"{u}r are with the Department of Electrical Engineering, Stanford University, Stanford, CA 94305, USA (e-mail: shaviv@stanford.edu; aozgur@stanford.edu).
H.~H.~Permuter is with the Department of Electrical and Computer Engineering, Ben-Gurion University of the Negev, Beer-Sheva 84105, Israel (e-mail: haimp@bgu.ac.il).}
}

\maketitle

\begin{abstract}
Motivated by recent developments in wireless power transfer, we study communication with a remotely powered transmitter. We propose an information-theoretic model where a charger can dynamically decide on how much power to transfer to the transmitter based on its side information regarding the communication, while the transmitter needs to dynamically adapt its coding strategy to its instantaneous energy state, which in turn depends on the actions previously taken by the charger. We characterize the capacity as an $n$-letter mutual information rate under various levels of side information available at the charger. When the charger is finely tunable to different energy levels, referred to as a ``precision charger'', we show that these expressions reduce to single-letter form and there is a simple and intuitive joint charging and coding scheme achieving capacity. The precision charger scenario is motivated by the observation that in practice the transferred energy can be controlled by simply changing the amplitude of the beamformed signal. When the charger does not have sufficient precision, for example when it is restricted to use a few discrete energy levels, we show that the computation of the $n$-letter capacity can be cast as a Markov decision process if the channel is noiseless. This allows us to numerically compute the capacity for specific cases and obtain insights on the corresponding optimal policy, or even to obtain closed form analytical solutions by solving the corresponding Bellman equations, as we demonstrate through examples. Our findings provide some surprising insights on how side information at the charger can be used to increase the overall capacity of the system.
\end{abstract}


\section{Introduction}
\label{sec:introduction}

Advancements in radio frequency (RF) power transfer over the recent decades have enabled wireless power transfer over longer distances (see \cite{huang2014enabling} and references therein).\footnote{Power can also be transferred by other modalities such as coherent optical radiation 
but such techniques are currently less common as compared to RF power transfer.}
Combined with synergistic recent developments in wireless communication, such as massive MIMO, small cells and millimeter wave communication, RF power transfer is expected to be one of the dominant modes for powering  Internet of Things (IoT)-type wireless devices in the near future. For example, Fig.~\ref{fig:illustrations} illustrates two topologies considered for indoor IoT applications. In Fig.~\ref{fig:illustrations}-\subref{subfig:receiver_charges}, wireless sensors distributed around a house communicate to a central sink node, which gathers all the information and serves as a gateway to the cloud. While the sink node has access to traditional power, the wireless sensors themselves do not have any traditional batteries. They  harvest the RF energy over the downlink channel in a small rechargeable battery, which allows them to transmit over the uplink. Eliminating the traditional battery at the sensor nodes is desirable for a number of reasons. First, it enables the sensors to operate in a maintenance-free fashion, which in turn enables more flexible deployment models and applications where sensors can be embedded in structures or put in hard to reach places. Second, it allows a significant decrease in the size and cost of these sensors, which is especially important if such IoT applications are to scale to massively large numbers.  

\begin{figure}
\centering
\captionsetup[subfigure]{width=9.5pc}
\vspace{-1em}
\subfloat[]{
\includegraphics[width=9pc]{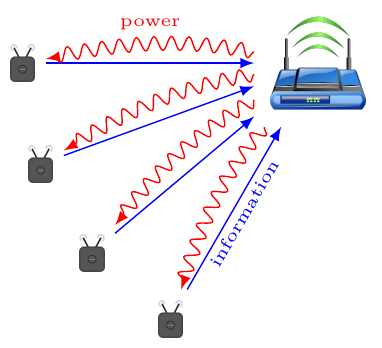}
\label{subfig:receiver_charges}
}
\subfloat[]{
\includegraphics[width=9pc]{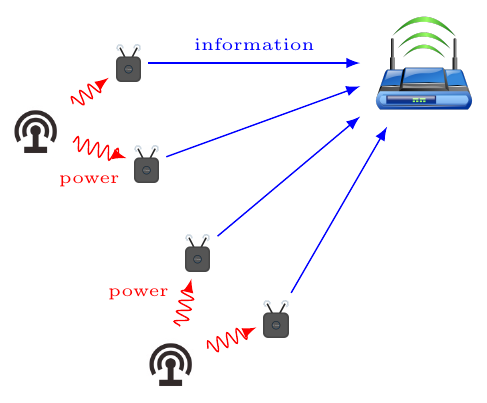}
\label{subfig:power_beacons}
}
\caption{(a) Access point powers nodes over the downlink channel which transmit over the uplink.
(b) Nodes are powered by dedicated power beacons.}
\label{fig:illustrations}
\end{figure}

In some cases, the separation between the sensors and the sink node can be too large to enable efficient wireless power transfer over the downlink channel with existing technologies. An alternative topology considered in this case is to deploy a few dedicated power beacons that have access to traditional power to wirelessly charge nearby sensor nodes while they communicate to the central sink node. 
See Fig.~\ref{fig:illustrations}-\subref{subfig:power_beacons}. Unlike access points or base stations, power beacons do not require any backhaul links, and therefore their low cost can allow dense deployments.  A similar setting arises in biomedical sensing applications where tiny wireless nodes inside the human body can be powered remotely by a wireless charger carried outside  the body.

This paper attempts to model and study such remotely powered communication systems from an information-theoretic perspective. We observe that while there are some straightforward ways to design such systems---for example, a constant level of power can be continuously transferred to the sensors, or the rechargeable batteries of the sensors can be replenished periodically, for example, by beamforming to different spatial clusters of sensors at different times, and these are the approaches currently taken in system implementations~\cite{Arb1journal}---in most cases the charger (being either the sink node in Fig.~\ref{fig:illustrations}-\subref{subfig:receiver_charges} or the power beacons in Fig.~\ref{fig:illustrations}-\subref{subfig:power_beacons}) can have access to instantaneous side information regarding the communication and it can use this side information to transfer power more intelligently to the sensors by operating in a dynamic fashion. For example, in the configuration in Fig.~\ref{fig:illustrations}-\subref{subfig:receiver_charges}, the charger is the sink node, i.e. the receiver itself, and therefore it can potentially utilize  its causal observations of the channel output to make better charging decisions.\footnote{We assume that the sink node does not have any  information of its own to communicate to the sensors and therefore communication is one-way, from the sensors to the sink node. For this reason, we will often refer to the sink node as the receiver and the sensor nodes as the transmitters in the sequel. In practice, the downlink channel is sometimes used to broadcast control and synchronization information. In the case of Fig.~\ref{fig:illustrations}-\subref{subfig:receiver_charges}, this would mean that the downlink channel is used for simultaneous energy and information transfer. This is not the case we consider here.} In Fig.~\ref{fig:illustrations}-\subref{subfig:power_beacons} on the other hand, the proximity of the charger (i.e. the power beacon) to the sensor nodes, i.e. the transmitters, can allow it to almost noiselessly observe the transmitters' inputs to the channel. These settings essentially give rise to an interactive system where the charger can dynamically decide on how much power to transfer to a transmitter based on its side information regarding the communication, while the transmitter needs to dynamically adapt its coding strategy to its instantaneous energy state, which in turn depends on the actions taken by the charger.

From a more fundamental perspective, this new way of powering wireless devices introduces a new paradigm in communications. Traditionally, the encoder and the energy source are co-located on the same device; codewords can be designed ahead of time and  the transmitter can transmit the desired codeword by simply drawing the necessary energy from the energy source readily available on the device. In this new paradigm, however, the encoder is physically separated from its energy source. This introduces novel constraints on the communication and requires both the encoder/transmitter and the energy source/charger to operate in a dynamic fashion. 

The contribution of the current paper is to develop an information-theoretic model for such remotely powered communication systems and $n$-letter expressions for their capacity under various assumptions on the side information available to the charger (Section~\ref{sec:capacity}).
We then proceed to explicitly compute these $n$-letter expressions in two important special cases.
In the first case, the charger is finely tunable to different energy levels, hence it is referred to as a \emph{precision charger}. This case is motivated by the observation that in practice it can be easy to finely control the amount of transferred energy by simply changing the amplitude of the beamformed signal. We show that the $n$-letter expressions reduce to single-letter capacity formulas in this case, and a simple capacity-achieving scheme is provided (Section~\ref{sec:precision_charger}).
The second special case is the noiseless channel. Note that although the capacity of noiseless channels is trivial when the channel is memoryless, this is not the case when memory is present. 
In our model, the input constraint has memory which depends not only on the input, but also on the energy arrivals determined by the charger.
In fact, one of the cases we consider turns out to be an extension of the \emph{constrained coding} problem, initially studied by Shannon in his 1948 paper~\cite{Shannon1948}, and with vast literature on the subject since (see~\cite{MarcusSiegelWolf1992,MarcusRothSiegel2001} for a good introduction). A bulk of the literature in this research line focuses on noiseless channels, while the \emph{noisy} capacity of even the most simple constrained coding systems, such as a binary erasure channel with no consecutive ones, is still open, with only asymptotic results known~\cite{li2016asymptotics}. 
In Section~\ref{sec:noiseless_channel}, we show that the $n$-letter capacity in the noiseless case can be cast as a Markov decision process (MDP), which enables leveraging tools from dynamic programming to efficiently compute it. By solving the Bellman equation either analytically or numerically using the value iteration algorithm, one can find the capacity along with the optimal input distributions, which provide insights into the capacity-achieving strategy. In Section~\ref{sec:example}, we solve the Bellman equation for a specific example, which also provides some surprising insights into the overall capacity gain with different levels of side information at the charger.


\paragraph*{Related Work} Previous work has considered an information-theoretic approach to communication with wireless devices that can harvest their energy from the natural resources in their environment~\cite{OzelUlukus2012,Tutuncuogluetal2013,JogAnantharam2014,
MaoHassibi2013,DongFarniaOzgur2015,ShavivNguyenOzgur2016,
ShavivOzgur2015,ShavivOzgurPermuter2015}. In that case, the energy arrival process is dictated by the outside world and is assumed to follow some given stochastic model, while in our case this process is controlled by the charger. Another related line of work is~\cite{varshney2008transporting,grover2010shannon,tandon2015subblock} which focuses on simultaneous information and energy transfer from the transmitter to the receiver. The setting considered there is the standard point-to-point channel and the emphasis is on designing transmission strategies that are simultaneously good for conveying information and energy. Note that this is different from our model; in our case the charger's signal need not carry any information and the transmitter is remotely powered, while in \cite{varshney2008transporting,grover2010shannon,tandon2015subblock} it is the transmitter charging the receiver.
In a related model~\cite{popovski2013interactive}, two nodes simultaneously transfer energy and information to each other in an interactive fashion.
Finally, there has been significant amount of work in the recent communication theory literature (see \cite{shaviv2015universally,review} and references therein), that focus on resource allocation for maximizing end-to-end throughput in networks where nodes can share both energy and information.
However, throughput optimization does not capture the coding aspect of communication, and this framework is substantially different than the information-theoretic framework.
From a technical perspective, our problem resembles a constrained coding problem where the constraint can be partially controlled.
Reference~\cite{fouladgar2014constrained} suggests using constrained codes for simultaneous energy and information transfer.
Our model is  reminiscent of \cite{
asnani2014feed},
where the encoder and/or the decoder can take \emph{actions} which influence the availability of feedback to the transmitter.
In our case, these actions are taken by the charger (which is not the transmitter nor the receiver), and represent energy transferred to the transmitter, thereby affecting the state of the system, although not directly controlling it.

\section{System Model}
\label{sec:system_model}

We begin by introducing the notation used throughout the paper. Let uppercase, lowercase, and calligraphic letters denote random variables, specific realizations of RVs, and alphabets, respectively.
For two jointly distributed RVs $(X,Y)$, let $p(x)$, $p(x,y)$, and $p(y|x)$, respectively denote the marginal of $X$, the joint distribution of $(X,Y)$, and the conditional distribution of $Y$ given $X$.
We will sometimes use the notation $y(x)$ to mean that $y$ is some deterministic function of~$x$.
Let $\mathbb{E}[\,\cdot\,]$ denote expectation.
For $m\leq n$, $X_m^n=(X_m,X_{m+1},\ldots,X_{n-1},X_n)$, and $X^n=X_1^n$.
Additionally, when the length is clear from the context, we sometimes denote vectors by boldface letters, e.g. $\mathbf{x}\in\mathcal{X}^n$.
All logarithms are to base 2. 

\begin{figure}
\centering
\begin{tikzpicture}
	\def \arlen {1cm};
	\def \blockwidth {0cm}
	\def \blockheight {0.7cm}

	\node[draw,rectangle,minimum width=\blockwidth,
		minimum height=\blockheight] 
		(Tx) at (0,0) {Transmitter};
		
	\node[draw,rectangle,right]
		(Channel) at ($(Tx.east)+(0.75,0)$) {$P_{Y|X}$};
	
	\node[draw,rectangle,right,minimum width=\blockwidth,
		minimum height=\blockheight]
		(Rx) at ($(Channel.east)+(0.75,0)$) {Receiver};
	
	\node[draw,rectangle,above,minimum width=0.7cm,minimum height=1cm] (Battery) at ($(Tx.north)+(0,0.2)$) {};
	\foreach \y in {0.2,0.4,0.6,0.8}
	{
		\draw[-] ($(Battery.south west)+(0,\y)$) -- 
			($(Battery.south east)+(0,\y)$);
	}
	\node[left] at ($(Battery.west)$) {$\bar{B}$};
	
	\node[draw,rectangle,above,minimum width=\blockwidth,
		minimum height=\blockheight]
		(Charger) at ($(Battery.north)+(0,1)$)
		{Charger};
	
	
	\node (M) at ($(Tx.west)-(\arlen,0)$) {$M$};
	
	\tikzstyle{every path}=[->]
	\draw (Battery) -- (Tx);
	\draw (Tx) -- node[below] {$X_t$} (Channel);
	\draw (Channel) -- node[below] {$Y_t$} (Rx);
	\draw (Charger) -- node[left] {$E_t$} (Battery);
	\draw (M) -- (Tx);
	\draw (Rx) -- node[right,pos=1] {$\hat{M}$} 
		($(Rx.east)+(0.5,0)$);
	
	\tikzstyle{every path}=[->,dashed];
	
	\draw (M) |- (Charger);
	
	
	\coordinate (ChannelRx) at 
		($(Channel.east) !.5! (Rx.west)$);
		
	\coordinate (TxChannel) at 
		($(Tx.east) !.5! (Channel.west)$);
	
	
	
	\coordinate (ChargerXConnection) at
		($(Charger.east)-(0,0.15)$);
	
	\coordinate (ChargerYConnection) at
		($(Charger.east)+(0,0.15)$);
	
	\draw (TxChannel) |- 
		(ChargerXConnection);	
		
	\draw (ChannelRx) |- 
		(ChargerYConnection);	
	
	
	
\end{tikzpicture}
\caption{Model of energy harvesting communication with a charger.
The transmitter is equipped with a battery of size $\bar{B}$, and is communicating to the receiver over a memoryless channel. A third terminal, the \emph{charger}, is charging the transmitter's battery via the energy process $E_t$.
The charger may observe some side information, as indicated by the dashed lines, such as the message; the channel input $X_t$; the channel output $Y_t$; or no side information at all.}
\label{fig:channel_model}
\end{figure}
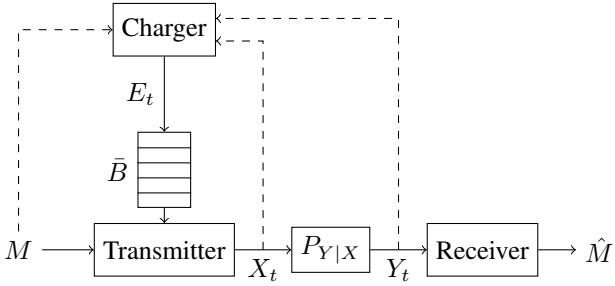

The model is depicted in Fig.~\ref{fig:channel_model}.
The physical channel is a discrete memoryless channel (DMC), with input $X_t\in\mathcal{X}$ and output $Y_t\in\mathcal{Y}$ at time $t$, and transition probability $p(y|x)$.
Additionally, the channel has an associated cost function $\phi:\mathcal{X}\to\mathbb{R}_+$, called the \emph{energy cost function}, denoting the amount of energy used for transmission by each symbol. 
The transmitter has a battery with finite capacity $\bar{B}$,
and this battery is charged by an \emph{energy arrival} $E_t\in\mathcal{E}\subset\mathbb{R}_+$ at each time slot $t$.

Let the amount of energy in the battery at the beginning of time slot $t$ be denoted by $B_t$. We employ a store-and-use model, in which the incoming energy is first stored in the battery before it can be used by the transmitter.
Hence, upon arrival of energy $E_t$, the amount of available energy for transmission is $\min\{B_t+E_t,\bar{B}\}$.
The transmitter observes the battery level $B_t$ and the energy arrival $E_t$, and outputs the input symbol $X_t$, where the symbol energy $\phi(X_t)$ is constrained by the available energy.
At the beginning of the next time slot $t+1$, the amount of energy in the battery is  whatever remains after using $\phi(X_t)$ for transmission.
The input energy constraint and the evolution of energy in the battery can be therefore described via the following \emph{energy constraints}:
\begin{align}
	\phi(X_t)&\leq\min\{B_t+E_t,\bar{B}\}, 
		\label{eq:EH_constraint}\\*
	B_t&=\min\{B_{t-1}+E_{t-1},\bar{B}\}-\phi(X_{t-1}).
		\label{eq:EH_battery}
\end{align}
Without loss of generality, we assume that $B_1=\bar{B}$.\footnote{This is essentially the commonly used store-and-use model in the energy harvesting literature (see for example~\cite{ShavivNguyenOzgur2016}), however we slightly change the notation as it turns out to be more convenient in the development of the MDP formulations later in the paper. In \cite{ShavivNguyenOzgur2016}, $B_t$ signifies the amount of energy in the battery \emph{after} being charged by $E_t$. 
}

We make the following assumptions:
\begin{enumerate}
\item
The symbol energies and arrival energies are non-negative: $\phi(x)\geq0$ for all $x\in\mathcal{X}$ and $e\geq0$ for all $e\in\mathcal{E}$.
\item 
There is at least one symbol $x\in\mathcal{X}$ such that $\phi(x)=0$; we call this the zero symbol and denote it by $x = 0$.
This way, the transmitter will always be able to transmit, even if the battery is completely empty and no energy arrives.
\item
Without loss of generality, we can assume $\phi(x)\leq\bar{B}$ for all $x\in\mathcal{X}$. This is because any $x$ with $\phi(x)>\bar{B}$ can never be transmitted, therefore we can remove it from $\mathcal{X}$ without changing the system.
\item
Similarly, $e\leq\bar{B}$ for all $e\in\mathcal{E}$, i.e. $\mathcal{E}\subseteq[0,\bar{B}]$. This is because any $e\in\mathcal{E}$ which is greater than $\bar{B}$ can be replaced by $e=\bar{B}$ without changing the system.
\item
There is at least one strictly positive $e\in\mathcal{E}$; otherwise, this is a degenerate case in which only symbols with zero energy are allowed.
\end{enumerate}

Based on the physical layout of the system, specifically the location of the charger, it may observe different side information.
This is illustrated in Fig. \ref{fig:channel_model} by the dashed arrows.
At the end of time slot $t-1$, the charger, based on its observations, decides on the energy $E_t$ that will be applied at the beginning of time $t$.
Since we want to account for the energy efficiency of the communication system, we impose an average cost constraint on the energy sequence $E_t$ :
\begin{equation}
\frac{1}{n}\sum_{t=1}^{n}E_t\leq\Gamma,
\label{eq:cost_constraint}
\end{equation}
where $\Gamma\geq0$ is a given cost constraint. Note that even though the charger may not be power-limited, such constraints are often imposed by regulatory bodies.

We define an $(R,n)$ code as a set of messages $\mathcal{M}=\{1,\ldots,2^{nR}\}$, transmitter encoding functions $f_{t}^{\mathrm{Tx}}$, charger encoding functions $f_t^{\mathrm{C}}$, and a decoding function $f^{\mathrm{Rx}}$.
The transmitter encoding functions are given by:
\begin{align}
	f_t^{\mathrm{Tx}}&:\mathcal{M}\times \mathcal{E}^t\to\mathcal{X}\qquad
		,t=1,\ldots,n.\label{eq:channel_encoding}
\end{align}
To transmit message $m\in\mathcal{M}$ at time $t=1,\ldots,n$, the transmitter sets $x_t=f_t^{\mathrm{Tx}}(m,e^t)$, where $e^t$ are the observed energy arrivals up to time $t$.
The battery state $b_t$ is a deterministic function of $(x^{t-1},e^{t-1})$, therefore also of $(m,e^{t-1})$.
The functions $f_t^{\mathrm{Tx}}$ must satisfy the energy constraint~\eqref{eq:EH_constraint} for every possible energy arrivals sequence $e^n$:
\[\phi\big(f_t^{\mathrm{Tx}}(m,e^t)\big)\leq \min\{
	b_t(m,e^{t-1})+e_t,\bar{B}\}.\]

The charger encoding functions $f_t^{\mathcal{C}}$ depend on the side information available to the charger. 
In Appendix~\ref{sec:strategies_capacity_and_special_cases} we provide capacity expressions for the general case where the side information can be an arbitrary function of the channel input and output, however in the main body of the paper we focus on several simpler cases motivated by different settings of practical interest:
\subsubsection{Generic Charger}
The charger does not observe any side information.
The charger encoding functions in this case are simply 
\begin{equation}
f_t^{\mathrm{C}}:\emptyset\to\mathcal{E},
\end{equation}
or in other words, the charger uses a predetermined \emph{fixed} (non-random) charging sequence $e^n$.
Denote the capacity of this system by $C_{\emptyset}$.

\subsubsection{Receiver Charges Transmitter}
This case studies the scenario where the receiver itself charges the transmitter, for example via wireless energy transfer. See Fig. \ref{fig:illustrations}-\subref{subfig:receiver_charges}.
At the end of time $t$, the charger observes $Y^{t-1}$ and decides on the amount of energy that will arrive at the transmitter's battery at the beginning of time $t$.
Hence the charger encoding functions are 
\begin{equation}
f_t^{\mathrm{C}}:\mathcal{Y}^{t-1}\to\mathcal{E}.
\end{equation}
The capacity in this case is denoted by $C_Y$.

\subsubsection{Charger Adjacent to Transmitter}
Suppose the charger is situated at a much closer distance to the transmitter than the actual receiver, as in Fig. \ref{fig:illustrations}-\subref{subfig:power_beacons}. We can model this case by assuming that the charger noiselessly observes the transmitter's inputs to the channel. The charger then has strictly causal observations of the input, therefore the charger encoding functions are 
\begin{equation}
f_t^{\mathrm{C}}:\mathcal{X}^{t-1}\to\mathcal{E}.
\end{equation}
We denote the capacity in this case by $C_X$.

\subsubsection{Fully Cognitive Charger}
In this case, we assume that the charger knows the message to be transmitted ahead of time. Equivalently, the transmitter can be thought of as charging itself using a remote power source. Note that this does not correspond to the conventional average power constraint, since the transmitter still has a finite battery and must satisfy the battery constraints. We assume the charger has full knowledge of the message to be transmitted, hence the charger can decide on its charging sequence ahead of time,
\begin{equation}
f^{\mathrm{C}}:\mathcal{M}\to\mathcal{E}^n.
\end{equation}
We denote the capacity in this case by $C_{M}$.

The charger outputs 
$E_t=f_t^{\mathrm{C}}(\cdot)$, where the appropriate input variable is considered for each of the cases described above.
The energy sequence must satisfy the cost constraint~\eqref{eq:cost_constraint}:
$\frac{1}{n}\sum_{t=1}^{n}f_t^{\mathrm{C}}(\,\cdot\,)\leq\Gamma$.
Note that this must hold for \emph{every} realization of the side information observed by the charger.

Finally, the receiver decoding function is
\begin{equation}
g:\mathcal{Y}^n\to\mathcal{M},
\label{eq:decoding}
\end{equation}
and the receiver sets $\hat{M}=g(Y^n)$.

The probability of error is
\[
P_e=2^{-nR}\sum_{m\in\mathcal{M}}\Pr(\hat{M}\neq m\ |\ m\text{ was transmitted}).\]
A rate $R$ is achievable if there exists a sequence of $(R,n)$ codes such that $P_e\to0$ as $n\to\infty$.
The capacity $C$ is the supremum of all achievable rates.

\section{Capacity}
\label{sec:capacity}

In this section we will state $n$-letter capacity expressions for each of the cases mentioned in the previous section.
To this end, define the following set:
\begin{IEEEeqnarray}{RLll}
\mathcal{A}_n(\Gamma)=\Big\{
&x^n\in\mathcal{X}^n,\ e^n\in\mathcal{E}^n:\nonumber\\*
&b_1=\bar{B},\nonumber\\*
&b_{t+1}=\min\{b_t+e_t,\bar{B}\}-\phi(x_t),\nonumber\\*
&\hspace{85pt}t=1,\ldots,n-1,\nonumber\\*
&\phi(x_t)\leq\min\{b_t+e_t,\bar{B}\},\nonumber\\*
&\hfill t=1,\ldots,n,\nonumber\\*
&\sum_{t=1}^{n}e_t\leq n\Gamma
&&\Big\}.\quad
\label{eq:A_def}
\end{IEEEeqnarray}
This is the set of all transmitter-charger codeword pairs $(x^n,e^n)$ that satisfy the energy and cost constraints \eqref{eq:EH_constraint}--\eqref{eq:cost_constraint}.
Hence we must have
\begin{equation}\label{eq:inputconst}
(X^n,E^n)\in\mathcal{A}_n(\Gamma)\quad\text{a.s.}
\end{equation}
Note that this condition defines a set of allowed input distributions.

Using this condition, we are now ready to give the expressions for capacity.
\begin{theorem}\label{thm:capacity}
The capacity of each of the cases defined in Section \ref{sec:system_model} is given by:
\begin{align}
C_\emptyset(\Gamma)&=\lim_{n\to\infty}\frac{1}{n}
	\max_{\substack{p(x^n),\ e^n:\\ (X^n,e^n)\in\mathcal{A}_n(\Gamma)\text{ a.s.}}}
	I(X^n;Y^n),\label{eq:C0}\\
C_Y(\Gamma)&=\lim_{n\to\infty}\frac{1}{n}
	\max_{\substack{p(x^n\|e^n),\ \{e_t(y^{t-1})\}_{t=1}^{n}:\\
		(X^n,e^n(Y^{n-1}))\in\mathcal{A}_n(\Gamma)\text{ a.s.}}}
	I(X^n\to Y^n),\label{eq:CY}
\intertext{
where $I(X^n\to Y^n)=\sum_{t=1}^{n}I(X^t;Y_t|Y^{t-1})$ is directed information and the maximum is over all functions $\{e_t(y^{t-1})\}_{t=1}^{n}$ and causally conditioned input distributions $p(x^n\|e^n)=\prod_{t=1}^{n}p(x_t|x^{t-1},e^t)$,
}
C_X(\Gamma)&=\lim_{n\to\infty}\frac{1}{n}
	\max_{\substack{p(x^n),\ \{e_t(x^{t-1})\}_{t=1}^{n}:\\
		(X^n,e^n(X^{n-1}))\in\mathcal{A}_n(\Gamma)\text{ a.s.}}}
	I(X^n;Y^n),\label{eq:CX}\\
C_M(\Gamma)&=\lim_{n\to\infty}\frac{1}{n}
	\max_{\substack{p(x^n,e^n):\\ 
	(X^n,E^n)\in\mathcal{A}_n(\Gamma)\text{ a.s.}}}
	I(X^n;Y^n).\label{eq:CM}
\end{align}
\end{theorem}

The theorem characterizes the capacity of the channel as a maximum mutual information rate between the input and the output over all input distributions and charging functions that are consistent with each other in the sense of \eqref{eq:inputconst}. For example, in the case of $C_\emptyset(\Gamma)$, the charger needs to fix a charging sequence $e^n$ ahead of time, and the allowable input distributions can assign positive probability to sequences $x^n$ which are consistent with $e^n$ in the sense that they can be transmitted under $e^n$. In the case of $C_Y(\Gamma)$, note that the causally conditioned input distribution and the charging functions chosen together with the channel transition probabilities induce a joint distribution on $(x^n, e^n)$. This joint distribution is constrained to assign positive probability to only the $(x^n, e^n)$ pairs that are consistent with each other, again in the sense that $x^n$ can be transmitted under $e^n$.

The complete proof is deferred to Appendix~\ref{sec:strategies_capacity_and_special_cases}, where we also provide general capacity expressions when the side information at the charger is an arbitrary function of the input and output. We provide here an outline of the proof:
Achievability follows by coding over blocks of length $n$. In each block we use a random code, generated from the capacity-achieving distribution, which uses only side information available in the current block and ignores side information from previous blocks. This is a legitimate scheme as long as the battery is fully-charged at the beginning of each block. For this purpose, the transmission blocks are interleaved with ``silent times'' of some fixed duration $\ell$, in which the transmitter remains silent and the charger transmits a fixed energy symbol $e_0>0$. By appropriately choosing $\ell$ and $e_0$, we can ensure that the battery will be completely recharged at the beginning of each block.
The average cost of this scheme is at most $\frac{n\Gamma+\ell e_0}{n+\ell}$.
By taking $n\to\infty$, we approach the rates given in Theorem~\ref{thm:capacity} while the average energy cost approaches $\Gamma$.
The converse follows easily from Fano's inequality.
Note that by standard operational arguments, the capacity $C(\Gamma)$ in each of the four cases is non-decreasing, concave, and continuous in $\Gamma$.

As a benchmark, we provide a simple upper bound on capacity, which is valid in all scenarios of charger cognition.
Regardless of the information available at the charger, the total amount of energy it can provide to the transmitter over $n$ channel uses cannot exceed $n\Gamma$.
The transmitter, in turn, is restricted by the same total energy constraint:
$\frac{1}{n}\sum_{t=1}^{n}\phi(X_t)\leq\Gamma$. Hence, the capacity of the system can not exceed the capacity of this channel under a simple average transmit energy constraint. This is made precise in the following proposition.
\begin{proposition}
\label{prop:average_power_upper_bound}
The capacity of the energy harvesting channel with a charger is upper bounded by:
\begin{equation}
\label{eq:Cub}
\Cub(\Gamma)=\max_{\substack{p(x):\\ \mathbb{E}[\phi(X)]\leq\Gamma}}I(X;Y).
\end{equation}
\end{proposition}
The proof is straightforward and will be deferred to Appendix~\ref{sec:average_power_upper_bound_proof}. Note that it is also easy to observe the following ordering between the capacities:
\begin{align*}
C_\emptyset(\Gamma)\leq C_X(\Gamma)\leq C_M(\Gamma)\leq\Cub(\Gamma),\\
C_\emptyset(\Gamma)\leq C_Y(\Gamma)\leq\Cub(\Gamma);
\end{align*}
while there may not be any strict ordering between $C_Y(\Gamma)$, and $C_X(\Gamma) $ and $C_M(\Gamma)$. Note that the capacity $C_{\text{EH}}(\Gamma)$ of a channel that harvests energy from the natural resources in its environment, modeled as a random process with mean $\Gamma$ as in~\cite{OzelUlukus2012,Tutuncuogluetal2013,JogAnantharam2014,
MaoHassibi2013,DongFarniaOzgur2015,ShavivNguyenOzgur2016,
ShavivOzgur2015,ShavivOzgurPermuter2015}, would be even smaller, i.e.,
\[
C_{\text{EH}}(\Gamma) \leq C_\emptyset(\Gamma),
\]
since even with no side information, the charger can emulate the random energy harvesting process.

The capacity expressions provided here are all multi-letter expressions and include an infinite limit, hence, in general, are hard to compute explicitly.
In the sequel, we consider a number of interesting special cases where capacity can be computed.
Specifically, Section~\ref{sec:precision_charger} shows that when the charger has sufficient precision,
the following capacities are all equal: $C_X(\Gamma)=C_M(\Gamma)=\Cub(\Gamma)$. Section \ref{sec:noiseless_channel} shows that, for the noiseless channel, these expressions can be formulated as MDPs, which can then be efficiently computed using dynamic programming.
In Section \ref{sec:example} we derive closed-form capacity formulas for a specific example by using these MDP formulations.

\section{Precision Charger}
\label{sec:precision_charger}

In this section, we consider the case when the charger is finely tunable to different energy levels, which is referred to as a \emph{precision charger}. This case is motivated by the observation that in practice the amount of transferred energy is controlled by the amplitude of the beamformed signal, which can be changed in a continuous fashion. However, in certain cases it is desirable to restrict the power of this beamformed signal to certain regimes, for example due to device limitations or non-linearities in the efficiency of the underlying circuit. These constraints can be modeled by restricting the energy alphabet $\mathcal{E}$, a case  we address in the following sections.
\begin{definition}\label{def:precision_charger}
A channel with a precision charger is a channel in which
$\phi(\mathcal{X})\subseteq\mathcal{E}$,
i.e. for every input symbol $x\in\mathcal{X}$ there exists an energy symbol $e\in\mathcal{E}$ such that $e=\phi(x)$.
\end{definition}
For example, the simple channel with binary input alphabet $\mathcal{X}=\{0,1\}$, cost function $\phi(x)=x$, binary energy alphabet $\mathcal{E}=\{0,1\}$, and an arbitrary output alphabet $\mathcal{Y}$, is a channel with a precision charger.
To take this example a step further, assume for simplicity that the channel is noiseless, i.e. $Y_t=X_t$.
The simplest way to power the transmitter is to charge it with $e_t=1$ every time slot, which will ensure that the transmitter's battery is always full, and therefore the encoder can always transmit its desired symbol, 0 or 1.
This will incur an average energy cost of $\Gamma=1$.
However, note that this is wasteful if the charger has side information regarding the transmission. Note here that even though the charger transfers energy at rate $\Gamma=1$, the encoder will only utilize 1/2 units per time slot on average to communicate at the maximal rate of~1. If the charger is fully cognitive ($C_M$), and therefore knows the corresponding codeword to be transmitted by the encoder, it could charge the encoder only when it intends to transfer a 1, therefore the unconstrained capacity could be achieved with only $\Gamma=1/2$.
Now consider the case where the charger causally observes the input ($C_X$) and is not aware of the message to be transmitted. It is interesting to observe that even in this case, with significantly weaker side information at the charger, the same performance can be achieved: given causal observations of the channel input, the charger can infer the state of the transmitter's battery, therefore only charge the encoder after it transmits a 1, i.e., when its battery is empty. 

The observations above may not hold in general, and in Section~\ref{sec:example} we will provide a concrete example of such a case.
It turns out, however, that Definition~\ref{def:precision_charger} provides a sufficient condition for the above observations to hold. Moreover, in this case the optimal strategy turns out to be a simple extension of the above strategy. The transmitter uses a codebook that is designed to satisfy a simple average energy constraint. It is the charger that ensures that the resultant codewords are always transmittable. At each time slot it simply transfers the exact  amount of energy that was used by the transmitter in the previous time slot, which ensures that the transmitter's battery is full at all times. Note that this strategy can be implemented with both side information $X$ and $M$. We state our result in the following theorem.
\begin{theorem}\label{thm:precision_charger}
For an energy harvesting channel with a precision charger (as in Definition~\ref{def:precision_charger}) the following holds:
\[
C_X(\Gamma)=C_M(\Gamma)=C_{\mathrm{ub}}(\Gamma).
\]
\end{theorem}
Although this paper deals with finite alphabets, Theorem~\ref{thm:precision_charger} can easily be extended to continuous alphabets.
Specifically, as an important example of a channel with a precision charger, consider the additive white Gaussian noise (AWGN) channel with the energy cost function $\phi(x)=x^2$. The input alphabet is the interval $\mathcal{X}=[-\sqrt{\bar{B}},\sqrt{\bar{B}}]$ (recall $\phi(x)\leq\bar{B}$), the output is $Y_t=X_t+N_t$, where $N_t\sim\mathcal{N}(0,1)$, and the energy alphabet is the interval $\mathcal{E}=[0,\bar{B}]$. The condition of Definition \ref{def:precision_charger} holds, and we have
\[
C_X(\Gamma)=C_M(\Gamma)=\Cub(\Gamma)
=\max_{\substack{|X|\leq\sqrt{\bar{B}}\\
	\mathbb{E}X^2\leq\Gamma}}
I(X;Y).
\]
This is the capacity of the Gaussian channel with an amplitude  and an average power constraint, found by Smith~\cite{Smith1971}.

\begin{proof}[Proof of Theorem~\ref{thm:precision_charger}]
First, it is clear that any code for $C_X$ can be applied also when the charger is fully cognitive: at time $t$, by knowing $M$ and $E^{t-1}$, the charger can deduce $X^{t-1}$. Along with Proposition~\ref{prop:average_power_upper_bound}, we have
$C_X(\Gamma)\leq C_M(\Gamma)\leq C_{\mathrm{ub}}(\Gamma)$.
Therefore, it is enough to show $\Cub(\Gamma)\leq C_X(\Gamma)$.

For this purpose, consider the conventional DMC $p(y|x)$ with an average input cost constraint $\sum_{t=1}^{n}\phi(X_t)\leq n\Gamma$.
The capacity of this channel is well-known to be 
(see e.g. \cite[Theorem~3.2]{ElGamalKim2011}):
\[\Cub(\Gamma)=\max_{p(x):\, \mathbb{E}[\phi(X)]\leq\Gamma}I(X;Y).\]
We will show that any code for this channel can also be applied for the channel in Definition~\ref{def:precision_charger} when the charger has causal observations of the input.

An $(R,n)$ code for the channel with average input cost constraint consists of a set of codewords $x^n(m)$, $m=1,\ldots,2^{nR}$, such that $\sum_{t=1}^{n}\phi(x_t(m))\leq n\Gamma$ for each $m$, and a decoding function $\hat{m}(y^n)$.
Consider the following $(R,n)$ code for the energy harvesting channel with a precision charger:
\begin{align}
f_t^{\mathrm{Tx}}(m,e^t) &= x_t(m),\nonumber\\
f_t^{\mathrm{C}}(x^{t-1}) &= 
	\begin{cases}
		0&,t=1,\\
		\phi(x_{t-1})&,t=2,\ldots,n,
	\end{cases}
	\label{eq:precision_charger_v_def}\\
f^{\mathrm{Rx}}(y^n) &= \hat{m}(y^n).\nonumber
\end{align}
Note that the symbol $e=\phi(x_{t-1})$ must exist in $\mathcal{E}$ by Definition~\ref{def:precision_charger}, and accordingly the energy symbol $e=0$ must exist by the existence of a zero input symbol $x=0$ (see Section~\ref{sec:system_model}).
Under this scheme, the charger simply recharges the battery every time slot, by charging exactly the amount of energy that was used by the transmitter.

Clearly, since the underlying physical channel is the same,
the probability of error will be the same as for the channel with average input cost constraint.
Therefore we only need to verify that this code is admissible, i.e. it satisfies the energy and cost constraints \eqref{eq:EH_constraint}--\eqref{eq:cost_constraint}.
First, observe that our scheme guarantees that $B_t+E_t=\bar{B}$ for $t=1,\ldots,n$. This is obvious for $t=1$ by our assumption that $B_1=\bar{B}$ and since $e_1=0$.
For $t>1$:
\begin{align*}
B_t+E_t&=\min\{B_{t-1}+E_{t-1},\bar{B}\}-\phi(X_{t-1})+E_t\\
&=\min\{B_{t-1}+E_{t-1},\bar{B}\}\\
&=\bar{B}.
\end{align*}
By our assumption that $\phi(x)\leq\bar{B}$ for all $x\in\mathcal{X}$, this implies that the energy constraints \eqref{eq:EH_constraint} and \eqref{eq:EH_battery} are always satisfied.

We are left with the average energy cost constraint~\eqref{eq:cost_constraint}, which, according to our construction \eqref{eq:precision_charger_v_def}, is given by:
\begin{align*}
\sum_{t=1}^{n}E_t=
\sum_{t=1}^{n-1}\phi(X_{t})
\leq n\Gamma,
\end{align*}
since the codewords must satisfy the average input cost constraint of the original $(R,n)$ code.
We deduce that $\Cub(\Gamma)\leq C_X(\Gamma)$, which concludes the proof.
\end{proof}

\section{Noiseless Channel}
\label{sec:noiseless_channel}

In this section we consider another special case of interest -- the noiseless communication channel, namely $Y_t=X_t$, which can serve as an approximation for high SNR scenarios. In particular, we show that  in this case the computation of the $n$-letter capacity formulas in Theorem~\ref{thm:capacity} can be cast as a Markov decision process which enables the use of dynamic programming techniques to numerically compute the capacity. The result of Theorem~\ref{thm:precision_charger} may lead one to wonder if the average power upper bound in \eqref{eq:Cub} can always be achieved with side information regarding the message, maybe even the channel input.  In Section \ref{sec:example}, we illustrate that by solving the Bellman equations for the Markov decision processes developed in this section, we can obtain closed-form capacity expressions for specific channels, which reveal that in general $C_X(\Gamma)<C_M(\Gamma)<\Cub(\Gamma)$.

We begin with some preliminaries, 
including a Lagrange multipliers formulation of the capacity expressions in Section~\ref{subsec:noiseless_capacities},
and a brief overview of Markov decision processes in Section~\ref{subsec:MDPs}.
We then continue to an MDP formulation for 
$C_\emptyset$, $C_X$, and $C_M$, in Sections~\ref{subsec:universal_charger_noiseless}, \ref{subsec:charger_adj_Tx_noiseless},
and~\ref{subsec:charger_knows_M_noiseless}, respectively.
Note that $C_X=C_Y$ for a noiseless channel.

\subsection{Capacity Expressions and Lagrange Multipliers}
\label{subsec:noiseless_capacities}

First, we claim that the constraint $\frac{1}{n}\sum_{t=1}^{n}E_t\leq\Gamma\text{ a.s.}$ in \eqref{eq:C0}--\eqref{eq:CM} can be relaxed to hold in expectation, i.e. $\frac{1}{n}\sum_{t=1}^{n}\mathbb{E}[E_t]\leq\Gamma$.
\begin{lemma}
\label{lemma:relaxed_power_constraint}
For the noiseless channel, the capacities of each of the cases defined in Section~\ref{sec:system_model} are given by
\begin{align}
C_\emptyset(\Gamma)&=\lim_{n\to\infty}\frac{1}{n}
	\max_{\substack{p(x^n),\ e^n:\\ (X^n,e^n)\in\mathcal{A}'_n\text{ a.s.}\\
	\sum_{t=1}^{n}e_t\leq n\Gamma}}
	H(X^n),\label{eq:C0_relaxed}\\
C_X(\Gamma)&=\lim_{n\to\infty}\frac{1}{n}
	\max_{\substack{p(x^n),\ \{e_t(x^{t-1})\}_{t=1}^{n}:\\
		(X^n,e^n(X^{n-1}))\in\mathcal{A}'_n\text{ a.s.}\\
	\sum_{t=1}^{n}\mathbb{E}[E_t]\leq n\Gamma}}
	H(X^n),\label{eq:CX_relaxed}\\
C_M(\Gamma)&=\lim_{n\to\infty}\frac{1}{n}
	\max_{\substack{p(x^n,e^n):\\ 
	(X^n,E^n)\in\mathcal{A}'_n\text{ a.s.}\\
	\sum_{t=1}^{n}\mathbb{E}[E_t]\leq n\Gamma}}
	H(X^n).\label{eq:CM_relaxed}
\end{align}
where $\mathcal{A}'_n$ denotes all feasible $(x^n,e^n)$ pairs without an average energy cost constraint:
\begin{IEEEeqnarray}{RLll}
\mathcal{A}'_n=\Big\{
&x^n\in\mathcal{X}^n,\ e^n\in\mathcal{E}^n:\nonumber\\*
&b_1=\bar{B},\nonumber\\*
&b_{t+1}=\min\{b_t+e_t,\bar{B}\}-\phi(x_t),\nonumber\\*
&\hfill t=1,\ldots,n-1,\nonumber\\*
&\phi(x_t)\leq\min\{b_t+e_t,\bar{B}\},\quad
	t=1,\ldots,n
	&\Big\}.
\label{eq:Anprime_def}
\end{IEEEeqnarray}
\end{lemma}
Note that \eqref{eq:C0_relaxed} is merely a change of notation from \eqref{eq:C0}.
The proof of \eqref{eq:CX_relaxed} and \eqref{eq:CM_relaxed} is given in Appendix~\ref{sec:expected_average_power_constraint}.
We note that a similar result can be shown also for the noisy channel, however this will not be needed here.


Next, we use the method of Lagrange multipliers to convert the constrained optimization problems \eqref{eq:C0_relaxed}--\eqref{eq:CM_relaxed} to unconstrained ones.
At the same time, we show that we can switch the order of limit and maximization in the resulting formulas. This will be necessary in order to formulate the problems as infinite horizon average reward MDPs.
In the following,
we denote infinite sequences $\{(x_t,e_t)\}_{t=1}^{\infty}$ by boldface letters $(\mathbf{x},\mathbf{e})$, and we say $(\mathbf{x},\mathbf{e})\in\mathcal{A}'_\infty$ if the subsequences $(x^n,e^n)\in\mathcal{A}'_n$ for every $n\geq1$.
For $\rho\geq0$, let\small
\begin{align}
J_\emptyset(\rho)&=\sup_{\substack{\{p(x_t|x^{t-1}),e_t\}_{t=1}^{\infty}:\\ (\mathbf{X},\mathbf{e})\in\mathcal{A}'_\infty\text{ a.s.}}}
\liminf_{n\to\infty}\frac{1}{n}
\big(H(X^n)-\rho\sum_{t=1}^{n}e_t\big),
\label{eq:J0_def}\\
J_X(\rho)&=\sup_{\substack{\{p(x_t|x^{t-1}),e_t(x^{t-1})\}_{t=1}^{\infty}:\\ (\mathbf{X},\mathbf{E})\in\mathcal{A}'_{\infty}\text{ a.s.}}}
\liminf_{n\to\infty}\frac{1}{n}
\big(H(X^n)-\rho\sum_{t=1}^{n}\mathbb{E}[E_t]\big),
\label{eq:JX_def}\\
J_M(\rho)&=\sup_{\substack{\{p(x_t,e_t|x^{t-1},e^{t-1})\}_{t=1}^{\infty}:\\ (\mathbf{X},\mathbf{E})\in\mathcal{A}'_{\infty}\text{ a.s.}}}
\liminf_{n\to\infty}\frac{1}{n}
\big(H(X^n)-\rho\sum_{t=1}^{n}\mathbb{E}[E_t]\big).
\label{eq:JM_def}
\end{align}\normalsize
For any $\epsilon>0$,
let the processes $(\mathbf{X}_\emptyset,\mathbf{E}_\emptyset)$, $(\mathbf{X}_X,\mathbf{E}_X)$, and $(\mathbf{X}_M,\mathbf{E}_M)$, approach the supremum in~\eqref{eq:J0_def}, \eqref{eq:JX_def}, and \eqref{eq:JM_def}, respectively, up to $\epsilon$.
Let
\begin{align*}
\underline{\Gamma}_\emptyset&=\liminf_{n\to\infty}\frac{1}{n}\sum_{t=1}^{n}\mathbb{E}[E_{\emptyset,t}],\\
\overline{\Gamma}_\emptyset&=\limsup_{n\to\infty}\frac{1}{n}\sum_{t=1}^{n}\mathbb{E}[E_{\emptyset,t}],
\end{align*}
and similar definitions for $\underline{\Gamma}_X$, $\overline{\Gamma}_X$, $\underline{\Gamma}_M$, and $\overline{\Gamma}_M$.
For $0\leq\alpha\leq1$, define $\Gamma_\emptyset^\alpha=\alpha\underline{\Gamma}_\emptyset+(1-\alpha)\overline{\Gamma}_\emptyset$, and similarly $\Gamma_X^\alpha$ and $\Gamma_M^\alpha$.
\begin{lemma}[Lagrange multipliers]
\label{lemma:Lagrange_multipliers}
For any $0\leq\alpha\leq 1$, the capacities \eqref{eq:C0_relaxed}--\eqref{eq:CM_relaxed}
are bounded by
\begin{IEEEeqnarray}{rCcCl}
J_{\emptyset}(\rho)-\epsilon&\leq&
C_{\emptyset}(\Gamma_\emptyset^\alpha)-\rho\Gamma_\emptyset^\alpha&\leq& J_{\emptyset}(\rho),
\label{eq:C0_bounds}\\
J_X(\rho)-\epsilon&\leq&
C_X(\Gamma_X^\alpha)-\rho\Gamma_X^\alpha&\leq&
J_X(\rho),
\label{eq:CX_bounds}\\
J_M(\rho)-\epsilon&\leq&
C_M(\Gamma_M^\alpha)-\rho\Gamma_M^\alpha&\leq&
J_M(\rho).
\label{eq:CM_bounds}
\end{IEEEeqnarray}
If the supremum in \eqref{eq:J0_def}--\eqref{eq:JM_def} is a maximum (i.e. $\epsilon=0$), then \eqref{eq:C0_bounds}--\eqref{eq:CM_bounds} become equalities and capacity is achieved by the sequence of marginals $\{p(x^n,e^n)\}_{n=1}^{\infty}$ of $(\mathbf{X}_\emptyset,\mathbf{E}_\emptyset)$, $(\mathbf{X}_X,\mathbf{E}_X)$, and $(\mathbf{X}_M,\mathbf{E}_M)$, respectively.
\end{lemma}
This is based on ideas that are common practice in the optimization literature, e.g.~\cite[Theorem 28.1]{rockafellar1970convex}.
See Appendix~\ref{sec:Lagrange_multipliers_proof} for the proof.

In the next section, we will give a brief description of the theory of Markov decision processes, and provide some tools that will help solve optimization problems of the form of {\eqref{eq:J0_def}--\eqref{eq:JM_def}}.

\subsection{Markov Decision Processes}
\label{subsec:MDPs}

Next, we bring a brief description of Markov decision processes.
An MDP is defined by a tuple $(\mathcal{S}$, $\mathcal{U}$, $\mathcal{W}$, $f$, $P_w$, $g)$.
The system evolves according to $s_{t+1}=f(s_t,u_t,w_t)$, $t=1,2,\ldots$
The \emph{state} $s_t$ takes values in a Borel space $\mathcal{S}$, called the \emph{state space}. We assume the initial state $s_1$ is fixed.
The \emph{action} $u_t$ takes values in a Borel space $\mathcal{U}(s_t)$, called the \emph{action space}, which may depend on $s_t$.
The \emph{disturbance} $w_t$ takes values in a measurable space $\mathcal{W}$, and is distributed according to $P_w(\cdot|s_t,u_t)$.

The history $h_t=(s_1,w_1,w_2,\ldots,w_{t-1})$ consists of the information available to a \emph{controller} at time $t$, which chooses the action $u_t$.
The action is chosen by a mapping of histories to actions: $u_t=\mu_t(h_t)$.
The collection of such mappings $\pi=\{\mu_1,\mu_2,\ldots\}$ is called a \emph{policy}.
Note that given a policy $\pi$ and the history $h_t$, one can deduce all past actions $u^{t}$ and states $s^t$.
A policy is called \emph{stationary} if there is a function $m:\mathcal{S}\to\mathcal{U}$ s.t. $\mu_t(h_t)=m(s_t)$, i.e. the action at time $t$ depends only on the current state.
The reward is a bounded function $g:\mathcal{S}\times\mathcal{U}\to\mathbb{R}$.

The optimal \emph{finite-horizon} expected reward for $n$ stages and initial state $s_1=s$ is:
\begin{equation}
J_n(s)=\sup_{\mu_1,\ldots,\mu_n}\sum_{t=1}^{n}\mathbb{E}[g(S_t,\mu_t(H_t))|S_1=s].\label{eq:finite_horizon_def}
\end{equation}
The Bellman principle of optimality~\cite[Ch.~1.3]{Bertsekas2001vol1} states that $J_n(s)$ can be computed recursively as follows:
\begin{equation}
J_{n+1}(s)=\sup_{u\in\mathcal{U}(s)}\left\{g(s,u)
+\mathbb{E}[J_n(f(s,u,W))]\right\},
\label{eq:value_iteration_def}
\end{equation}
where the expectation is w.r.t. the conditional distribution $P_w(\,\cdot\,|s,u)$.
This recursive algorithm is called \emph{value iteration}.
This is a remarkable property of MDPs: it enables computation of $J_n(s)$ with complexity which is \emph{linear} in $n$, and not exponential as initially suggested by~\eqref{eq:finite_horizon_def}.

More important for our problem is the \emph{infinite-horizon average expected reward},
which we wish to maximize over all policies:
\begin{equation}
J(s)=\sup_{\pi}\liminf_{n\to\infty}\frac{1}{n}\sum_{t=1}^{n}
\mathbb{E}\left[g(S_t,\mu_t(H_t))|S_1=s\right]
\label{eq:average_reward_def}
\end{equation}
Note that if one can show that $J(s)=\lim_{n\to\infty}\frac{1}{n}J_n(s)$ then the value iteration algorithm can be used to numerically compute $J(s)$ in an efficient manner.
Moreover, the following theorem provides a mechanism for finding the optimal infinite-horizon average reward, along with the optimal policy that achieves it.
\begin{theorem}[{Bellman Equation \cite[Theorem~6.1]{arapostathis1993discrete}}]
\label{thm:Bellman}
If there exist a scalar $\lambda\in\mathbb{R}$ and a bounded function $h:\mathcal{S}\to\mathbb{R}$ that satisfy:
\begin{equation}
\lambda+h(s)
=\sup_{u\in\mathcal{U}(s)}\Big\{
g(s,u)+\mathbb{E}\big[h(f(s,u,W))\big]\Big\}
\quad\forall s\in\mathcal{S},
\label{eq:Bellman}
\end{equation}
where the expectation is w.r.t. the conditional distribution $P_w(\cdot|s,u)$,
then the optimal average reward does not depend on the initial state $s$ and it is given by $J=\lambda$.
Furthermore, if there exists a function $u^\star(s)$ that attains the supremum in~\eqref{eq:Bellman}, then there exists an optimal stationary policy in \eqref{eq:average_reward_def} which is given by $u^\star(s)$, i.e. $\mu_t(H_t)=u^\star(b_t)$.
\end{theorem}

In what follows, we will show how each of the capacity expressions $C_\emptyset(\Gamma)$, $C_X(\Gamma)$, and $C_M(\Gamma)$, can be cast as an MDP.

\subsection{Generic Charger}
\label{subsec:universal_charger_noiseless}

Let $\mathcal{B}$ be the set of all possible states for the battery that can be reached by some finite transmitter and charger sequences $(x^n,e^n)$ under the assumption $B_1=\bar{B}$.
We will assume that $\mathcal{B}$ is finite.
This is true, for example, when all problem parameters are rational numbers:
Let $\phi(\mathcal{X})=\{\phi(x):\ x\in\mathcal{X}\}$.
If $\phi(\mathcal{X})$ and $\mathcal{E}$ contain only rational numbers, as well as $\bar{B}$ is a rational number, then we can scale all quantities of interest $(\phi(X),E,\bar{B},\Gamma)$ by a large enough integer without changing the problem.
When the above parameters are all integers, clearly $\mathcal{B}\subseteq\{0,1,\ldots,\bar{B}\}$, hence $\mathcal{B}$ is finite and has at most $\bar{B}+1$ elements.

By Lemma~\ref{lemma:Lagrange_multipliers}, we will attempt to solve the following, for all $\rho\geq0$:
\begin{equation}
J_\emptyset(\rho)=\sup_{\substack{\{p(x_t|x^{t-1}),\ e_t\}_{t=1}^{\infty}:\\
(\mathbf{X},\mathbf{e})\in\mathcal{A}'_{\infty}\text{ a.s.}}}\liminf_{n\to\infty}\frac{1}{n}\big(H(X^n)-\rho\sum_{t=1}^{n}e_t\big).
\label{eq:J0_entropy}
\end{equation}
For a fixed $e^n\in\mathcal{E}^n$ and $b\in\mathcal{B}$, denote by $N_n(b)$ the number of sequences $x^n$ for which $(x^n,e^n)\in\mathcal{A}'_n$ and $B_{n+1}=b$, where $\mathcal{A}'_n$ was defined in~\eqref{eq:Anprime_def}.
That is, $N_n(b)$ is the total number of sequences that satisfy the energy constraints and end in battery state $b$.
Denote the $|\mathcal{B}|$-dimensional row vector 
$\mathbf{N}_n=[N_n(b):\ b\in\mathcal{B}]$.
Under this notation, we have:
\begin{align*}
\max_{\substack{p(x^n):\\ (X^n,e^n)\in\mathcal{A}'_n\text{ a.s.}}}H(X^n)
&=\log\Big(\sum_{b\in\mathcal{B}}N_n(b)\Big)\\*
&=\log(\mathbf{N}_n\cdot\mathbf{1}),
\end{align*}
where $\mathbf{1}$ is a $|\mathcal{B}|$-dimensional column vector of 1's.
This suggests that we should be able to replace the maximal $H(X^n)$ with $\log \mathbf{N}_n\cdot\mathbf{1}$.
Indeed, it is shown in Appendix~\ref{sec:generic_charger_num_of_seq} that \eqref{eq:J0_entropy} is equivalent to
\begin{equation}
J_\emptyset(\rho)=\sup_{\{e_t\}_{t=1}^{\infty}}
\liminf_{n\to\infty}\frac{1}{n}\big(\log(\mathbf{N}_n\cdot\mathbf{1})-\rho\sum_{t=1}^{n}e_t\big).
\label{eq:J0}
\end{equation}

To evaluate this expression, we propose a way to recursively compute $\mathbf{N}_n$ from $e^n$ and $\mathbf{N}_{n-1}$.
To this end, with each channel specified by $(\mathcal{X},\mathcal{E},\phi(\mathcal{X}),\bar{B})$ we associate
a set of $|\mathcal{E}|$ labeled graphs, denoted by $G_e$ for each $e\in\mathcal{E}$.
Each graph has $|\mathcal{B}|$ vertices, one for each state $b\in\mathcal{B}$.
In graph $G_e$, we draw an edge labeled $x$ from node $b$ to node $b'$ if $x\in\mathcal{X}$ and $b'=\min\{b+e,\bar{B}\}-\phi(x)$.
An example of such a set of graphs is plotted in Fig.~\ref{fig:battery_state_graphs} for a channel with $\mathcal{X}=\mathcal{E}=\{0,1,2\}$, $\phi(x)=x$, and $\bar{B}=2$.
\begin{figure}
\centering
\def \nodedistance {1.5cm}
\subfloat[Graph $G_0$.]{
\begin{tikzpicture}[->,>=stealth',shorten >=1pt,semithick,font=\tiny]
\tikzstyle{every state}=[draw=none,fill=blue!40]
\node[state] (b2) {$b=2$};
\node[state] (b1) at (15:\nodedistance) {$b=1$};
\node[state] (b0) at (75:\nodedistance) {$b=0$};

\path
(b0) edge [in=90,out=120,loop] node [above] {0} ()
(b1) edge [in=-30,out=0,loop] node [right] {0} ()
	 edge [bend right] node [above right] {1} (b0)
(b2) edge [in=-150,out=-120,loop] node [below] {0} (b2)
	 edge [bend left] node [left] {2} (b0)
	 edge [bend right] node [below] {1} (b1);
\end{tikzpicture}
\label{subfig:G0}
}
\subfloat[Graph $G_1$.]{
\begin{tikzpicture}[->,>=stealth',shorten >=1pt,semithick,font=\tiny]
\tikzstyle{every state}=[draw=none,fill=blue!40]
\node[state] (b2) {$b=2$};
\node[state] (b1) at (15:\nodedistance) {$b=1$};
\node[state] (b0) at (75:\nodedistance) {$b=0$};

\path
(b0) edge [in=90,out=120,loop] node [above] {1} ()
	 edge node [below left] {0} (b1)
(b1) edge [in=-30,out=0,loop] node [right] {1} ()
     edge [bend right] node [above right] {2} (b0)
     edge node [above] {0} (b2)
(b2) edge [in=-150,out=-120,loop] node [below] {0} (b2)
	 edge [bend left] node [left] {2} (b0)
	 edge [bend right] node [below] {1} (b1);
\end{tikzpicture}
\label{subfig:G1}
}

\subfloat[Graph $G_2$.]{
\begin{tikzpicture}[->,>=stealth',shorten >=1pt,semithick,font=\tiny]
\tikzstyle{every state}=[draw=none,fill=blue!40]
\node[state] (b2) {$b=2$};
\node[state] (b1) at (15:\nodedistance) {$b=1$};
\node[state] (b0) at (75:\nodedistance) {$b=0$};

\path
(b0) edge [in=90,out=120,loop] node [above] {2} ()
	 edge node [below left] {1} (b1)
	 edge node [right] {0} (b2)
(b1) edge [in=-30,out=0,loop] node [right] {1} ()
     edge [bend right] node [above right] {2} (b0)	     	 edge node [above] {0} (b2)
(b2) edge [in=-150,out=-120,loop] node [below] {0} (b2)
	 edge [bend left] node [left] {2} (b0)
	 edge [bend right] node [below] {1} (b1);
\end{tikzpicture}
\label{subfig:G2}
}
\caption{Battery state graphs for the noiseless channel with alphabets $\mathcal{X}=\mathcal{E}=\{0,1,2\}$ and $\bar{B}=2$. Each graph represents a different value of $e_t$, the states represent the battery level $b_t$, and the edge labels represent $x_t$.
The output state is $b_{t+1}$ as given by~\eqref{eq:EH_battery}.
}
\label{fig:battery_state_graphs}
\end{figure}
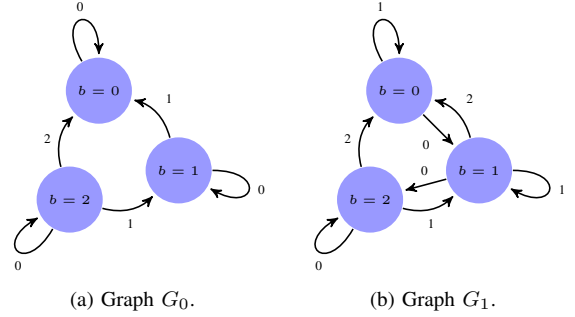
If at time $t$ the battery state is $b_t=b$ and the charging energy is $e_t=e$, then all possible values of $x_t$ are represented by the outgoing edges of state $b$ in $G_e$. By following the edge labeled $x$, we will obtain the battery state $b_{t+1}$ if $x_{t}=x$ was transmitted.

Given $e_t$ and $\mathbf{N}_{t-1}$, we can compute $N_t(b')$ for any $b'\in\mathcal{B}$ by summing up $N_{t-1}(b)$ for all edges going into state $b'$ in graph $G_{e_t}$.
Continuing with our example, if $e_t=0$ we can compute $N_t(0)$ from Fig.~\ref{fig:battery_state_graphs}-\subref{subfig:G0}:
\begin{align*}
N_t(0)&=N_{t-1}(0)+N_{t-1}(1)+N_{t-1}(2),
\end{align*}
since there is an edge going into node $b=0$ from every other node in the graph.

In general, we can write $\mathbf{N}_t$ in vector notation as \begin{equation}
\mathbf{N}_t=\mathbf{N}_{t-1}\cdot\mathbf{A}_{e_t},
\label{eq:Nt_vector_relation}
\end{equation}
where $\mathbf{A}_e$ is the adjacency matrix of graph $G_e$, defined by $(A_e)_{b,b'}=1$ if there is an edge going from $b$ into $b'$ and $0$ otherwise.
For example, the adjacency matrix of the graph in Fig.~\ref{fig:battery_state_graphs}-\subref{subfig:G0} is
\begin{equation*}
\mathbf{A}_0=\begin{bmatrix}
1&0&0\\
1&1&0\\
1&1&1
\end{bmatrix}.
\end{equation*}

\begin{table}
\renewcommand{\arraystretch}{1.5}
\caption{MDP Formulation of Generic Charger Capacity}
\label{tab:universal_charger_general_DP}
\centering
\begin{tabular}{|l|l|}
\hline
state
&$\mathbf{s}_t=\frac{\mathbf{N}_{t-1}}{\mathbf{N}_{t-1}\cdot\mathbf{1}}$\\
\hline
state space& $\mathcal{S}$, the probability simplex in $\mathbb{R}^{|\mathcal{B}|}$\\
\hline
action&$e_t$ (energy at time $t$)\\
\hline
action space&$\mathcal{E}$ (energy alphabet)\\
\hline
reward&$g(\mathbf{s},e)=\log(\mathbf{s}\mathbf{A}_{e}\mathbf{1})-\rho e$\\
\hline
state dynamics&
$\mathbf{s}_{t+1}=f(\mathbf{s}_t,e_t)
=\frac{\mathbf{s}_t\mathbf{A}_{e_t}}{\mathbf{s}_t\mathbf{A}_{e_t}\mathbf{1}}$\\
\hline
\end{tabular}
\end{table}

This recursive relation suggests we can optimize sequentially over $e_t$, i.e. given the optimal $e_{t-1}$ and $\mathbf{N}_{t-1}$, we can choose $e_t$ that maximizes $\log(\mathbf{N}_t\cdot\mathbf{1})-\rho\sum_{t=1}^{n}e_t$ in \eqref{eq:J0}.
More precisely, define
\[
r_t\triangleq\log\frac{\mathbf{N}_t\cdot\mathbf{1}}
{\mathbf{N}_{t-1}\cdot\mathbf{1}}.
\]
Then under the convention $\mathbf{N}_0=[0\ 0\ \cdots\ 0\ 1]$ (which conforms to our assumption that $B_1=\bar{B}$),
we have the telescoping series
$\log(\mathbf{N}_n\cdot\mathbf{1})=\sum_{t=1}^{n} r_t$.
Defining $g_t=r_t-\rho e_t$, eq.~\eqref{eq:J0} can be equivalently expressed in terms of $g_t$ as:
\[
J_\emptyset(\rho)=\sup_{\{e_t\}_{t=1}^{\infty}}
\liminf_{n\to\infty}\frac{1}{n}\sum_{t=1}^{n}g_t.
\]
We will show that this is an infinite-horizon average-reward Markov decision process.
Define the state vector:
\begin{equation}
\mathbf{s}_t\triangleq\frac{\mathbf{N}_{t-1}}{\mathbf{N}_{t-1}\cdot\mathbf{1}}.
\label{eq:st_def}
\end{equation}
Observe that $\mathbf{s}_t\cdot\mathbf{1}=1$ and $s_t(b)\geq 0$ for all $b\in\mathcal{B}$, so $\mathbf{s}_t$ is in fact a probability distribution on the set $\mathcal{B}$.
Let $e_t$ be the action at time $t$.
The reward at time $t$ is given as a function of $\mathbf{s}_t$ and $e_t$:
\begin{align*}
g_t&=\log\frac{\mathbf{N}_t\cdot\mathbf{1}}{\mathbf{N}_{t-1}\cdot\mathbf{1}}-\rho e_t\\
&\overset{\text{(i)}}{=}\log\frac{\mathbf{N}_{t-1}\cdot\mathbf{A}_{e_t}\cdot\mathbf{1}}
{\mathbf{N}_{t-1}\cdot\mathbf{1}}-\rho e_t\\
&\overset{\text{(ii)}}{=}\log(\mathbf{s}_{t}\mathbf{A}_{e_t}\mathbf{1})-\rho e_t,
\end{align*}
where (i) is due to \eqref{eq:Nt_vector_relation} and (ii) is by \eqref{eq:st_def}.
The state at time $t+1$ is given by:
\begin{align*}
\mathbf{s}_{t+1}&=\frac{\mathbf{N}_{t}}{\mathbf{N}_{t}\mathbf{1}}\\
&=\frac{\mathbf{N}_{t-1}\mathbf{A}_{e_t}}{\mathbf{N}_{t-1}\mathbf{A}_{e_t}\mathbf{1}}\\
&=\frac{\mathbf{s}_t\mathbf{A}_{e_t}}{\mathbf{s}_{t}\mathbf{A}_{e_t}\mathbf{1}}.
\end{align*}
Note that the state evolves in a deterministic fashion, as a function of the action $e_t$ and the state $\mathbf{s}_t$.
We summarize the components of this \emph{deterministic} infinite-horizon average-reward Markov decision process in Table~\ref{tab:universal_charger_general_DP}.

Having an MDP formulation for $J_\emptyset(\rho)$ enables us to efficiently compute it using value iteration.
We show in Appendix~\ref{sec:Lagrangian_converges} that $J_\emptyset(\rho)$ in~\eqref{eq:J0_entropy} can be written as:
\begin{align*}
J_\emptyset(\rho)&=\lim_{n\to\infty}\frac{1}{n}
\max_{\substack{p(x^n),e^n:\\ (X^n,e^n)\in\mathcal{A}'_n\text{ a.s.}}}
\big(H(X^n)-\rho\sum_{t=1}^{n}e_t\big)\\*
&=\lim_{n\to\infty}\frac{1}{n}\max_{e^n}
\big(\log(\mathbf{N}_n\cdot\mathbf{1})-\rho\sum_{t=1}^{n}e_t\big).
\end{align*}
This is exactly the limit of the $n$-stage finite-horizon rewards, normalized by $n$.
As mentioned in Section~\ref{subsec:MDPs},
this suggests that the value iteration algorithm will always converge to $J_\emptyset(\rho)$ when the number of iterations goes to infinity.
Lemma~\ref{lemma:Lagrange_multipliers} guarantees that if our numeric solution is close to $J_\emptyset(\rho)$ up to some $\epsilon$, the same approximation bound holds for the resulting $C_\emptyset(\Gamma)$.

Alternatively, we can attempt to apply Theorem~\ref{thm:Bellman} to obtain an analytic solution to $J_\emptyset(\rho)$.
Solving the Bellman equation in this case consists of finding $\lambda\in\mathbb{R}$ and $h:\mathcal{S}\to\mathbb{R}$ such that:
\begin{equation}
\lambda+h(\mathbf{s})
=\max_{e\in\mathcal{E}}\big\{
g(\mathbf{s},e)+h(f(\mathbf{s},e))\big\}
\qquad\forall s\in\mathcal{S}.
\label{eq:Bellman_universal_charger}
\end{equation}
Note that the equation does not involve an expectation since the problem is deterministic.
The optimal reward will be given by $J_\emptyset(\rho)=\lambda$.
Moreover, if $e^\star(\mathbf{s})$ is the optimal charging policy (i.e. it attains the maximum in \eqref{eq:Bellman_universal_charger}), 
we can find $e_t=e^\star(\mathbf{s}_t)$ for every $t\geq1$, where $\mathbf{s}_t$ can be computed recursively from $\mathbf{s}_1$.
The capacity $C_\emptyset(\Gamma)$ can be found by applying Lemma~\ref{lemma:Lagrange_multipliers}.

\begin{figure}
\centering
\subfloat[$(0,1)$-{RLL} constrained system.]{
\def \nodedistance {3cm}
\begin{tikzpicture}[->,>=stealth',shorten >=1pt,semithick,font=\footnotesize]
\tikzstyle{every state}=[draw]
\node[state] (s0) at (0,0) {$0$};
\node[state] (s1) at (\nodedistance,0) {$1$};

\path
(s0) edge [bend left] node [above] {1} (s1)
(s1) edge [bend left] node [below] {0} (s0)
	 edge [out=15,in=-15,loop] node [right] {1} ();
\end{tikzpicture}
\label{subfig:(0,1)-RLL}
}

\subfloat[Unconstrained system.]{
\def \nodedistance {3cm}
\begin{tikzpicture}[->,>=stealth',shorten >=1pt,semithick,font=\footnotesize]
\tikzstyle{every state}=[draw]
\node[state] (s0) at (0,0) {$0$};
\node[state] (s1) at (\nodedistance,0) {$1$};

\path
(s0) edge [bend left] node [above] {1} (s1)
	 edge [out=165,in=195,loop] node [left] {0} ()
(s1) edge [bend left] node [below] {0} (s0)
	 edge [out=15,in=-15,loop] node [right] {1} ();
\end{tikzpicture}
\label{subfig:no_constraints}
}
\caption{Graph presentation of (a) $(0,1)$-{RLL} constrained system and unconstrained system.}
\label{fig:constrained_systems_graphs}
\end{figure}
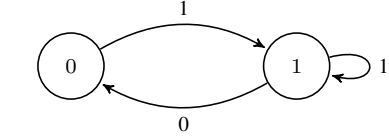
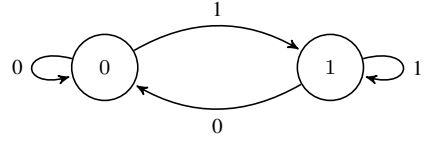

\begin{figure*}[!b]
\vspace*{4pt}
\hrulefill
\normalsize
\setcounter{MYtempeqncnt}{\value{equation}}
\setcounter{equation}{37}
\begin{align}
J&=\sup_{\pi}\liminf_{n\to\infty}\frac{1}{n}\sum_{t=1}^{n}\mathbb{E}[g(\mu_t(H_t))]\nonumber\\
&=\sup_{\substack{\{p(x_t|x^{t-1}),\ e_t(x^{t-1})\}_{t=1}^{\infty}:\nonumber\\ (\mathbf{X},\mathbf{E})\in\mathcal{A}'_\infty}}\liminf_{n\to\infty}\frac{1}{n}\sum_{t=1}^{n}\big(H(X_t|X^{t-1})-\rho\mathbb{E}[e_t(X^{t-1})]\big)\nonumber\\
&=\sup_{\substack{\{p(x_t|x^{t-1}),\ e_t(x^{t-1})\}_{t=1}^{\infty}:\nonumber\\ (\mathbf{X},\mathbf{E})\in\mathcal{A}'_\infty}}\liminf_{n\to\infty}\frac{1}{n}\big(H(X^n)-\rho\sum_{t=1}^{n}\mathbb{E}[E_t]\big)\nonumber\\
&=J_X(\rho).\label{eq:J_is_MDP}
\end{align}
\setcounter{equation}{\value{MYtempeqncnt}}
\end{figure*}

\subsubsection{Controlled Constrained Coding}
\label{subsubsec:controlled_constrained_coding}

The methodology developed above can be utilized for a new framework, which we call \emph{controlled constrained coding}.
In traditional constrained coding, we are interested in only the sequences that satisfy a certain constraint.
With applications in both data storage (magnetic recording, optical recording) and communications (bandwidth limits, DC bias elimination, synchronization),
there is vast literature on the subject (see e.g.~\cite{MarcusSiegelWolf1992,MarcusRothSiegel2001}).
For example, in the $(d,k)$-runlength-limited (RLL) constraint for binary sequences, any run of 0's between consecutive 1's must have length at least $d$ and no more than $k$.

A \emph{constrained system} can normally be presented by a labeled directed graph.
When coding over alphabet $\mathcal{X}$, each state (vertex) can have at most $|\mathcal{X}|$ outgoing edges, which are labeled according to different input symbols.
A sequence of symbols satisfying the constraint can be produced by walking along different edges of the graph.
For example, the $(0,1)$-{RLL} constrained system can be presented by the graph in Fig.~\ref{fig:constrained_systems_graphs}-\subref{subfig:(0,1)-RLL}.

A \emph{controlled} constrained system is presented by a set of graphs, all sharing the same set of states, where each graph has a specified cost.
More precisely, a controlled constrained system consists of a set of states $\mathcal{S}$, a set of costs $\mathcal{E}\subset\mathbb{R}_+$, and a collection of labeled directed graphs $\{G_e\}_{e\in\mathcal{E}}$ such that each graph's vertices are given by the set $\mathcal{S}$.
An input sequence $x^n$ in this controlled constrained system can be generated by walking along the edges of \emph{any} of the graphs $G_e$, however, a cost $e$ is paid when utilizing the graph $G_e$. Our goal is to choose a \emph{fixed} sequence of graphs $e^n$, such that the number of possible input sequences $x^n$ is maximized, but the average cost does not exceed some constraint $\Gamma\geq0$.
For a fixed graph sequence $e^n$, an admissible input sequence $x^n$ in this controlled constrained system can be generated as follows. Suppose we start from some state $s_1\in\mathcal{S}$. The first symbol $x_1$ is obtained by walking on an edge going out of state $s_1$ in graph $G_{e_1}$, and suppose the edge leads to state $s_2\in\mathcal{S}$. Then $x_2$ can be obtained by walking on an outgoing edge from $s_2$ in graph $G_{e_2}$, and so on.
Clearly, the dynamic programming approach developed previously can be used in this setting in exactly the same manner.

This framework may have applications in data storage.
Suppose we wish to design a storage device. We have at hand a set of different materials (or implementation technologies), and we can choose any of them to construct our storage device. Some materials may be better than others, and may impose different constraints on the recorded sequence. In order to minimize costs, we may choose to combine several different materials, in such a way that we balance the use of ``good and expensive'' materials and ``bad and cheap'' materials.

%

\subsection{Charger Adjacent to Transmitter}
\label{subsec:charger_adj_Tx_noiseless}


We again use Lemma~\ref{lemma:Lagrange_multipliers}, and consider the following optimization problem for $\rho\geq0$:\small
\begin{equation}
J_X(\rho)=\sup_{\substack{\{p(x_t|x^{t-1}),e_t(x^{t-1})\}_{t=1}^{\infty}:\\ (\mathbf{X},\mathbf{E})\in\mathcal{A}'_{\infty}\text{ a.s.}}}
\liminf_{n\to\infty}\frac{1}{n}
\big(H(X^n)-\rho\sum_{t=1}^{n}\mathbb{E}[E_t]\big).
\label{eq:JX}
\end{equation}\normalsize
Consider the following MDP.
Let the state be $b_t$, the battery state of the transmitter.
The state space is $\mathcal{B}$, defined in Section~\ref{subsec:universal_charger_noiseless}.
The action $u_t$ is a pair $(e_t,p(x_t))$: an energy symbol $e_t\in\mathcal{E}$ and a probability distribution over $\mathcal{X}$.
Given $b_t$ and $e_t$, eq. \eqref{eq:EH_constraint} determines if $x_t$ is admissible. Accordingly, the action space is
\[
\mathcal{U}(b)\triangleq\big\{
e,\ p(x):\ 
\phi(X)\leq \min\{b+e,\bar{B}\}\text{ a.s.}\big\}.
\]
The disturbance is the channel input $x_t$, the distribution of which is dictated by the action $p(x_t)$.
Hence, the disturbance depends only on the action at time $t$.
The stage reward is a function of the action, given by:
\[
g(u_t)=H(X_t)-\rho e_t.
\]
Finally, the state evolves as in equation~\eqref{eq:EH_battery}:
\[
b_{t+1}=f(b_t,u_t,x_t)=\min\{b_t+e_t,\bar{B}\}-\phi(x_t).
\]
Observe that since the initial state $b_t=\bar{B}$ is fixed, the history is $h_t=x^{t-1}$.
Note that a policy $\pi$ consists of mappings $\mu_t:\mathcal{X}^{t-1}\to\mathcal{U}(b_t)$, which specifies an action $u_t$ for each $x^{t-1}$.
Hence, this defines a function $e_t(x^{t-1})$ and a conditional probability distribution $p(x_t|x^{t-1})$.
\addtocounter{equation}{1}
The average expected reward for this MDP is given by~\eqref{eq:J_is_MDP} at the bottom of the page, where it is shown to be equivalent to $J_X(\rho)$.
This MDP formulation is summarized in Table~\ref{tab:charger_adjacent_Tx_general_DP}.
\begin{table*}
\renewcommand{\arraystretch}{1.5}
\caption{MDP Formulation of Capacity with Charger Adjacent to Transmitter}
\label{tab:charger_adjacent_Tx_general_DP}
\centering
\begin{tabular}{|l|l|}
\hline
state
&$b_t$, the state of the transmitter's battery\\
\hline
state space& $\mathcal{B}$, the set of possible battery states\\
\hline
action&$u_t=(e_t,p(x_t))$, the energy applied by the charger and the input distribution applied by the transmitter\\
\hline
action space&$\mathcal{U}(b)=\big\{e,\ p(x):\ \phi(X)\leq\min\{{b+e},\bar{B}\}\text{ a.s.}\big\}$\\
\hline
reward&$g(u)=g(e,p(x))=H(X)-\rho e$\\
\hline
disturbance&$x_t$, the channel input\\
\hline
disturbance distribution&$p(x_t)$, determined by $u_t$\\
\hline
state dynamics&
$b_{t+1}=f(b_t,u_t,x_t)=\min\{{b_t+e_t},\bar{B}\}-\phi(x_t)$\\
\hline
\end{tabular}
\end{table*}

As in the previous section, by Appendix~\ref{sec:Lagrangian_converges} we can equivalently write~\eqref{eq:JX} as:
\[
J_X(\rho)=\lim_{n\to\infty}\frac{1}{n}\max_{\substack{p(x^n),\{e_t(x^{t-1})\}_{t=1}^{n}:\\ (X^n,E^n)\in\mathcal{A}'_n\text{ a.s.}}}
\big(H(X^n)-\rho\sum_{t=1}^{n}\mathbb{E}[E_t]\big),
\]
which implies the value iteration algorithm will converge to the desired limit, with similar approximation guarantees as in the previous section.
Additionally, by Theorem \ref{thm:Bellman},
if there exist a scalar $\lambda\in\mathbb{R}$ and a vector $h\in\mathbb{R}^{|\mathcal{B}|}$ that satisfy the following Bellman equation:
\begin{equation}
\lambda+h(b)
=\max_{u\in\mathcal{U}(b)}\Big\{
g(u)+\mathbb{E}\big[h(f(b,u,X))\big]\Big\}
\qquad\forall b\in\mathcal{B},
\label{eq:DP_Bellman_CX_general}
\end{equation}
then $J_X(\rho)=\lambda$.
Furthermore, if $u^\star(b)=(e^\star(b),p^\star(x|b))$ attains the maximum in~\eqref{eq:DP_Bellman_CX_general}, then the optimal policy is stationary and is given by $u^\star(b)$, i.e. $e^\star_t(x^{t-1})=e^\star(b_t)$ and $p^\star(x_t|x^{t-1})=p^\star(x_t|b_t)$.

\subsection{Fully Cognitive Charger}
\label{subsec:charger_knows_M_noiseless}

The final case we study is when the charger observes the message.
Recall the capacity expression \eqref{eq:CM_relaxed}.
Before applying Lemma~\ref{lemma:Lagrange_multipliers}, we show that we can restrict the set of input distributions $p(x^n,e^n)=\prod_{t=1}^{n}p(x_t,e_t|x^{t-1},e^{t-1})$ over which to maximize.
The battery state $b_t$ is a deterministic function of $(x^{t-1},e^{t-1})$, hence
\[
p(x_t,e_t|x^{t-1},e^{t-1})=p(x_t,e_t|x^{t-1},e^{t-1},b_t).
\]
We claim that it is in fact enough to choose only marginal distributions of the form $\{p(x_t,e_t|x^{t-1},b_t)\}_{t=1}^{n}$.
In other words, we can restrict the optimization domain to consist of only input distributions that depend on the past energy arrivals $e^{t-1}$ only through the battery state $b_t$.
More precisely, we state the following lemma:
\begin{lemma}
The capacity of the noiseless channel with a fully cognitive charger can be written as:
\begin{equation}
C_M(\Gamma)=\lim_{n\to\infty}\frac{1}{n}
\max_{\substack{\{p(x_t,e_t|x^{t-1},b_t)\}_{t=1}^{n}:\\
	(X^n,E^n)\in\mathcal{A}'_n\text{ a.s.}\\
	\sum_{t=1}^{n}\mathbb{E}[E_t]\leq n\Gamma}}
	H(X^n),
\end{equation}
where it is understood that the input distribution is given by
\[
p(x_t,e_t|x^{t-1},e^{t-1})=p(x_t,e_t|x^{t-1},b_t),
\quad t=1,\ldots,n,
\]
where $b_t$ is a deterministic function of $x^{t-1},e^{t-1}$, given by the battery evolution equation~\eqref{eq:EH_battery}, i.e. $b_{t+1}=\min\{b_t+e_t,\bar{B}\}-\phi(x_t)$, $b_1=\bar{B}$.
\end{lemma}
\begin{proof}
We will show that for every $n\geq1$:
\begin{equation}
\max_{\substack{p(x^n,e^n):\\ (X^n,E^n)\in\mathcal{A}'_n\text{ a.s.}\\ \sum_{t=1}^{n}\mathbb{E}[E_t]\leq n\Gamma}}
H(X^n)
=\max_{\substack{\{p(x_t,e_t|x^{t-1},b_t)\}_{t=1}^{n}:\\
	(X^n,E^n)\in\mathcal{A}'_n\text{ a.s.}\\
	\sum_{t=1}^{n}\mathbb{E}[E_t]\leq n\Gamma}}
H(X^n).
\label{eq:CM_optimal_input}
\end{equation}
We will show inequalities in both directions.
Clearly, LHS $\geq$ RHS,
since the maximization domain in the RHS is a subset of the one in the LHS.
Hence, we only need to show the following:
\begin{equation}
\max_{\substack{p(x^n,e^n):\\ (X^n,E^n)\in\mathcal{A}'_n\text{ a.s.}\\ \sum_{t=1}^{n}\mathbb{E}[E_t]\leq n\Gamma}}
H(X^n)
\leq
\max_{\substack{\{p(x_t,e_t|x^{t-1},b_t)\}_{t=1}^{n}:\\
	(X^n,E^n)\in\mathcal{A}'_n\text{ a.s.}\\
	\sum_{t=1}^{n}\mathbb{E}[E_t]\leq n\Gamma}}
H(X^n).
\label{eq:enough_to_max_marginals}
\end{equation}
Let $p^\star(x^n,e^n)$ be the maximizer of the LHS, and let $(X^{\star n},E^{\star n})\sim p^\star(x^n,e^n)$.
Let $\{p^\star(x_t,e_t|x^{t-1},b_t)\}_{t=1}^{n}$ be the corresponding set of marginals. Specifically, they are obtained from $p^\star(x^n,e^n)$ as follows:
\[
p^\star(x_t,e_t|x^{t-1},b_t)
=\frac{\sum_{e^{t-1}}p^\star(x^t,e^t)1(b_t|x^{t-1},e^{t-1})}
	{\sum_{e^{t-1}}p^\star(x^{t-1},e^{t-1})1(b_t|x^{t-1},e^{t-1})},
\]
where $1(b_t|x^{t-1},e^{t-1})$ is an indicator function, which equals 1 if $b_t$ is given by the battery evolution relation as before, and 0 otherwise.

Let $(\tilde{X}^n,\tilde{E}^n)\sim\tilde{p}(x^n,e^n)$, where
\[
\tilde{p}(x_t,e_t|x^{t-1},e^{t-1})=p^\star(x_t,e_t|x^{t-1},b_t).
\]
First, observe that $(\tilde{X}^n,\tilde{E}^n)\in\mathcal{A}'_n$ a.s.:
This follows because $(X^{\star n},E^{\star n})\in\mathcal{A}'_n$ a.s., which implies 
$p^\star(x_t,e_t|x^{t-1},e^{t-1})=0$
for all $x^t,e^t$ such that $\phi(x_t)>\min\{b_t+e_t,\bar{B}\}$.
Then $\tilde{p}(x_t,e_t|x^{t-1},e^{t-1})=p^\star(x_t,e_t|x^{t-1},b_t)=0$ for all such $x^t,e^t$.

Next, we will show by induction that 
\begin{equation}
\{\tilde{p}(x^t,e_t,b_t)\}_{t=1}^{n}=\{p^\star(x^t,e_t,b_t)\}_{t=1}^{n}.
\label{eq:ptilde_eq_pstar}
\end{equation}
This is true for $t=1$:
\begin{align*}
\tilde{p}(x_1,e_1,b_1)
&=\tilde{p}(b_1)\tilde{p}(x_1,e_1|b_1)\\
&=1\{b_1=\bar{B}\}\cdot p^\star(x_1,e_1|b_1)\\
&=p^\star(x_1,e_1,b_1).
\end{align*}
Next, for $t>1$, assume 
\[\tilde{p}(x^{t-1},e_{t-1},b_{t-1})=p^\star(x^{t-1},e_{t-1},b_{t-1}).\]
Since $b_t$ is a deterministic function of $x_{t-1},e_{t-1},b_{t-1}$, this implies
\[\tilde{p}(x^{t-1},e_{t-1},b_{t-1},b_{t})=p^\star(x^{t-1},e_{t-1},b_{t-1},b_{t}),\]
which in turn implies $\tilde{p}(x^{t-1},b_{t})=p^\star(x^{t-1},b_{t})$.
Hence,
\begin{align*}
\tilde{p}(x^{t},e_t,b_t)
&=\tilde{p}(x^{t-1},b_t)\tilde{p}(x_t,e_t|x^{t-1},b_t)\\
&\overset{\text{($*$)}}{=}p^\star(x^{t-1},b_t)p^\star(x_t,e_t|x^{t-1},b_t)\\
&=p^\star(x^t,e_t,b_t),
\end{align*}
where ($*$) is by definition of $\tilde{p}$.

Finally, from~\eqref{eq:ptilde_eq_pstar}, it is clear that
\[\sum_{t=1}^{n}\mathbb{E}[\tilde{E}_t]=\sum_{t=1}^{n}\mathbb{E}[E^{\star}_n]\leq n\Gamma\]
and $H(\tilde{X}^n)=H(X^{\star n})$, which yields~\eqref{eq:enough_to_max_marginals}.
\end{proof}

%

In light of Lemma~\ref{lemma:Lagrange_multipliers}, we consider the following optimization problem, for $\rho\geq0$:
\begin{align}
J_M(\rho)&=\hspace{-1em}\sup_{\substack{\{p(x_t,e_t|x^{t-1},b_t)\}_{t=1}^{\infty}:\\ (\mathbf{X},\mathbf{E})\in\mathcal{A}'_\infty\text{ a.s.}}}\liminf_{n\to\infty}\frac{1}{n} \big(H(X^n)-\rho\sum_{t=1}^{n}\mathbb{E}[E_t]\big).
\label{eq:JM}
\end{align}

We will now formulate this problem as an MDP.
We will start from the disturbance $w_t$, which we take to be the channel input $x_t$.
The history is then $h_t=w^{t-1}=x^{t-1}$. Hence, the state, reward, and action, will be implicitly a function of $x^{t-1}$.
The state of the system $\mathbf{s}_t$ will be a vector in the probability simplex $\mathbb{R}^{|\mathcal{B}|}$ similarly to Section~\ref{subsec:universal_charger_noiseless}, where $\mathcal{B}$ is the set of battery states.
Specifically, the state is given by the probability distribution $s_t(b_t)=p(b_t|x^{t-1})$, where $b_t$ is the battery state of the transmitter.
The initial state is $s_1(b)=1\{b=\bar{B}\}$.

The action $\mathbf{u}_t$ is a stochastic matrix representing a conditional probability distribution $p(x,e|b)$, that is, $\mathbf{u}_t\in\mathbb{R}^{|\mathcal{B}|\times|\mathcal{X}|\cdot|\mathcal{E}|}$.
The action space $\mathcal{U}$ is the set of all such conditional probability distributions such that the energy constraint~\eqref{eq:EH_battery} holds, namely
$p(x,e|b)=0$ for all $(x,e,b)$ that satisfy $\phi(x)>\min\{b+e,\bar{B}\}$.
The action is chosen by a mapping $\mu_t:\mathcal{X}^{t-1}\to\mathcal{U}$, which defines a conditional probability distribution $p(x_t,e_t|b_t,x^{t-1})$.
This implies that a policy $\pi$ is exactly a set of probability distributions $\{p(x_t,e_t|b_t,x^{t-1})\}_{t=1}^{\infty}$ that satisfy 
$p(x_t,e_t|b_t,x^{t-1})=0$ if $\phi(x_t)>\min\{b_t+e_t,\bar{B}\}$,
which is equivalent to satisfying the constraint
$(\mathbf{X},\mathbf{E})\in\mathcal{A}'_\infty$ a.s. in~\eqref{eq:JM}.

Before we proceed to define the reward, we will show that the disturbance distribution depends only on the current state and action, and that the state evolves as a function of the current state, action, and disturbance.
Starting with the disturbance, we have:
\begin{align}
p(w_t|w^{t-1})&=p(x_t|x^{t-1})\nonumber\\
&=\sum_{e_t,b_t}p(x_t,e_t,b_t|x^{t-1})\nonumber\\
&=\sum_{e_t,b_t}p(b_t|x^{t-1})p(x_t,e_t|x^{t-1},b_t)\nonumber\\
&=\sum_{e_t,b_t}s_t(b_t)u_t(x_t,e_t|b_t)\nonumber\\
&=\sum_{e,b}s_t(b)u_t(w_t,e|b).
\label{eq:CM_MDP_px}
\end{align}
The state at time $t+1$ is given by a state evolution function $\mathbf{s}_{t+1}=f(\mathbf{s}_t,\mathbf{u}_t,w_t)$, as shown below:
\begin{align}
s_{t+1}(b_{t+1})
&=p(b_{t+1}|x^t)\nonumber\\*
&=\frac{p(b_{t+1},x_t|x^{t-1})}{p(x_t|x^{t-1})}\nonumber\\
&=\frac{\sum_{e_t,b_t}p(x_t,e_t,b_t,b_{t+1}|x^{t-1})}
	{\sum_{e_t,b_t}p(x_t,e_t,b_t|x^{t-1})}\nonumber\\
&=\frac{\sum_{e_t,b_t}p(b_t|x^{t-1})p(x_t,e_t|b_t,x^{t-1})1(b_{t+1}|x_t,e_t,b_t)}
{\sum_{e_t,b_t}p(b_t|x^{t-1})p(x_t,e_t|b_t,x^{t-1})}\nonumber\\
&=\frac{\sum_{e_t,b_t}s_t(b_t)u_t(x_t,e_t|b_t)1(b_{t+1}|x_t,e_t,b_t)}
{\sum_{e_t,b_t}s_t(b_t)u_t(x_t,e_t)}\nonumber\\*
&=\frac{\sum_{e_t,b_t}s_t(b_t)u_t(w_t,e_t|b_t)1(b_{t+1}|w_t,e_t,b_t)}
{\sum_{e_t,b_t}s_t(b_t)u_t(w_t,e_t)},
\label{eq:CM_DP_state_dynamics}
\end{align}
where $1(b_{t+1}|x_t,e_t,b_t)$ indicates that $b_{t+1}$ is a deterministic function of $b_t,x_t,e_t$, given by the battery evolution relation~\eqref{eq:EH_battery}:
\[
1(b_{t+1}|x_t,e_t,b_t)=
1\big\{b_{t+1}=\min\{b_t+e_t,\bar{B}\}-\phi(x_t)\big\}.
\]

Finally, we define the reward at time $t$ as
\[
H(X_t|X^{t-1}=x^{t-1})-\rho\mathbb{E}[E_t|X^{t-1}=x^{t-1}].
\]
Observe that it depends only on the marginal $p(x_t,e_t|x^{t-1})$, which is given by:
\begin{align*}
p(x_t,e_t|x^{t-1})
&=\sum_{b_t}p(b_t|x^{t-1})p(x_t,e_t|b_t,x^{t-1})\\
&=\sum_{b}s_t(b)u_t(x_t,e_t|b)
\end{align*}
Hence we can write the reward as a function of the current state and action:
\begin{align*}
&g(\mathbf{s}_t,\mathbf{u}_t)\\*
&\quad =\sum_{x,e,b}s_t(b)u_t(x,e|b)\big(-\log\sum_{e',b'}s_t(b')u_t(x,e'|b')-\rho e\big).
\end{align*}
We conclude that the formulation defined above indeed satisfies the conditions for an MDP. We summarize this formulation in Table~\ref{tab:charger_knows_M_general_DP}.
\begin{table*}
\renewcommand{\arraystretch}{1.5}
\caption{MDP Formulation of Capacity with Fully Cognitive Charger}
\label{tab:charger_knows_M_general_DP}
\centering
\begin{tabular}{|l|p{11cm}|}
\hline
state
&$\mathbf{s}_t=[s_t(b): b\in\mathcal{B}]$, a probability distribution of the battery state, represents $p(b_t|x^{t-1})$\\
\hline
state space& $\mathcal{S}$, the probability simplex in $\mathbb{R}^{|\mathcal{B}|}$\\
\hline
action&$\mathbf{u}_t=[u_t(x,e|b):(x,e,b)\in\mathcal{X}\times\mathcal{E}\times\mathcal{B}]$, a conditional probability of the channel input and energy given the battery state, represents $p(x_t,e_t|b_t,x^{t-1})$\\
\hline
action space&$\mathcal{U}$, the set of all stochastic matrices $p(x,e|b)$ s.t. $\phi(X)\leq\min\{B+E,\bar{B}\}$ w.p. 1
\\
\hline
reward&$g(\mathbf{s},\mathbf{u})=H(X_t|x^{t-1})-\rho\mathbb{E}[E_t|x^{t-1}]$\\
\hline
disturbance&$w_t$, represents the channel input $x_t$\\
\hline
disturbance distribution&determined by $\mathbf{s}_t$ and $\mathbf{u}_t$ in \eqref{eq:CM_MDP_px}\\
\hline
state dynamics&$\mathbf{s}_{t+1}=f(\mathbf{s}_t,\mathbf{u}_t,w_t)$ given by \eqref{eq:CM_DP_state_dynamics}\\
\hline
\end{tabular}
\end{table*}

What remains to be verified is that this MDP indeed solves our original problem in~\eqref{eq:JM}.
This follows by observing that the average expected reward per stage is given by:
\begin{align*}
J&=\sup_{\pi}\liminf_{n\to\infty}\frac{1}{n}\sum_{t=1}^{n}\mathbb{E}[g(\mathbf{s}_t,\mu_t(H_t))]\\
&\overset{\text{(i)}}{=}\sup_{\substack{\{p(x_t,e_t|x^{t-1},b_t)\}_{t=1}^{\infty}:\\ (\mathbf{X},\mathbf{E})\in\mathcal{A}'_\infty\text{ a.s.}}}\liminf_{n\to\infty}\frac{1}{n}
\sum_{t=1}^{n}\big(H(X_t|X^{t-1})-\rho\mathbb{E}[E_t]\big)\\
&=\sup_{\substack{\{p(x_t,e_t|x^{t-1},b_t)\}_{t=1}^{\infty}:\\ (\mathbf{X},\mathbf{E})\in\mathcal{A}'_\infty\text{ a.s.}}}\liminf_{n\to\infty}\frac{1}{n}
\big(H(X^n)-\rho\sum_{t=1}^{n}\mathbb{E}[E_t]\big)\\
&=J_M(\rho),
\end{align*}
where (i) is because the set of policies $\pi$ is exactly the set of marginals that satisfy the constraints as argued earlier, and because the expectation at time $t$ is over the disturbances distribution $p(w^{t-1})$, which is equivalent to $p(x^{t-1})$.

\begin{remark}
The channel discussed here falls into the category of finite-state channels with feedback discussed in \cite{yang2005feedback,permuter2008capacity}, where we identify the state as $B_t$, the input as $(X_t,E_t)$, and the output as $X_t$.
The treatment here followed the same lines as in these works.
\end{remark}

As in the previous sections,
using the result in Appendix~\ref{sec:Lagrangian_converges}, $J_M(\rho)$ is given also by:
\[
J_M(\rho)=\lim_{n\to\infty}\frac{1}{n}\max_{\substack{\{p(x_t,e_t|x^{t-1},b_t)\}_{t=1}^{n}:\\ (X^n,E^n)\in\mathcal{A}'_n\text{ a.s.}}}
\big(H(X^n)-\rho\sum_{t=1}^{n}\mathbb{E}[E_t]\big),
\]
which implies convergence of the value iteration algorithm.
The Bellman equation (Theorem \ref{thm:Bellman}) is given by:
\begin{equation}
\lambda+h(\mathbf{s})=\max_{u\in\mathcal{U}}\left\{
g(\mathbf{s},\mathbf{u})+\mathbb{E}[h(f(\mathbf{s},\mathbf{u},X))]\right\}
\qquad \forall s\in\mathcal{S},
\label{eq:CM_DP_Bellman}
\end{equation}
where $\lambda\in\mathbb{R}$ and $h:\mathcal{S}\to\mathbb{R}$.


\section{Example}
\label{sec:example}

In this section we will consider a concrete example and compute the capacity for the different cases discussed in Theorem~\ref{thm:capacity}.
The input alphabet is $\mathcal{X}=\{0,1,2\}$, with energy cost function $\phi(x)=x$.
The channel is assumed to be noiseless: $Y_t=X_t$.
The battery capacity is $\bar{B}=2$, and the energy alphabet is $\mathcal{E}=\{0,2\}$ (we do not allow the charger to choose $e=1$).
This models a system in which the charger is not accurate, and can only release large bursts of energy that completely charge the battery.


In Appendix~\ref{sec:example_derivations} we derive closed-form expressions for $C_\emptyset(\Gamma)$, $C_X(\Gamma)$, and $C_M(\Gamma)$.
Specifically, we have the following propositions.
\begin{proposition}\label{prop:example_C0}
The capacity of the channel defined above for the case of a generic charger is given as follows:
For any $0<\Gamma\leq 2$ which satisfies $\Gamma=\tfrac{2}{\ell}$ for some integer $\ell\geq1$, capacity is given by
\begin{align}
C_\emptyset(\Gamma)
&=\frac{\Gamma}{2}\log\frac{(\Gamma+1)(\Gamma+2)}{\Gamma^2}.
\label{eq:example_C0_integer}
\end{align}
For all other values of $0<\Gamma\leq2$, let $\Gamma_1=\frac{2}{\lceil2/\Gamma\rceil}$ and $\Gamma_2=\frac{2}{\lfloor2/\Gamma\rfloor}$, and let $\alpha=\frac{\Gamma_2-\Gamma}{\Gamma_2-\Gamma_1}$. Then
\begin{equation}
C_\emptyset(\Gamma)=\alpha C_\emptyset(\Gamma_1)+(1-\alpha)C_\emptyset(\Gamma_2).
\label{eq:example_C0_Gamma}
\end{equation}
Moreover, $C_\emptyset(0)=0$ and $C_\emptyset(\Gamma)=\log3$ for $\Gamma>2$.
\end{proposition}
When $\Gamma=\tfrac{2}{\ell}$ for some integer $\ell\geq1$, the optimal charging sequence is simply to charge $e=2$ once every $\ell$ time slots, i.e.
\[
e_t=\begin{cases}
2&,t=1\mod \ell\\
0&,\text{otherwise}
\end{cases}
\]
The transmitter generates its codewords by considering all possible sequences in a frame of length $\ell$ (the time between consecutive energy arrivals) that have cost smaller than or equal to $2$. The possible sequences are those containing exactly one 2 (there are $\ell$ such sequences); containing exactly one 1 ($\ell$ such sequences); containing exactly two 1's ($\binom{\ell}{2}$ such sequences); and the all-zero sequence. Therefore, there is a total of $\frac{\ell^2+3\ell+2}{2}$ sequences for a frame of length $\ell$, giving a rate $R=\frac{1}{\ell}\log\frac{(\ell+1)(\ell+2)}{2}$. For other values of $\Gamma$, time-sharing between the strategies corresponding to the two closest integer values for $\ell$, i.e. $\ell_1=\lceil2/\Gamma\rceil$ and $\ell_2=\lfloor2/\Gamma\rfloor$, is optimal, giving the expression in~\eqref{eq:example_C0_Gamma}.

\begin{proposition}\label{prop:example_CX}
The capacity of the channel defined above when the charger is adjacent to the transmitter is given by the following expression:
\begin{equation}
{C}_X(\Gamma)=\begin{cases}
\lefteqn{\big(1+\tfrac{\Gamma}{2}\big)\log\tfrac{2+\Gamma}{2\Gamma}
	-\big(1-\tfrac{\Gamma}{2}\big)\log\tfrac{2-\Gamma}{2\Gamma}}\\
	&,0\leq\Gamma<2/3\\
\frac{1}{2}\Gamma+1\hspace{70pt}
	&,2/3\leq\Gamma<1\\
\frac{\Gamma}{2}+H_2\left(\frac{\Gamma}{2}\right)
	&,1\leq\Gamma<4/3\\
\log 3
	&,4/3\leq\Gamma
\end{cases}
\label{eq:example_CX_Gamma}
\end{equation}
\end{proposition}
The optimal charging strategy $e^\star(b)$, as found in Appendix~\ref{sec:example_derivations}, is $e^\star(0)=2$, $e^\star(2)=0$, and
\[
e^\star(1)=\begin{cases}
0&,0<\Gamma\leq2/3\\
2&,1\leq\Gamma\leq 4/3
\end{cases}
\]
i.e. when the battery is empty, one must always charge; when the battery is full, one must never charge; and when the battery level is $b=1$, charging depends on $\Gamma$: we never charge when $\Gamma\leq2/3$ and always charge when $1\leq\Gamma$ in this case.
Note that it is indeed surprising that the charging actions depend only on $b$ and not on the transmitted symbols $x^{t-1}$, and exhibit a dichotomy as a function of $\Gamma$. Note that, as discussed next, even though the transmitter's strategy is the same for a range of $\Gamma$'s, the average cost of the charging sequences will be different for different values of $\Gamma$ (and equal to $\Gamma$), since the emitted charging sequence also depends on the transmitted sequence by the transmitter. Note that $e^\star(1)$ is not specified for $2/3<\Gamma<1$; in that case, time-sharing between the charging strategies corresponding to $\Gamma=2/3$ and $\Gamma=1$ is optimal.

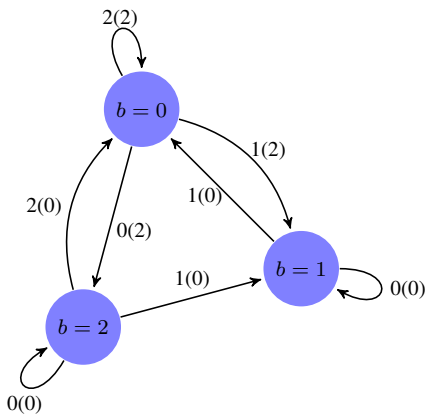
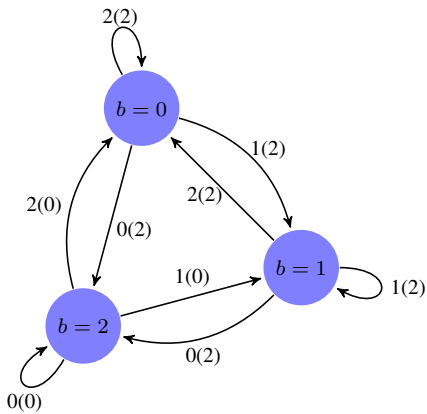
\begin{figure}
\centering
\subfloat[$0<\Gamma\leq2/3$]{
\centering
\def \nodedistance {3cm}
\begin{tikzpicture}[->,>=stealth',shorten >=1pt,semithick,font=\footnotesize]
\tikzstyle{every state}=[draw=none,fill=blue!50]
\node[state] (b2) {$b=2$};
\node[state] (b1) at (15:\nodedistance) {$b=1$};
\node[state] (b0) at (75:\nodedistance) {$b=0$};

\path
(b0) edge [in=90,out=120,loop] node [above=-3] {2(2)} ()
	 edge [bend left] node [above right=-4] {1(2)} (b1)
	 edge node [below right=-2] {0(2)} (b2)
(b1) edge [in=-30,out=0,loop] node [right] {0(0)} ()
     edge node [below left=-4] {1(0)} (b0)
(b2) edge [in=-150,out=-120,loop] node [below] {0(0)} (b2)
	 edge [bend left] node [left] {2(0)} (b0)
	 edge node [above] {1(0)} (b1);
\end{tikzpicture}
\label{subfig:small_Gamma}
}\qquad
\subfloat[$1\leq\Gamma\leq4/3$]{
\centering
\def \nodedistance {3cm}
\begin{tikzpicture}[->,>=stealth',shorten >=1pt,semithick,font=\footnotesize]
\tikzstyle{every state}=[draw=none,fill=blue!50]
\node[state] (b2) {$b=2$};
\node[state] (b1) at (15:\nodedistance) {$b=1$};
\node[state] (b0) at (75:\nodedistance) {$b=0$};

\path
(b0) edge [in=90,out=120,loop] node [above=-3] {2(2)} ()
	 edge [bend left] node [above right=-4] {1(2)} (b1)
	 edge node [below right=-2] {0(2)} (b2)
(b1) edge [in=-30,out=0,loop] node [right] {1(2)} ()
     edge node [below left=-4] {2(2)} (b0)
     edge [bend left] node [below] {0(2)} (b2)
(b2) edge [in=-150,out=-120,loop] node [below] {0(0)} (b2)
	 edge [bend left] node [left] {2(0)} (b0)
	 edge node [above] {1(0)} (b1);
\end{tikzpicture}
\label{subfig:large_Gamma}
}
\caption{Constrained systems for the transmitter's input codeword $x^n$, for the noiseless channel with $\mathcal{X}=\{0,1,2\}$ and $\mathcal{E}=\{0,2\}$ when the charger is adjacent to the transmitter.
The nodes represent the battery state $b_t$, and the edges are labeled $x_t (e_t)$, such that transmitting $x_t$ incurs a cost of $e_t$.}
\label{fig:CX_example_constrained_systems}
\end{figure}

Once the charging strategy is fixed as a function of the battery state (or equivalently the input), it is up to the transmitter to choose its codewords so that they satisfy the battery constraints and at the same time they induce no more than cost $\Gamma$ at the charger on average.
This defines a constrained system with costs for the transmitter, plotted in Fig.~\ref{fig:CX_example_constrained_systems}.
For each edge on the graph we associate a symbol $x_t$ and a cost $e_t$ which is incurred when this symbol is transmitted via the charging strategy described above, which is fixed at the charger. Note that since the optimal charging strategy exhibits  a dichotomy as a function of $\Gamma$, we have two constrained systems with costs corresponding to each strategy. An encoder for each of these systems can be designed using the state-splitting algorithm for channels with cost constraints suggested in~\cite{KhayrallahNeuhoff1996}.\footnote{In fact, these are simpler systems than the ones studied in~\cite{KhayrallahNeuhoff1996}, since the cost depends only on the current state and not on the transmitted symbol. Nevertheless, we can apply their results here.}

\begin{proposition}\label{prop:example_CM}
The capacity of the channel defined above with a fully cognitive charger 
for $0<\Gamma\leq10/9$ is given as follows:
\begin{equation}
C_M(\Gamma)=
\log(\zeta^2+2\zeta-3)+\frac{\Gamma-2}{2}\log(\zeta^2-5)-\Gamma,
\label{eq:example_CM_Gamma}
\end{equation}
where $\zeta$ is the unique real root of the cubic equation
\begin{equation}
\Gamma\zeta^3+2(\Gamma-1)\zeta^2-(3\Gamma+4)\zeta-10=0.
\label{eq:zeta_cubic_equation}
\end{equation}
Additionally, $C_M(0)=0$ and $C_M(\Gamma)=\log3$ for $\Gamma>10/9$.
\end{proposition}
In this case, the transmitter and the charger code jointly according to the message. The set of possible codewords consists of all sequences that can be transmitted under some energy sequence that satisfy the average cost constraint. We show in Appendix~\ref{sec:example_derivations} that this scenario reduces to a single constrained system with a cost constraint, plotted in Fig.~\ref{fig:CM_example_constrained_system}.
As before, we can construct an encoder following the state-splitting algorithm in~\cite{KhayrallahNeuhoff1996}.

\begin{figure}
\centering
\def \nodedistance {3cm}
\begin{tikzpicture}[->,>=stealth',shorten >=1pt,semithick,font=\footnotesize]
\tikzstyle{every state}=[draw=none,fill=blue!50]
\node[state] (b0) at (0,0) {$b=0$};
\node[state] (b1) at (\nodedistance,0) {$b=1$};

\path
(b0) edge [in=135,out=165,loop] node [left] {2(2)} ()
	 edge [in=195,out=225,loop] node [left] {0(0)} ()
	 edge [bend left] node [above] {1(2)} (b1)
(b1) edge [in=-15,out=15,loop] node [right] {0(0)} ()
     edge node [below=-3] {2(2)} (b0)
     edge [bend left] node [below] {1(0)} (b0);
\end{tikzpicture}
\caption{Constrained systems for the transmitter-charger codeword pair $(x^n,e^n)$, for the noiseless channel with $\mathcal{X}=\{0,1,2\}$ and $\mathcal{E}=\{0,2\}$ when the charger is fully cognitive.
The nodes represent the battery state $b_t$, and the edges are labeled $x_t (e_t)$, such that transmitting $x_t$ incurs a cost of $e_t$.}
\label{fig:CM_example_constrained_system}
\end{figure}
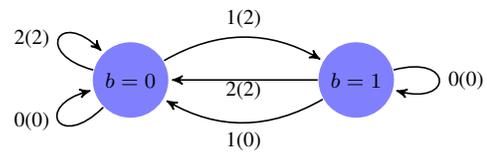

The above capacity expressions are plotted in Fig.~\ref{fig:ternary_channel_binary_energy_capacities} along with the average power upper bound $C_{\mathrm{ub}}(\Gamma)$.
Observe that, at least for some values of $\Gamma$, the capacities are distinct: 
\[C_\emptyset(\Gamma)<C_X(\Gamma)<C_M(\Gamma)<\Cub(\Gamma).\]
Consider the following simple example, which can provide some insight as to why these capacities can be different when the channel does not have a precision charger.
Let $n=3$, and suppose the codebook consists of only two codewords, $\mathbf{x}_1=(1,1,2)$ and $\mathbf{x}_2=(1,2,0)$.
When the charger is fully cognitive ($C_M$), it can set the following charging sequences: $\mathbf{e}_1=(0,0,2)$ and $\mathbf{e}_2=(0,2,0)$, for $\mathbf{x}_1$ and $\mathbf{x}_2$ respectively (recall $B_1=2$).
The total energy cost for this code is $2$.
Now, in the case when the charger is adjacent to the transmitter and observes only the previous input symbols ($C_X$), the charger must have a set of functions $e_t(x^{t-1})$, $t=1,2,3$, such that the energy constraints will be satisfied for both codewords.
Obviously, $e_1=0$ since $B_1=2$. For $t=2$, the charger will observe $x_1=1$ if either codeword was transmitted.
It knows that $B_2=1$, however it can only set $e_2=2$ to be able to support both codewords, even though energy will be clearly wasted.
If $\mathbf{x}_1$ was transmitted this will require $e_3=2$, which ultimately results in total energy cost $4$.

\section{Conclusion}
\label{sec:conclusion}

A general model for remotely powered communication was formulated, and $n$-letter capacity expressions were derived for several cases of interest, based on the availability of side information at the charger. For channels with a precision charger, we showed that side information enables achieving the single-letter capacity under a simple average transmit energy constraint. For noiseless channels, we formulated the capacity in each case as an MDP. This enables the use of the value iteration algorithm for efficient computation of the capacity. Moreover, we showed that the Bellman equation can be explicitly solved for a specific example yielding an analytic expression for the capacities in each case.


\begin{figure}
\centering
\includegraphics[width=\graphicswidth]{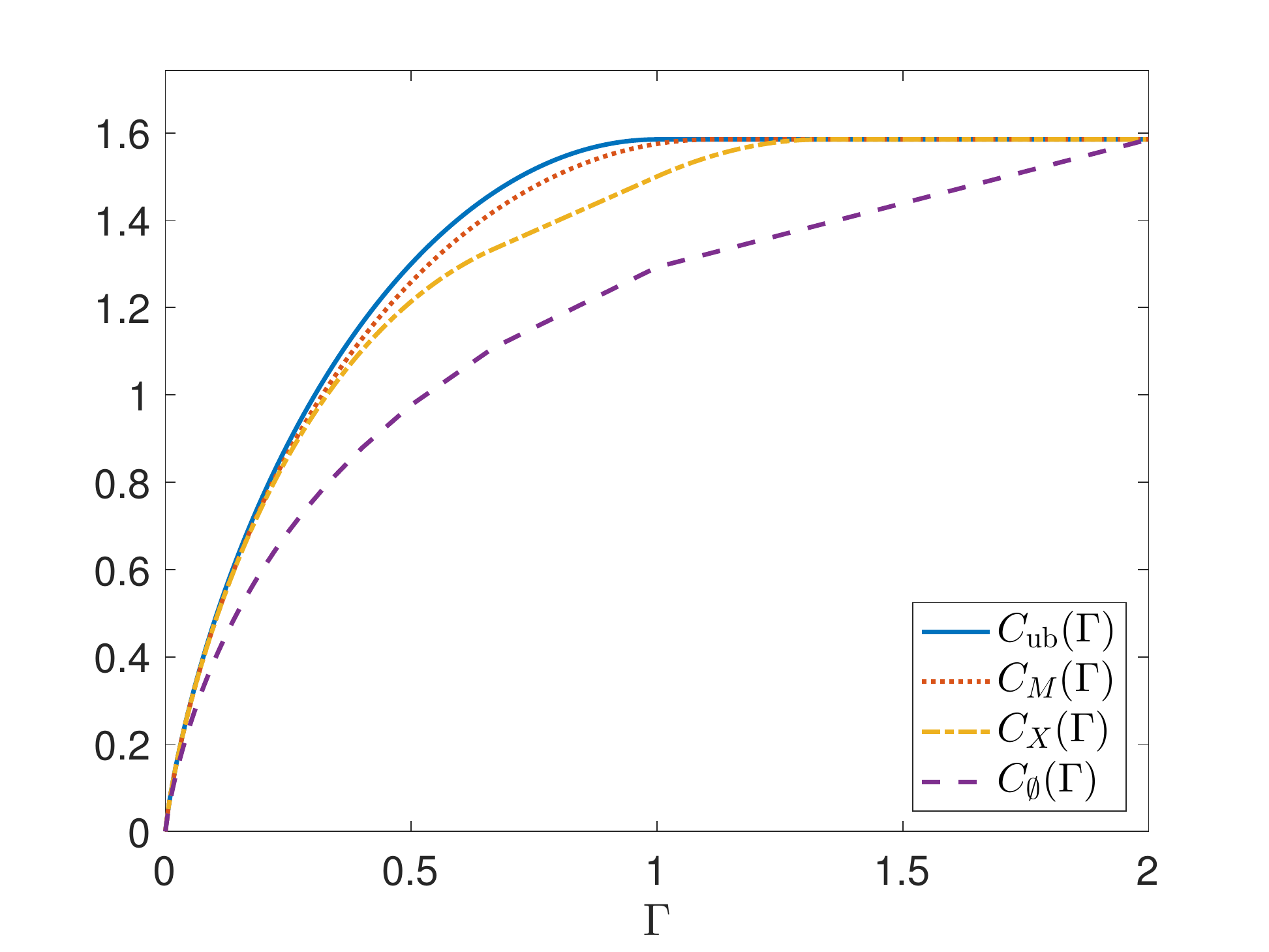}
\caption{Capacity of the noiseless ternary channel $\mathcal{X}=\{0,1,2\}$ with binary energy alphabet $\mathcal{E}=\{0,2\}$ for various levels of charger cognition.}
\label{fig:ternary_channel_binary_energy_capacities}
\end{figure}

An interesting question that remains open is to find single-letter capacity expressions or MDP formulations when the channel is noisy. Note that while our results for the precision charger hold equally well for noisy and noiseless channels, the development of the MDP formulations was restricted to the noiseless case. A particularly interesting scenario, which remains largely open, is when the receiver charges the transmitter, $C_Y(\Gamma)$. In this case, the charger can convey not only energy but also feedback information with its actions to the transmitter. A first order question towards addressing this case is whether feedback information can increase the capacity of a discrete memoryless channel with battery constraints as given in \eqref{eq:EH_constraint} and \eqref{eq:EH_battery}. Interestingly, this question relates to an old claim by Shannon from his 1956 paper \cite{Shannon1956}, which says that feedback would be useless for increasing the capacity of such channels.  We show in~\cite{ShavivOzgurPermuter2015} that feedback can indeed increase the capacity of these channels, providing a counter-example to Shannon's claim.  
Additional future directions include incorporating processing cost or battery leakage in the model.

\appendices

\section{A General Capacity Theorem}
\label{sec:strategies_capacity_and_special_cases}

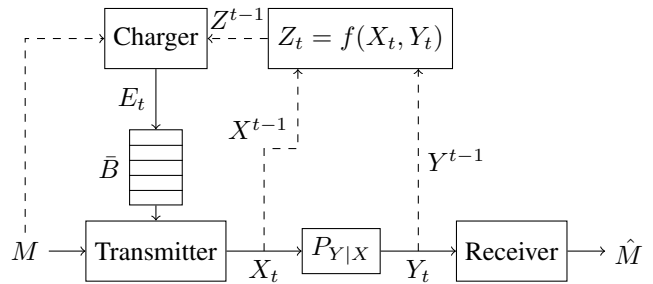
\begin{figure}
\centering
\begin{tikzpicture}
	\def \arlen {0.8cm};
	\def \blockwidth {1cm}
	\def \blockheight {0.8cm}

	\node[draw,rectangle,minimum width=\blockwidth,
		minimum height=\blockheight] 
		(Tx) at (0,0) {Transmitter};
		
	\node[draw,rectangle,right]
		(Channel) at ($(Tx.east)+(1,0)$) {$P_{Y|X}$};
	
	\node[draw,rectangle,right,minimum width=\blockwidth,
		minimum height=\blockheight]
		(Rx) at ($(Channel.east)+(1,0)$) {Receiver};
	
	\node[draw,rectangle,above,minimum width=0.7cm,minimum height=1cm] (Battery) at ($(Tx.north)+(0,0.2)$) {};
	\foreach \y in {0.2,0.4,0.6,0.8}
	{
		\draw[-] ($(Battery.south west)+(0,\y)$) -- 
			($(Battery.south east)+(0,\y)$);
	}
	\node[left] at ($(Battery.west)$) {$\bar{B}$};
	
	\node[draw,rectangle,above,minimum width=\blockwidth,
		minimum height=\blockheight]
		(Charger) at ($(Battery.north)+(0,\arlen)$)
		{Charger};
	
	\node[draw,rectangle,right,minimum height=\blockheight]
		(FeedbackSampler) at ($(Charger.east)+(\arlen,0)$)
		{$Z_t=f(X_t,Y_t)$};
	
	\node (M) at ($(Tx.west)-(\arlen,0)$) {$M$};
	
	\tikzstyle{every path}=[->]
	\draw (Battery) -- (Tx);
	\draw (Tx) -- node[below] {$X_t$} (Channel);
	\draw (Channel) -- node[below] {$Y_t$} (Rx);
	\draw (Charger) -- node[left] {$E_t$} (Battery);
	\draw (M) -- (Tx);
	\draw (Rx) -- node[right,pos=1] {$\hat{M}$} 
		($(Rx.east)+(0.5,0)$);
	
	\tikzstyle{every path}=[->,dashed];
	
	\draw (M) |- (Charger);
	
	\draw (FeedbackSampler) -- 
		node[above] {$Z^{t-1}$} (Charger);
	
	\coordinate (ChannelRx) at 
		($(Channel.east) !.5! (Rx.west)$);
		
	\coordinate (TxChannel) at 
		($(Tx.east) !.5! (Channel.west)$);
	
	\coordinate (FeedbackSamplerXConnection) at
		($(FeedbackSampler.south west)+(0.4,0)$);
	
	\coordinate (XLineMidpoint) at
		($ (FeedbackSampler.south) !.5! (Channel.north) $);	
	
	\draw (TxChannel) |- 
		(XLineMidpoint -| FeedbackSamplerXConnection)
		-- node[left,pos=0.25] {$X^{t-1}$}
		(FeedbackSamplerXConnection);
	
	\draw (ChannelRx) -- 
		node[right,pos=0.5] {$Y^{t-1}$}
		(ChannelRx |- FeedbackSampler.south);
	
\end{tikzpicture}
\caption{Model of energy harvesting communication with a charger and generalized feedback signal.}
\label{fig:channel_model_general}
\end{figure}

In order to prove Theorem~\ref{thm:capacity}, we first consider a slightly more general system model. See Fig.~\ref{fig:channel_model_general}.
At time $t$, the charger observes a \emph{feedback signal} $Z^{t-1}$.
The feedback signal is the output of some deterministic function 
\[f:\mathcal{X}\times\mathcal{Y}\to\mathcal{Z},\]
where $\mathcal{Z}$ is the alphabet of the feedback signal.
This function is part of the channel parameters, i.e. it is determined by nature.
By considering different functions, such as $f(x,y)=0$, $f(x,y)=x$, and $f(x,y)=y$, we can recover the scenarios discussed in Section~\ref{sec:system_model}.

The transmitter encoding functions and receiver decoding function remain the same as in \eqref{eq:channel_encoding} and \eqref{eq:decoding}.
For the charger encoding functions, we will consider two cases:
\begin{description}[\IEEEsetlabelwidth{Case II.}]
\item[Case I.] The charger does not observe the message $M$.
The encoding functions are
\begin{align}
f_t^{\mathrm{C}}&:\mathcal{Z}^{t-1}
	\to\mathcal{E}
	,&t=1,\ldots,n.\label{eq:charger_encoding_noM}
\end{align}
\item[Case II.] The charger observes the message $M$. The encoding functions are
\begin{align}
f_t^{\mathrm{C}}&:\mathcal{M}\times\mathcal{Z}^{t-1}\to\mathcal{E},&&
t=1,\ldots,n.
\label{eq:charger_encoding}
\end{align}
\end{description}

We use the notion of Shannon strategies~\cite{Shannon1958} in a similar manner to~\cite{MaoHassibi2013,ShavivNguyenOzgur2016}.
The channel can be converted into an equivalent channel with no side information at the transmitter or at the charger, but with a different input alphabet.
Define two types of strategies: encoder strategies $u^n$ and charger strategies $v^n$.
An encoder strategy at time $t$ is a mapping $u_t:\mathcal{E}^t\to\mathcal{X}$, and the appropriate alphabet for blocklength $n$ is
$\mathcal{U}^n=\{u^n|\ u_t:\mathcal{E}^t\to\mathcal{X},
\ t=1,\ldots,n\}$.
A charger strategy is $v_t:\mathcal{Z}^{t-1}\to\mathcal{E}$, with alphabet
$\mathcal{V}^n=\{v^n|\ v_t:\mathcal{Z}^{t-1}\to\mathcal{E},
\ t=1,\ldots,n\}$.
Note that $\mathcal{U}^n$ and $\mathcal{V}^n$ are not Cartesian products of $n$ copies of a single alphabet, but a set of $n$-tuples where each element is defined above.
At time $t$, given the transmitter's observations $E^t$, the input symbol $X_t=U_t(E^t)$ is transmitted over the original channel. 
Similarly, given its own observations $Z^{t-1}$, the charger sets $E_t=V_t(Z^{t-1})$.
It is easy to see that the capacity of this channel is equal to that of our original channel, as coding strategies for one can be immediately translated to the other.

The strategies must satisfy the constraints as in the original channel, namely $U^n$ must satisfy the energy constraints~\eqref{eq:EH_constraint} and~\eqref{eq:EH_battery}:
\begin{align}
\phi(U_t(E^t))&\leq \min\{B_t+E_t,\bar{B}\},\label{eq:EH_constraint_strategies}\\*
B_t&=\min\{B_{t-1}+E_{t-1},\bar{B}\}-\phi(U_{t-1}(E^{t-1})),
\label{eq:EH_battery_strategies}
\end{align}
and $V_t$ must satisfy the cost constraint~\eqref{eq:cost_constraint}:
\begin{equation}
\frac{1}{n}\sum_{t=1}^{n}V_t(Z^{t-1})\leq\Gamma.
\label{eq:cost_constraint_strategies}
\end{equation}
As in Section~\ref{sec:capacity}, this is equivalent to writing $(X^n,E^n)\in\mathcal{A}_n(\Gamma)$ a.s., where
$X^n=U^n(E^n)$, $E^n=V^n(Z^{n-1})$, and $\mathcal{A}_n(\Gamma)$ was defined in~\eqref{eq:A_def}.

We will derive a capacity formula for each of the two cases using the above strategies.
Then we will derive the capacity expressions in Theorem~\ref{thm:capacity} by considering different functions for the feedback signal.
\begin{theorem}\label{thm:capacity_strategies}
The capacity of the channel in Fig.~\ref{fig:channel_model_general} when the charger does not observe the message (Case I) is given by:
\begin{equation}\label{eq:capacity_noM}
C_{\mathrm{I}}(\Gamma)=\lim_{n\to\infty}\frac{1}{n}
\max_{\substack{p(u^n),\ v^n:\\ (X^n,E^n)\in\mathcal{A}_n(\Gamma)\text{ a.s.}}}
I(U^n;Y^n),
\end{equation}
where the maximum is over all distributions $p(u^n)$ and \emph{deterministic} strategies $v^n$ such that 
$(X^n,E^n)=(U^n(E^n),v^n(Z^{n-1}))\in\mathcal{A}_n(\Gamma)$ a.s.\footnote{Note that even though $v^n$ is fixed, the energy sequence $E^n$ can be random, since it is a function of the RV $Z^{n-1}$.}

The capacity of the same channel when the charger observes the message (Case II) is given by:
\begin{equation}\label{eq:capacity_M}
C_{\mathrm{II}}(\Gamma)=\lim_{n\to\infty}\frac{1}{n}
\max_{\substack{p(u^n,v^n):\\ (X^n,E^n)\in\mathcal{A}_n(\Gamma)\text{ a.s.}}}
I(U^n,V^n;Y^n).
\end{equation}
\end{theorem}
\begin{proof}
We begin with the proof of achievability for Case I, where the charger does not observe the message.
The proof follows similar ideas as in~\cite{ShavivNguyenOzgur2016}. In fact, it more closely resembles the preceding~\cite{JogAnantharam2014}.

Recall the model in Fig.~\ref{fig:channel_model_general}.
We construct an achievable scheme composed of $k$ blocks with each block containing a transmitter strategy vector of length $n$ which is an element of $\mathcal{U}^n$, and a single charger strategy vector $v^n$, which will be used for all blocks.
As such, each codeword is a function of only the past $n$ energy arrivals and feedback samples, which means we ignore information regarding the previous blocks.
These codewords are designed to satisfy the energy constraints for initial battery level $B_1=\bar{B}$, in accordance with the assumption of Section~\ref{sec:system_model}. For this reason, we must ensure that the battery is completely full in the beginning of each block.
To achieve this, we allow the battery to ``recharge'' after we transmit each codeword by waiting a sufficient amount of time ($\ell$ time slots), during which the transmitter remains silent, and the charger transmits some fixed charging sequence.
By choosing $\ell$ large enough and an appropriate charging sequence, we can ensure the battery will be completely recharged at the beginning of the next block.

Fix integers $n$ and $\ell$, and fix a deterministic strategy vector $\mathbf{v}=v^n\in\mathcal{V}^n$ and a distribution $p(u^n)$ such that $(X^n,E^n)=(U^n(E^n),v^n(Z^{n-1}))\in\mathcal{A}_n(\Gamma)$ a.s.
For each message $m$, generate $k$ random strategy vectors independently $\mathbf{u}_i\sim p(u^n)$, $i=1,\ldots,k$. Recall that each $\mathbf{u}_i$ is a function on $\mathcal{E}^n$, and $\mathbf{v}$ is a function on $\mathcal{Z}^{n-1}$. The chosen message $m$ will be transmitted over $k$ blocks, each of size $n+\ell$, for a total transmit time of $k(n+\ell)$.
Hence, we will define codewords $u^{k(n+\ell)}\in\mathcal{U}^{k(n+\ell)}$ and $v^{k(n+\ell)}\in\mathcal{V}^{k(n+\ell)}$ using the above $\mathbf{u}_i$ and $\mathbf{v}$.

Fix $e_0\in\mathcal{E}$, $e_0>0$. By assumption, at least one such $e_0$ must exist (see Section~\ref{sec:system_model}).
This energy symbol will be used repeatedly by the charger to recharge the battery at the end of each block.
The scheme operates as follows:
consider block $i$, $1\leq i\leq k$, which takes place during times $t=(i-1)(n+\ell)+1$ to $t=i(n+\ell)$.
For the first part of the block, which consists of the first $n$ symbols, the transmitter sends the codeword $\mathbf{u}_i$, i.e. $u_{(i-1)(n+\ell)+1}^{(i-1)(n+\ell)+n}(e^{(i-1)(n+\ell)+n})=\mathbf{u}_i(\mathbf{e}_i)$, where $\mathbf{e}_i\triangleq e_{(i-1)(n+\ell)+1}^{(i-1)(n+\ell)+n}$ is the energy sequence during the first part of the block.
The charger uses the single fixed strategy vector $\mathbf{v}$, i.e. $v_{(i-1)(n+\ell)+1}^{(i-1)(n+\ell)+n}(z^{(i-1)(n+\ell)+n-1})=\mathbf{v}(\mathbf{z}_i)$, where $\mathbf{z}_i\triangleq z_{(i-1)(n+\ell)+1}^{(i-1)(n+\ell)+n-1}$ is the feedback signal observed by the charger during the first part of the block.
For the second part of the block, which consists of the remaining $\ell$ symbols, the transmitter remains silent, i.e. it simply sends zeros (recall the definition of the zero symbol in Section~\ref{sec:system_model}).
The charger sends the energy symbol $e_0$ for the duration of the second part.
To summarize, the transmitted block is
\begin{align}
	u_{(i-1)(n+\ell)+1}^{i(n+\ell)}\big(e^{i(n+\ell)}\big)
	&=\big[\mathbf{u}_i(\mathbf{e}_i),\ \mathbf{0}\big],
	\label{eq:u_block_def}\\*
	v_{(i-1)(n+\ell)+1}^{i(n+\ell)}\big(z^{i(n+\ell)-1}\big)
	&=\big[\mathbf{v}(\mathbf{z}_i),\ e_0\cdot\mathbf{1}\big],
	\label{eq:v_block_def}
\end{align}
where $\mathbf{0}$ is a vector of $\ell$ zeros and $\mathbf{1}$ is a vector of $\ell$ ones.

Observe that $u^{k(n+\ell)}(e^{k(n+\ell)})$ defined in~\eqref{eq:u_block_def} and $v^{k(n+\ell)}(z^{k(n+\ell)-1})$ defined in~\eqref{eq:v_block_def} are well-defined elements in $\mathcal{U}^{k(n+\ell)}$ and $\mathcal{V}^{k(n+\ell)}$, respectively.
Moreover, by choosing $\ell=\lceil\tfrac{\bar{B}}{e_0}\rceil$, we can ensure that the battery level at the beginning of each block will always be $\bar{B}$, even if the battery is empty at the end of the first part of the previous block. This implies that the energy constraints \eqref{eq:EH_constraint_strategies} and \eqref{eq:EH_battery_strategies} are satisfied.
Next, the average energy cost can be upper bounded as follows:
\begin{align*}
\hspace{2em}\lefteqn{\hspace{-2em}\frac{1}{k(n+\ell)}\sum_{t=1}^{k(n+\ell)}V_t(Z^{t-1})}\nonumber\\*
&=\frac{1}{k}\sum_{i=1}^{k}\frac{1}{n+\ell}\left(\sum_{j=1}^{n}v_j\big(Z_{(i-1)(n+\ell)+1}^{(i-1)(n+\ell)+j}\big)+\ell \cdot e_0\right)\\
&\overset{\text{(i)}}{\leq} \frac{1}{k}\sum_{i=1}^{k}\frac{1}{n+\ell}\left(n\Gamma+\ell e_0\right)\\
&=\Gamma+\frac{\ell}{n+\ell}(e_0-\Gamma)\\
&\overset{\text{(ii)}}{\leq}\Gamma+\epsilon,
\end{align*}
where (i) is by construction of the strategy $\mathbf{v}$, and (ii) is true for any $\epsilon>0$ if $n$ is large enough (recall $\ell=\lceil\bar{B}/e_0\rceil$ is fixed).
Hence, the average cost is less than $\Gamma+\epsilon$.
By continuity of $C_I(\Gamma)$, this implies~\eqref{eq:cost_constraint_strategies} is satisfied.

Denote the channel output during the first part of block $i$ by $\mathbf{y}_i=y_{(i-1)(n+\ell)+1}^{(i-1)(n+\ell)+n}$.
The receiver observes $y^{k(n+\ell)}$ but makes use only of $\mathbf{y}^k=(\mathbf{y}_1,\ldots,\mathbf{y}_k)$ for decoding, by applying standard jointly typical decoding with~$\mathbf{u}^k$.
By construction, $\mathbf{y}^k$ is essentially $k$ samples of the output of a  memoryless vector channel from $U^n$ to $Y^n$.
Taking $k\to\infty$, we get by standard joint typicality arguments that the rate $\frac{1}{n+\ell}I(\mathbf{U};\mathbf{Y})$ is achievable.
Therefore, for $n$ large enough, we have
\[
C_{\mathrm{I}}(\Gamma)\geq\frac{1}{n+\ell}I(U^n;Y^n).
\]
Since $p(u^n)$ and $v^n$ were arbitrary, we can maximize over them and take $n\to\infty$, yielding:\footnote{It can be argued that the maximum exists: the cardinality of the set $\mathcal{V}^n$ is finite, and for every $v^n\in\mathcal{V}^n$ the set of $p(u^n)$ that satisfy the constraint is convex, and mutual information is continuous in $p(u^n)$.}
\begin{equation}
C_{\mathrm{I}}(\Gamma)\geq\limsup_{n\to\infty}\frac{1}{n}
\max_{\substack{p(u^n),\ v^n:\\ (X^n,E^n)\in\mathcal{A}_n(\Gamma)\text{ a.s.}}}I(U^n;Y^n).
\label{eq:CaseI_achievability}
\end{equation}

For the converse part, consider an $(R,n)$ code for the channel.
This consists of a set of functions $f_t^{\mathrm{C}}:\mathcal{Z}^{t-1}\to\mathcal{E}$, $t=1,\ldots,n$, as in~\eqref{eq:charger_encoding_noM}, such that the average cost constraint is satisfied.
These functions constitute one fixed strategy vector $v^n$;
namely $v_t=f_t^{\mathcal{C}}$ for $t=1,\ldots,n$. 
Moreover, the code consists of a set of strategy vectors $u^n(m)$, $m\in\mathcal{M}$, where each strategy vector is again a set of functions $u_t:\mathcal{E}^t\to\mathcal{X}$.
By Fano's inequality, we have:
\[
	H(M|Y^n)\leq H_2(P_e)+P_e\cdot nR.
\]
By the data processing inequality:
\begin{align*}
R&\leq \frac{1}{(1-P_e)n}[I(M;Y^n)+H_2(P_e)]\\
&\leq\frac{1}{(1-P_e)n}[I(U^n;Y^n)+H_2(P_e)],
\end{align*}
where $I(U^n;Y^n)$ is the mutual information evaluated for $p(u^n)$ induced by the code.
Since all codewords must satisfy the input constraints~\eqref{eq:EH_constraint_strategies}--\eqref{eq:cost_constraint_strategies}, this implies $(X^n,E^n)\in\mathcal{A}_n(\Gamma)$ a.s.
Therefore
\[
	R\leq\frac{1}{(1-P_e)n}
	\left[\max_{\substack{p(u^n),\ v^n:\\ (X^n,E^n)\in\mathcal{A}_n(\Gamma)\text{ a.s.}}}
	I(U^n;Y^n)+H_2(P_e)\right].
\]
If $R$ is achievable, there exists a sequence of $(R,n)$ codes for which $P_e\to0$ as $n\to\infty$, hence
\begin{equation*}
	R\leq\liminf_{n\to\infty}\frac{1}{n}
	\max_{\substack{p(u^n),\ v^n:\\ (X^n,E^n)\in\mathcal{A}_n(\Gamma)\text{ a.s.}}}
	I(U^n;Y^n),
\end{equation*}
which implies
\begin{equation}
	C_{\mathrm{I}}(\Gamma)\leq\liminf_{n\to\infty}\frac{1}{n}
	\max_{\substack{p(u^n),\ v^n:\\ (X^n,E^n)\in\mathcal{A}_n(\Gamma)\text{ a.s.}}}
	I(U^n;Y^n).
	\label{eq:CaseI_converse}
\end{equation}
Together with~\eqref{eq:CaseI_achievability}, this implies that the limit exists and is given by~\eqref{eq:capacity_noM}.

In Case II, where the charger observes the message, the transmitter and the charger can agree on a joint codebook.
A codebook is a list of $2^{nR}$ pairs of strategy vectors $(u^n,v^n)$, where each message $m$ is assigned a pair $(u^n(m),v^n(m))$.
Achievability follows by similar arguments as before.
In this case, we generate random strategies according to some $p(u^n,v^n)$ such that $(X^n,E^n)\in\mathcal{A}_n(\Gamma)$.
We transmit $k$ blocks of size $n+\ell$, where $\ell=\lceil\frac{\bar{B}}{e_0}\rceil$ as before, for some fixed $e_0\in\mathcal{E}$, $e_0>0$.
The first part of the block consists of the $n$ symbols of the random codeword, i.e. the transmitter outputs $u^n$ and the charger outputs $v^n$.
During the second part, the transmitter outputs zeros and the charger outputs $e_0>0$.
As before, we can write the codewords $u^{k(n+\ell)}$ and $v^{k(n+l)}$ as follows:
\begin{align*}
	u_{(i-1)(n+\ell)+1}^{i(n+\ell)}\big(e^{i(n+\ell)}\big)
	&=\big[\mathbf{u}_i(\mathbf{e}_i),\ \mathbf{0}\big],\\
	v_{(i-1)(n+\ell)+1}^{i(n+\ell)}\big(z^{i(n+\ell)-1}\big)
	&=\big[\mathbf{v}_i(\mathbf{z}_i),\ e_0\cdot\mathbf{1}\big].
\end{align*}
The battery will be full at the beginning of each block, and the energy constraints are satisfied.
The average energy cost is again upper bounded by $\Gamma+\epsilon$ for any small $\epsilon$ if $n$ is large enough.
The receiver observes $\mathbf{y}^k$ and applies standard joint typical decoding with $(\mathbf{u}^k,\mathbf{v}^k)$.
When the number of blocks $k\to\infty$, we can achieve the rate $\frac{1}{n+\ell}I(U^n,V^n;Y^n)$.
Then:
\[
C_{\mathrm{II}}(\Gamma)\geq\limsup_{n\to\infty}\frac{1}{n}
\max_{\substack{p(u^n,v^n):\\ (X^n,E^n)\in\mathcal{A}_n(\Gamma)\text{ a.s.}}}
I(U^n,V^n;Y^n).
\]
The converse follows from Fano's inequality similarly to the previous case.
We conclude that~\eqref{eq:capacity_M} holds.
\end{proof}


In what follows, we apply Theorem~\ref{thm:capacity_strategies} to derive the results of Theorem~\ref{thm:capacity}.

\begin{proof}[Proof of Theorem~\ref{thm:capacity}]

\ 
\subsubsection{Generic Charger}

When the charger does not observe any side information, we consider the model in Case I and apply Theorem~\ref{thm:capacity_strategies} with feedback signal $Z_t=0$.
Since the strategy $v^n$ is essentially a function over an empty set, it can be replaced with a fixed sequence $e^n$, hence \eqref{eq:capacity_noM} reduces to:
\[
C_{\emptyset}(\Gamma)
=\lim_{n\to\infty}\frac{1}{n}
\max_{\substack{p(u^n),\ e^n:\\ (X^n,e^n)\in\mathcal{A}_n(\Gamma)\text{ a.s.}}}
I(U^n;Y^n).
\]
Since $e^n$ is deterministic, we have:
\begin{align*}
I(U^n;Y^n)
&=I(U^n,X^n;Y^n)\\
&=I(X^n;Y^n),
\end{align*}
where the second line is due to the Markov chain $U^n-X^n-Y^n$.
Observe that for fixed $e^n$, $p(u^n)$ induces a distribution $p(x^n)$, and since the objective does not depend on $U^n$, this yields \eqref{eq:C0}.

\subsubsection{Receiver Charges Transmitter}

Here the charger observes $Y^{t-1}$.
This scenario is realized by considering Case I, where the charger does not observe the message, and setting the feedback signal to be $Y_t$, i.e. $Z_t=f(X_t,Y_t)=Y_t$.
Theorem~\ref{thm:capacity_strategies} gives:
\[
C_Y(\Gamma)=\lim_{n\to\infty}\frac{1}{n}
\max_{\substack{p(u^n),\ v^n:\\ (X^n,E^n)\in\mathcal{A}_n(\Gamma)\text{ a.s.}}}
I(U^n;Y^n),
\]
where the deterministic charger strategies are $v_t:\mathcal{Y}^{t-1}\to\mathcal{E}$, and the encoder strategies are $U_t:\mathcal{E}^t\to\mathcal{X}$, for $t=1,\ldots,n$.
For fixed $v^n$ and $p(u^n)$:
\begin{align*}
I(U^n;Y^n)
&=\sum_{t=1}^{n}I(U^n;Y_t|Y^{t-1})\\
&\overset{\text{(i)}}{=}\sum_{t=1}^{n}I(U^n;Y_t|Y^{t-1},E^t)\\
&\overset{\text{(ii)}}{=}\sum_{t=1}^{n}I(U^n,X^t;Y_t|Y^{t-1},E^t)\\
&\overset{\text{(iii)}}{=}\sum_{t=1}^nI(X^t;Y_t|Y^{t-1},E^t)\\
&\overset{\text{(iv)}}{=}\sum_{t=1}^{n}I(X^t;Y_t|Y^{t-1})\\
&=I(X^n\to Y^n),
\end{align*}
where (i) is since $E^t=v^t(Y^{t-1})$, a deterministic function of $Y^{t-1}$;
(ii) is because $X^t=U^t(E^t)$;
(iii) is due to the Markov chain 
$Y_t-X_t-(Y^{t-1},E^t,X^{t-1},U^n)$
which implies the Markov chain $Y_t-(X^t,E^t,Y^{t-1})-U^n$;
and (iv) is again because $E^t$ is a deterministic function of $Y^{t-1}$.
A distribution $p(u^n)$ induces a causally conditioned distribution $p(x^n\|e^n)$,
and since the objective does not depend on $U^n$, we can optimize over $p(x^n\|e^n)$ to yield~\eqref{eq:CY}.

\subsubsection{Charger Adjacent to Transmitter}

When the charger observes the input $X^{t-1}$,
we set the feedback signal to $Z_t=f(X_t,Y_t)=X_t$
and apply Theorem~\ref{thm:capacity_strategies} for Case I:
\[
C_X(\Gamma)=\lim_{n\to\infty}\frac{1}{n}
\max_{\substack{p(u^n),\ v^n:\\ (X^n,E^n)\in\mathcal{A}_n(\Gamma)\text{ a.s.}}}
I(U^n;Y^n),
\]
where the encoder strategies are $U_t:\mathcal{E}^t\to\mathcal{X}$ and the charger strategies are $v_t:\mathcal{X}^{t-1}\to\mathcal{E}$.
For any $p(u^n)$ and $v^n$ we have
\begin{align*}
I(U^n;Y^n)
&\overset{\text{(i)}}{=}I(U^n,E_1;Y^n)\\
&\overset{\text{(ii)}}{=}I(U^n,E_1,X_1;Y^n)\\
&\overset{\text{(iii)}}{=}I(U^n,E^2,X_1;Y^n)\\
&=\ldots\\
&\overset{\text{(iv)}}{=}I(U^n,E^n,X^n;Y^n)\\
&\overset{\text{(v)}}{=}I(X^n;Y^n),
\end{align*}
where (i) is because $v_1$ is a strategy over an empty set, therefore $E_1=v_1$ is fixed;
(ii) is because $X_1=U_1(E_1)$;
(iii) is because $E_2=v_2(X_1)$;
and so forth.
In general, $E_t=v_t(X^{t-1})$ and $X_t=U_t(E^t)$ for $t=1,\ldots,n$, giving equality~(iv).
Finally, (v) is due to the Markov chain $(U^n,E^n)-X^n-Y^n$.

Observe that $p(u^n)$ induces a causally conditioned distribution on $p(x^n\|e^n)$.
Moreover, since $E^t$ is a deterministic function of $X^{t-1}$, we have the Markov chain $X_t-X^{t-1}-E^t$, which implies $p(x^n\|e^n)=p(x^n)$.
Since the objective depends only on $p(x^n)$, capacity can be written as~\eqref{eq:CX}.

\subsubsection{Fully Cognitive Charger}

When the charger observes the message, we consider Case II and set the feedback signal to zero: $Z_t=0$.
We apply Theorem~\ref{thm:capacity_strategies}. As in the derivation of $C_\emptyset(\Gamma)$, the strategies $V^n$ are over an empty set, therefore we can code directly over $E^n$:
\[
C_M(\Gamma)=\lim_{n\to\infty}\frac{1}{n}
\max_{\substack{p(u^n,e^n):\\ (X^n,E^n)\in\mathcal{A}_n(\Gamma)\text{ a.s.}}}
I(U^n,E^n;Y^n).
\]
We have:
\begin{align*}
I(U^n,E^n;Y^n)
&\overset{\text{(i)}}{=}I(U^n,E^n,X^n;Y^n)\\*
&\overset{\text{(ii)}}{=}I(X^n;Y^n),
\end{align*}
where (i) is because $X^n=U^n(E^n)$ and (ii) is due to the Markov chain $(U^n,E^n)-X^n-Y^n$.
Since the objective does not depend on $U^n$, we can equivalently optimize over all distributions $p(x^n,e^n)$, which gives~\eqref{eq:CM}.
\end{proof}

\section{Proof of Proposition~\ref{prop:average_power_upper_bound}}
\label{sec:average_power_upper_bound_proof}

First, observe that by standard arguments, the function $\Cub(\Gamma)$ is concave and non-decreasing in $\Gamma$ (see e.g.~\cite[Lemma 10.4.1]{CoverThomas2006}).
Fix an $(R,n)$ code for the channel.
From the energy constraints~\eqref{eq:EH_constraint} and~\eqref{eq:EH_battery}:
\begin{align}
\phi(X_t)&\leq B_t+E_t,\label{eq:X_power_constraint}\\
\phi(X_{t-1})+B_t&\leq B_{t-1}+E_{t-1}.\label{eq:X_power_battery}
\end{align}
Summing up~\eqref{eq:X_power_constraint} for $t=n$ with~\eqref{eq:X_power_battery} for $t=2,\ldots,n$ yields:
\[
\sum_{t=1}^{n}\phi(X_t)+\sum_{t=2}^{n}B_t
\leq \sum_{t=1}^{n}B_t+\sum_{t=1}^{n}E_t,
\]
which, along with the energy cost constraint~\eqref{eq:cost_constraint}, gives
\begin{align}
\frac{1}{n}\sum_{t=1}^{n}\phi(X_t)
&\leq\frac{1}{n}B_1+\frac{1}{n}\sum_{t=1}^{n}E_t\nonumber\\
&\leq\frac{1}{n}\bar{B}+\Gamma.\label{eq:average_power_constraint}
\end{align}

By Fano's inequality, $H(M|Y^n)\leq n\epsilon_n$, where $\epsilon_n\to0$ as the probability or error $P_e\to0$.
We have:
\begin{align*}
nR=H(M)
&\leq I(M;Y^n)+n\epsilon_n\\
&\leq \sum_{t=1}^{n}I(M,E^t;Y_t|Y^{t-1})+n\epsilon_n\\
&\overset{\text{(i)}}{\leq} 
	\sum_{t=1}^{n}I(X_t;Y_t|Y^{t-1})+n\epsilon_n\\
&\overset{\text{(ii)}}{\leq} 
	\sum_{t=1}^{n}I(X_t;Y_t)+n\epsilon_n\\
&\overset{\text{(iii)}}{\leq} \sum_{t=1}^{n}\Cub(\mathbb{E}[\phi(X_t)])+n\epsilon_n\\
&\overset{\text{(iv)}}{\leq} n\Cub\left(\frac{1}{n}\sum_{t=1}^{n}\mathbb{E}[\phi(X_t)]\right)
	+n\epsilon_n\\
&\overset{\text{(v)}}{\leq} n\Cub\left(\Gamma+\tfrac{1}{n}\bar{B}\right)+n\epsilon_n,
\end{align*}
where (i) is due to~\eqref{eq:channel_encoding} and the data processing inequality;
(ii) is because the channel is memoryless;
(iii) is by the definition of $\Cub(\Gamma)$;
(iv) is by concavity of $\Cub(\Gamma)$;
and (v) is by~\eqref{eq:average_power_constraint} and because $\Cub(\Gamma)$ is non-decreasing.
Finally, since concavity implies continuity, taking $n\to\infty$ gives $R\leq\Cub(\Gamma)$.\qed

\section{Proof of Lemma~\ref{lemma:relaxed_power_constraint}}
\label{sec:expected_average_power_constraint}

First, since $\sum_{t=1}^{n}E_t\leq n\Gamma$ a.s. implies $\sum_{t=1}^{n}\mathbb{E}[E_t]\leq n\Gamma$, from~\eqref{eq:CX} and~\eqref{eq:CM} we clearly have
\begin{align*}
C_X(\Gamma)&\leq\liminf_{n\to\infty}\frac{1}{n}
	\max_{\substack{p(x^n),\ \{e_t(x^{t-1})\}_{t=1}^{n}:\\
		(X^n,e^n(X^{n-1}))\in\mathcal{A}'_n\text{ a.s.}\\
	\sum_{t=1}^{n}\mathbb{E}[E_t]\leq n\Gamma}}
	H(X^n),\\
C_M(\Gamma)&\leq\liminf_{n\to\infty}\frac{1}{n}
	\max_{\substack{p(x^n,e^n):\\ 
	(X^n,E^n)\in\mathcal{A}'_n\text{ a.s.}\\
	\sum_{t=1}^{n}\mathbb{E}[E_t]\leq n\Gamma}}
	H(X^n).
\end{align*}
To show the other direction, we propose an achievable scheme, which is similar to the one in Appendix~\ref{sec:strategies_capacity_and_special_cases}.
Fix $n$ and $\ell$, and generate i.i.d. codewords of length $n$ from the appropriate distribution in
the above capacity expressions.
Although these codewords clearly do not satisfy the a.s. energy constraint imposed by our system, we will use them to construct a new codebook that does satisfy it.
As in Appendix~\ref{sec:strategies_capacity_and_special_cases}, concatenate $k$ such codewords along with recharge times of length $\ell$, in which the transmitter remains silent and the charger sends a fixed positive energy symbol $e_0\in\mathcal{E}$.
Denote $N=k(n+\ell)$.

If the channel did not have a total energy cost constraint $\sum_{t=1}^{N}E_t\leq N\Gamma$, clearly this strategy would be admissible. By applying joint typicality decoding at the receiver, it would yield a vanishing probability of error if $R\geq\frac{1}{n+\ell}H(X^n)$.
To be precise, if we denote the error event of this code by $\mathcal{E}'$, we have $\Pr(\mathcal{E}')\leq\varepsilon_k$, where $\varepsilon_k\to0$ when $k\to\infty$.
This follows from an analysis similar to the one in Appendix~\ref{sec:strategies_capacity_and_special_cases}.

However, indeed the charging sequence $e^{N}$ may not satisfy the cost constraint $\sum_{t=1}^{N}E_t\leq N\Gamma$. To rectify this, first observe that for $C_X(\Gamma)$ and $C_M(\Gamma)$, the transmitter knows the entire charging sequence $e^{N}$ ahead of time, based on the message $m$.
It will then transmit $x^{N}(m)$ if $\sum_{t=1}^{N}e_t(m)\leq N(\Gamma+\epsilon)$ for some fixed $\epsilon>0$, and the all-zeros codeword otherwise. 
The decoder performs joint typicality decoding as in Appendix~\ref{sec:strategies_capacity_and_special_cases}.

To analyze the probability of error, denote the event of decoding error at the receiver by $\mathcal{E}$, and let
\[
\mathcal{E}_0=\big\{\sum_{t=1}^{N}E_t>N(\Gamma+\epsilon)\big\}.
\]
The probability of error can be upper bounded as follows:
\begin{align*}
\Pr(\mathcal{E})
&=\Pr(\mathcal{E}\cap\mathcal{E}_0^c)
	+\Pr(\mathcal{E}\cap\mathcal{E}_0)\\
&\leq\Pr(\mathcal{E}\cap\mathcal{E}_0^c)
	+\Pr(\mathcal{E}_0)\\
&\overset{\text{(i)}}{=}\Pr(\mathcal{E}'\cap\mathcal{E}_0^c)
	+\Pr(\mathcal{E}_0)\\
&\leq\Pr(\mathcal{E}')+\Pr(\mathcal{E}_0)\\
&\leq\varepsilon_k+\Pr(\mathcal{E}_0),
\end{align*}
where (i) follows from the fact that when the energy cost constraint is satisfied, the original codeword is transmitted as if there was no cost constraint.
For the second term, denote the average energy of the first $n$ samples of block $i$ by $\Gamma_i$:
\[
\Gamma_i=\frac{1}{n}\sum_{t=(i-1)(n+\ell)+1}^{(i-1)(n+\ell)+n}E_t.
\]
Then
\begin{align*}
\Pr(\mathcal{E}_0)&=\Pr\Big\{\sum_{i=1}^{k}(n\Gamma_i+\ell e_0)>k(n+\ell)(\Gamma+\epsilon)\Big\}\\
&=\Pr\Big\{\tfrac{1}{k}\sum_{i=1}^{k}\Gamma_i>\tfrac{n+\ell}{n}(\Gamma+\epsilon)-\tfrac{\ell e_0}{n}\Big\}
\end{align*}
Since $\ell=\lceil\bar{B}/e_0\rceil$ is fixed, for $n$ large enough we can write
\[
\Pr(\mathcal{E}_0)=\Pr\Big\{\tfrac{1}{k}\sum_{i=1}^{k}\Gamma_i>\Gamma+\epsilon'\Big\}
\]
for some small $\epsilon'>0$.
Since by construction the $\Gamma_i$'s are i.i.d. with mean $\Gamma$, we have by the law of large numbers that $\Pr(\mathcal{E}_0)\to0$ when $k\to\infty$.

Taking $n\to\infty$, we get
\begin{align*}
C_X(\Gamma)&\geq\limsup_{n\to\infty}\frac{1}{n}
	\max_{\substack{p(x^n),\ \{e_t(x^{t-1})\}_{t=1}^{n}:\\
		(X^n,e^n(X^{n-1}))\in\mathcal{A}'_n\text{ a.s.}\\
	\sum_{t=1}^{n}\mathbb{E}[E_t]\leq n\Gamma}}
	H(X^n),\\
C_M(\Gamma)&\geq\limsup_{n\to\infty}\frac{1}{n}
	\max_{\substack{p(x^n,e^n):\\ 
	(X^n,E^n)\in\mathcal{A}'_n\text{ a.s.}\\
	\sum_{t=1}^{n}\mathbb{E}[E_t]\leq n\Gamma}}
	H(X^n),
\end{align*}
which concludes the proof.\qed

\section{Proof of Lemma~\ref{lemma:Lagrange_multipliers}}
\label{sec:Lagrange_multipliers_proof}

To simplify the exposition,
observe that each one of the capacity expressions \eqref{eq:C0_relaxed}--\eqref{eq:CM_relaxed} can be expressed as 
\begin{equation}
\label{eq:CD}
C_{\mathcal{D}}(\Gamma)=\lim_{n\to\infty}\frac{1}{n}\max_{\substack{p(x^n,e^n)\in\mathcal{D}_n:\\
(X^n,E^n)\in\mathcal{A}'_n\text{ a.s.}\\
\sum_{t=1}^{n}\mathbb{E}[E_t]\leq n\Gamma}}H(X^n),
\end{equation}
where the sets $\mathcal{D}_n$ are:
\[
\mathcal{D}_n=\{p(x^n,e^n)\ |\ p(x^n,e^n)=1(e^n)p(x^n)\}
\]
for $C_\emptyset(\Gamma)$;
\[
\mathcal{D}_n=\{p(x^n,e^n)\ |\ p(x^n,e^n)=\prod_{t=1}^{n}p(x_t|x^{t-1})1(e_t|x^{t-1})\}
\]
for $C_X(\Gamma)$;
and for $C_M(\Gamma)$ we take $\mathcal{D}_n$ to be simply the set of all probability distributions over $\mathcal{X}^n\times\mathcal{E}^n$.

Consider the following optimization problem:
\begin{equation}
J_{\mathcal{D}}(\rho)=\sup\liminf_{n\to\infty}\frac{1}{n}
\big(H(X^n)-\rho\sum_{t=1}^{n}\mathbb{E}[E_t]\big),
\label{eq:Jrho_def}
\end{equation}
where the supremum is over all probability distributions $\{p(x_t,e_t|x^{t-1},e^{t-1})\}_{t=1}^{\infty}\in\mathcal{D}_\infty$ s.t. $(\mathbf{X},\mathbf{E})\in\mathcal{A}'_\infty\text{ a.s.}$
We will show that if the process $(\mathbf{X}^\star,\mathbf{E}^\star)$ approaches the supremum in~\eqref{eq:Jrho_def} up to $\epsilon>0$ with
\begin{align*}
\underline{\Gamma}&=\liminf_{n\to\infty}\frac{1}{n}\sum_{t=1}^{n}\mathbb{E}[E_t^\star],\\
\overline{\Gamma}&=\limsup_{n\to\infty}\frac{1}{n}\sum_{t=1}^{n}\mathbb{E}[E_t^\star],
\end{align*}
and $\Gamma^\alpha=\alpha\underline{\Gamma}+(1-\alpha)\overline{\Gamma}$,
then for any $0\leq\alpha\leq 1$, the capacity in~\eqref{eq:CD} is bounded by
\begin{equation}
J_{\mathcal{D}}(\rho)+\rho\Gamma^\alpha-\epsilon\leq
C_{\mathcal{D}}(\Gamma^\alpha)\leq J_{\mathcal{D}}(\rho)+\rho\Gamma^\alpha.
\label{eq:CD_bounds}
\end{equation}

We start with the upper bound.
Fix $\Gamma\geq0$,
and let $(\tilde{X}^n,\tilde{E}^n)\sim\tilde{p}(x^n,e^n)$ attain the maximum in \eqref{eq:CD}.
We construct a process $(\mathbf{X},\mathbf{E})$ in a similar manner to the construction in Appendix~\ref{sec:strategies_capacity_and_special_cases}:
concatenate i.i.d. copies of $(\tilde{X}^n,\tilde{E}^n)$, where each $n$-block is followed by $\ell$ time slots in which $X=0$ and $E=e_0$, for some positive $e_0\in\mathcal{E}$.
More precisely, for $t=(i-1)(n+\ell)+j$, $i\geq1$, $1\leq j\leq n+\ell$, set
\[
p(x_t,e_t|x^{t-1},e^{t-1})=\tilde{p}(x_t,e_t|x_{(i-1)(n+\ell)+1}^{t-1},e_{(i-1)(n+\ell)+1}^{t-1})
\]
if $1\leq j\leq n$, and
\[
p(x_t,e_t|x^{t-1},e^{t-1})=1\{x_t=0\}\cdot 1\{e_t=e_0\}
\]
if $n+1\leq j\leq n+\ell$.
By choosing $\ell=\lceil\bar{B}/e_0\rceil$, we guarantee $(\mathbf{X},\mathbf{E})\in\mathcal{A}'_{\infty}$.
Under this distribution, we have:
\begin{align}
&\lim_{k\to\infty}\frac{1}{k}\big(H(X^k)-\rho\sum_{t=1}^{k}\mathbb{E}[E_t]\big)\nonumber\\*
&\qquad=\frac{1}{n+\ell}\big(H(\tilde{X}^n)-\rho\sum_{t=1}^{n}\mathbb{E}[\tilde{E}_t]-\rho\ell e_0\big).
\label{eq:Jrho_suboptimal_pmf}
\end{align}
Define
\[
C_{\mathcal{D}}^{(n)}(\Gamma)\triangleq \frac{1}{n}\max_{\substack{p(x^n,e^n)\in\mathcal{D}_n:\\
(X^n,E^n)\in\mathcal{A}'_n\text{ a.s.}\\
\sum_{t=1}^{n}\mathbb{E}[E_t]\leq n\Gamma}}H(X^n).
\]
Then $H(\tilde{X}^n)=n\cdot C_{\mathcal{D}}^{(n)}(\Gamma)$, and we have
\begin{align*}
&\frac{1}{n+\ell}\big(H(\tilde{X}^n)-\rho\sum_{t=1}^{n}\mathbb{E}[\tilde{E}_t]-\rho\ell e_0\big)\nonumber\\*
&\qquad\geq
\frac{n}{n+\ell}\big(C_{\mathcal{D}}^{(n)}(\Gamma)-\rho(\Gamma+\tfrac{\ell e_0}{n})\big).
\end{align*}
Substituting in \eqref{eq:Jrho_suboptimal_pmf} and taking $n\to\infty$:
\begin{equation}
J_{\mathcal{D}}(\rho)\geq C_{\mathcal{D}}(\Gamma)-\rho\Gamma.
\label{eq:Jrho_lower_bound}
\end{equation}
Note that this holds for \emph{any} $\Gamma,\rho\geq0$ for which there is at least one feasible solution to \eqref{eq:CD}.

For the lower bound, let $(\mathbf{X}^\star,\mathbf{E}^\star)$ be such that
\begin{equation}
\liminf_{n\to\infty}\frac{1}{n}\big(H(X^{\star n})-\rho\sum_{t=1}^{n}\mathbb{E}[E_t^\star]\big)\geq J_{\mathcal{D}}(\rho)-\epsilon.
\label{eq:process_sup_epsilon_inequality}
\end{equation}
For $n\geq1$, let $p^\star(x^n,e^n)$ be the corresponding marginal, and denote $\Gamma_n\triangleq\frac{1}{n}\sum_{t=1}^{n}\mathbb{E}[E_t^\star]$.
We have:
\[
\frac{1}{n}\big(H(X^{\star n})-\rho\sum_{t=1}^{n}\mathbb{E}[E^\star_t]\big)
\leq C_{\mathcal{D}}^{(n)}(\Gamma_n)-\rho\Gamma_n.
\]
Consider $\underline{\Gamma}=\liminf_{n\to\infty}\Gamma_n$. For any $\delta>0$, there exists a subsequence $\{n_k\}_{k=1}^{\infty}$ such that $|\Gamma_{n_k}-\underline{\Gamma}|\leq\delta$.
Since $C_{\mathcal{D}}^{(n)}(\Gamma)$ is non-decreasing in $\Gamma$, we get:
\[
\frac{1}{n_k}\big(H(X^{\star n_k})-\rho\sum_{t=1}^{n_k}\mathbb{E}[E^\star_t]\big)
\leq C_{\mathcal{D}}^{(n_k)}(\underline{\Gamma}+\delta)-\rho(\underline{\Gamma}-\delta).
\]
Taking $k\to\infty$:
\begin{align*}
&\liminf_{k\to\infty}\frac{1}{n_k}\big(H(X^{\star n_k})-\rho\sum_{t=1}^{n_k}\mathbb{E}[E^\star_t]\big)\nonumber\\*
&\qquad\leq \liminf_{k\to\infty}\big\{C_{\mathcal{D}}^{(n_k)}(\underline{\Gamma}+\delta)-\rho(\underline{\Gamma}-\delta)\big\}.
\end{align*}
Since $C_{\mathcal{D}}^{(n)}(\Gamma)$ converges to a finite limit, any subsequence also converges to the same limit.
For the LHS, the $\liminf$ of any subsequence is bounded below by the $\liminf$ of the sequence itself. Therefore:
\[
\liminf_{n\to\infty}\frac{1}{n}\big(H(X^{\star n})-\rho\sum_{t=1}^{n}\mathbb{E}[E^\star_t]\big)
\leq C_{\mathcal{D}}(\underline{\Gamma}+\delta)-\rho(\underline{\Gamma}-\delta).
\]
By standard time-sharing arguments, $C_{\mathcal{D}}(\Gamma)$ is concave and therefore also continuous.
Hence, substituting in~\eqref{eq:process_sup_epsilon_inequality} and taking $\delta\to0$ we obtain:
\begin{equation}
J_{\mathcal{D}}(\rho)-\epsilon\leq C_{\mathcal{D}}(\underline{\Gamma})-\rho\underline{\Gamma}.
\label{eq:Jrho_upper_bound_liminf}
\end{equation}
Next, we can repeat exactly the same steps for $\overline{\Gamma}=\limsup_{n\to\infty}\Gamma_n$:
\begin{equation}
J_{\mathcal{D}}(\rho)-\epsilon\leq C_{\mathcal{D}}(\overline{\Gamma})-\rho\overline{\Gamma}.
\label{eq:Jrho_upper_bound_limsup}
\end{equation}
Combining \eqref{eq:Jrho_upper_bound_liminf} and \eqref{eq:Jrho_upper_bound_limsup}, we have for any $0\leq\alpha\leq1$:
\begin{align}
J_{\mathcal{D}}(\rho)-\epsilon
&\leq \alpha \big(C_{\mathcal{D}}(\underline{\Gamma})
	-\rho\underline{\Gamma}\big)
	+(1-\alpha)\big( C_{\mathcal{D}}(\overline{\Gamma})
	-\rho\overline{\Gamma}\big)\nonumber\\
&\leq C_{\mathcal{D}}(\Gamma^\alpha)-\rho\Gamma^\alpha,
\label{eq:Jrho_upper_bound}
\end{align}
where 
the second inequality is due to concavity of $C_{\mathcal{D}}(\Gamma)$.
Finally, since \eqref{eq:Jrho_lower_bound} holds for any $\Gamma$, we get \eqref{eq:CD_bounds}.\qed

\section{Equivalent Expression for $J_\emptyset(\rho)$ via Sequence Counting}
\label{sec:generic_charger_num_of_seq}

Recall the definition of $J_\emptyset(\rho)$ and $\mathbf{N}_n$ in Section~\ref{subsec:universal_charger_noiseless}.
For any $p(x^n)$ and $e^n$ s.t. $(X^n,e^n)\in\mathcal{A}'_n$ a.s., we clearly have
$H(X^n)\leq \log(\mathbf{N}_n\cdot\mathbf{1})$.
Therefore:
\begin{align}
J_\emptyset(\rho)&=\sup_{\{e_t\}_{t=1}^{\infty}}
\sup_{\substack{\{p(x_t|x^{t-1})\}_{t=1}^{\infty}:\\
(\mathbf{X},\mathbf{e})\in\mathcal{A}'_\infty\text{ a.s.}}}
\liminf_{n\to\infty}\frac{1}{n}\big(H(X^n)-\rho\sum_{t=1}^{n}e_t\big)\nonumber\\*
&\leq\sup_{\{e_t\}_{t=1}^{\infty}}\liminf_{n\to\infty}\frac{1}{n}\big(\log(\mathbf{N}_n\cdot\mathbf{1})-\rho\sum_{t=1}^{n}e_t\big),
\label{eq:J0_dir1}
\end{align}
because $(\mathbf{X},\mathbf{e})\in\mathcal{A}'_\infty$ implies $(X^n,e^n)\in\mathcal{A}'_n$ for every $n$.

To show the other direction, take an arbitrary sequence $\{\tilde{e}_t\}_{t=1}^{\infty}$ and fix an integer $n\geq1$.
Consider $\tilde{e}^n$, the first $n$ symbols of the sequence,
and let $\tilde{X}^n\sim\tilde{p}(x^n)$, where
\[
\tilde{p}(x^n)=\argmax_{\substack{p(x^n):\\ (X^n,\tilde{e}^n)\in\mathcal{A}'_n}}H(X^n).
\]
Clearly $H(\tilde{X}^n)=\log(\mathbf{N}_n\cdot\mathbf{1})$.
For a fixed positive energy symbol $e_0\in\mathcal{E}$ and $\ell=\lceil\bar{B}/e_0\rceil$, we construct an  energy sequence $\mathbf{e}$ by replicating copies of $\tilde{e}^n$, followed by $\ell$ times the symbol $e_0$.
We similarly construct a process $\mathbf{X}$ by replicating i.i.d. copies of $\tilde{X}^n$ followed by $\ell$ zero symbols.
To be precise, for any $t=(i-1)(n+\ell)+j$, $i\geq1$, $1\leq j\leq n+\ell$:
\begin{align*}
e_t&=\begin{cases}
\tilde{e}_j&,1\leq j\leq n\\
e_0&,n+1\leq j\leq n+\ell
\end{cases}\\
p(x_t|x^{t-1})&=\begin{cases}
\tilde{p}(x_t|x_{(i-1)(n+\ell)+1}^{t-1})&,1\leq j\leq n\\
1(x_t=0)&,n+1\leq j\leq n+\ell
\end{cases}
\end{align*}
This construction guarantees that the energy constraints are satisfied: $(\mathbf{X},\mathbf{e})\in\mathcal{A}'_\infty$ a.s.
For this construction, we compute:
\begin{align}
J_\emptyset(\rho)&\geq
\lim_{k\to\infty}\frac{1}{k}\big(H(X^k)-\rho\sum_{t=1}^{k}e_t\big)\nonumber\\*
&=\frac{1}{n+\ell}\big(\log(\mathbf{N}_n\cdot\mathbf{1})-\rho\sum_{t=1}^{n}\tilde{e}_t-\rho\ell e_0\big).
\end{align}
The LHS does not depend on $n$ nor on the choice of sequence $\{\tilde{e}_t\}_{t=1}^{\infty}$, hence we can take $n\to\infty$ and take supremum over all sequence $\{{e}_t\}_{t=1}^{\infty}$:
\begin{equation}
J_\emptyset(\rho)\geq\sup_{\{e_t\}_{t=1}^{\infty}}\liminf_{n\to\infty}\frac{1}{n}\big(\log(\mathbf{N}_n\cdot\mathbf{1})-\rho\sum_{t=1}^{n}e_t\big).
\label{eq:J0_dir2}
\end{equation}
This inequality along with \eqref{eq:J0_dir1} implies \eqref{eq:J0}.

\section{Convergence of Optimal Finite-Horizon Rewards}
\label{sec:Lagrangian_converges}

We adopt the notation of Appendix~\ref{sec:Lagrange_multipliers_proof} to show all three cases, namely $J_\emptyset(\rho)$, $J_X(\rho)$, and $J_M(\rho)$, simultaneously.
Specifically, for every set $\mathcal{D}$ we wish to prove:
\begin{equation}
J_{\mathcal{D}}(\rho)=\lim_{n\to\infty}\frac{1}{n}\tilde{J}_{\mathcal{D}}^{(n)}(\rho),
\label{eq:Jrho_lim}
\end{equation}
where $J_{\mathcal{D}}(\rho)$ is given by~\eqref{eq:Jrho_def} and
\begin{equation}
\tilde{J}_{\mathcal{D}}^{(n)}(\rho)\triangleq
\max_{\substack{p(x^n,e^n)\in\mathcal{D}_n:\\ (X^n,E^n)\in\mathcal{A}'_n\text{ a.s.}}} \big\{H(X^n)-\rho\sum_{t=1}^{n}\mathbb{E}[E_t]\big\}.
\label{eq:Jrho_n}
\end{equation}

We prove this in a similar manner to Appendix~\ref{sec:Lagrange_multipliers_proof}.
Let $(\tilde{X}^n,\tilde{E}^n)\sim \tilde{p}(x^n,e^n)$ be the maximizing distribution in~\eqref{eq:Jrho_n},
and construct a process $(\mathbf{X},\mathbf{E})$ exactly as in Appendix~\ref{sec:Lagrange_multipliers_proof}: concatenate i.i.d. copies of $(\tilde{X}^n,\tilde{E}^n)$ followed by $\ell$ time slots with $X=0$ and $E=e_0$, for some positive $e_0\in\mathcal{E}$ and $\ell=\lceil\bar{B}/e_0\rceil$.
Then
\begin{align*}
J_{\mathcal{D}}(\rho)&\geq
\frac{1}{n+\ell}\big(H(\tilde{X}^n)-\rho\sum_{t=1}^{n}\mathbb{E}[\tilde{E}_t]-\rho\ell e_0\big)\\*
&=\frac{1}{n+\ell}({\tilde{J}_\mathcal{D}^{(n)}(\rho)}-{\rho\ell e_0})
\end{align*}
Taking $n\to\infty$ yields:
\begin{equation}
J_{\mathcal{D}}(\rho)\geq \limsup_{n\to\infty}\frac{1}{n}
{\tilde{J}_{\mathcal{D}}^{(n)}(\rho)}.
\label{eq:Jrho_n_limsup}
\end{equation}

Next, consider an arbitrary process $(\mathbf{X},\mathbf{E})$ such that $\{p(x_t,e_t|x^{t-1},e^{t-1})\}_{t=1}^{\infty}\in\mathcal{D}_\infty$ and $(\mathbf{X},\mathbf{E})\in\mathcal{A}'_{\infty}$ a.s.,
and let $(X^n,E^n)$ denote the first $n$ symbols of this process.
Clearly,
\begin{align*}
H(X^n)-\rho\sum_{t=1}^{n}\mathbb{E}[E_t]&\leq \tilde{J}_{\mathcal{D}}^{(n)}(\rho),\\
\liminf_{n\to\infty}\frac{1}{n}\big(H(X^n)-\rho\sum_{t=1}^{n}\mathbb{E}[E_t]\big)&\leq\liminf_{n\to\infty}\frac{1}{n}{\tilde{J}_{\mathcal{D}}^{(n)}(\rho)}.
\end{align*}
Since the RHS does not depend on the process $(\mathbf{X},\mathbf{E})$, we can take supremum over all such processes to obtain:
\begin{equation}
J_{\mathcal{D}}(\rho)\leq\liminf_{n\to\infty}\frac{1}{n}{\tilde{J}_{\mathcal{D}}^{(n)}(\rho)}.
\label{eq:Jrho_n_liminf}
\end{equation}
Combining~\eqref{eq:Jrho_n_limsup} and~\eqref{eq:Jrho_n_liminf} implies the limit in~\eqref{eq:Jrho_lim} exists and is equal to $J_\mathcal{D}(\rho)$.

\section{Example: Derivation of Capacity Expressions}
\label{sec:example_derivations}

\begin{figure*}[!t]
\normalsize
\setcounter{MYtempeqncnt}{\value{equation}}
\setcounter{equation}{98}
\begin{equation}
\lambda+h(b)=
\max_{e\in\{0,2\}}
\max_{\substack{p(x):\\ X\leq\min\{b+e,2\}\text{ a.s.}}}
\Big\{H(X)-\rho e+\mathbb{E}\big[h(\min\{b+e,2\}-X)\big]\Big\}
\qquad
,b=0,1,2.
\label{eq:example_CX_Bellman}
\end{equation}
\setcounter{equation}{\value{MYtempeqncnt}}
\hrulefill
\vspace*{4pt}
\end{figure*}

\subsection{Proof of Proposition~\ref{prop:example_C0}: Generic Charger}

To apply the results of Section~\ref{subsec:universal_charger_noiseless}, 
we explicitly write the MDP parameters in Table~\ref{tab:universal_charger_general_DP}.
First, observe that $\mathcal{B}=\{0,1,2\}$ and the battery state graphs are given by $G_0$ and $G_2$ in Fig.~\ref{fig:battery_state_graphs}, for $e=0$ and $e=2$ respectively. The adjacency matrices are given by
\begin{align*}
&\mathbf{A}_0=\begin{bmatrix}
1&0&0\\
1&1&0\\
1&1&1
\end{bmatrix},
\qquad
\mathbf{A}_2=\begin{bmatrix}
1&1&1\\
1&1&1\\
1&1&1
\end{bmatrix}.
\end{align*}
Writing the state vector at time $t$ as 
$\mathbf{s}_t=\big(s_t(0),\ s_t(1),\ s_t(2)\big)$,
we have the following state evolution function:
\begin{align}
\mathbf{s}_{t+1}&=f(\mathbf{s}_t,e_t)\nonumber\\*
&=
\begin{cases}
\tfrac{1}{1+s_t(1)+2s_t(2)}\big(1,\ s_t(1)+s_t(2),\ 
s_t(2)\big)&,e=0\\
\big(\tfrac{1}{3},\ \tfrac{1}{3},\ \tfrac{1}{3}\big)
&,e=2
\end{cases}
\end{align}
with the initial state being $\mathbf{s}_1=\big(0,\ 0,\ 1\big)$.
Now, unless the charging sequence is all-zeros (which will yield zero capacity), there is a finite time $t$ in which necessarily $e_t=2$.
Since the system will then move to state $\mathbf{s}=\big(\tfrac{1}{3},\ \tfrac{1}{3},\ \tfrac{1}{3}\big)$,
it is evident that beyond this point the system can only visit the following set of states: $\{\mathbf{s}^{(k)}\}_{k=1}^{\infty}$, where
\begin{align}
\mathbf{s}^{(1)}&=\big(\tfrac{1}{3},\ \tfrac{1}{3},\ \tfrac{1}{3}\big),\nonumber\\
\mathbf{s}^{(k+1)}&=f(\mathbf{s}^{(k)},0)
\qquad ,k\geq2.
\end{align}
Since the reward function is bounded, we may ignore the transient phase until the system reaches this state,
and assume the initial state is $\mathbf{s}_1=\mathbf{s}^{(1)}$ and the state space is $\mathcal{S}=\{\mathbf{s}^{(k)}\}_{k=1}^{\infty}$.
This observation will greatly simplify the solution of the Bellman equation, since we will only need to find a function $h$ with a countable domain, instead of the entire probability simplex in $\mathbb{R}^3$.

It can be easily verified that $\mathbf{s}^{(k)}\in\mathcal{S}$ can be written as
\begin{equation}
\mathbf{s}^{(k)}=\left(\frac{k}{k+2},\ \frac{2k}{(k+1)(k+2)},\ \frac{2}{(k+1)(k+2)}\right).
\end{equation}
The reward function for our reduced state space is given by
\begin{align}
g(\mathbf{s}^{(k)},e)&=\begin{cases}
\log\tfrac{k+3}{k+1}&,e=0\\
\log3-2\rho&,e=2
\end{cases}
\end{align}
To write the Bellman equation~\eqref{eq:Bellman_universal_charger}, observe that the function $h:\mathcal{S}\to\mathbb{R}$ is in fact a sequence $\{h_k\}_{k=1}^{\infty}$, where $h_k=h(\mathbf{s}^{(k)})$.
The Bellman equation can be written as follows:
\begin{equation}
\lambda+h_k=\max\{
\log\tfrac{k+3}{k+1}+h_{k+1},\ 
\log3-2\rho+h_1
\},
\qquad k\geq1,
\label{eq:Bellman_generic_charger_example}
\end{equation}
where the first term in the maximum corresponds to $e=0$ and the second term corresponds to $e=2$.

We will now find a solution to this equation.
Note that any stationary policy $e(\mathbf{s})$ consists of some (possibly empty) finite sequence of zeros, until $e=2$ for some state. Then, since the state will revert to $\mathbf{s}^{(1)}$, this pattern will repeat indefinitely.
Hence, fix an integer $\ell\geq1$. We guess a solution of the following form:
\begin{align}
\lambda+h_k&=\log\tfrac{k+3}{k+1}+h_{k+1}
,&&k=1,\ldots,\ell-1,\label{eq:advance_state}\\
\lambda+h_{\ell}&=\log3-2\rho+h_1,\label{eq:return_to_inital_state}
\end{align}
where $h_k$ for $k>\ell$ will be determined later.
The system will cycle through states $\{\mathbf{s}^{(k)}\}_{k=1}^{\ell}$.
This corresponds to the following charging sequence:
\begin{equation}
e^{(\ell)}_t=\begin{cases}
2&,t=1\mod \ell\\
0&,\text{otherwise}
\end{cases}
\label{eq:periodic_charging_sequence_l}
\end{equation}

Without loss of generality, we choose $h_1=0$.
Equations \eqref{eq:advance_state} and \eqref{eq:return_to_inital_state} are then a system of $\ell$ equations with $\ell$ variables, the solution of which is given by:
\begin{align}
\lambda&=\frac{1}{\ell}\left[\log\frac{(\ell+1)(\ell+2)}{2}-2\rho\right],\label{eq:generic_charger_example_lambda}\\
h_k&=(k-1)\lambda-\log\frac{(k+1)(k+2)}{6}
,\qquad k=1,\ldots,\ell.
\end{align}
Additionally, we choose $h_k=h_\ell$, $k>\ell$, which implies~\eqref{eq:return_to_inital_state} becomes 
\[
\lambda+h_k=\log3-2\rho+h_1,
\qquad k\geq\ell.
\]

It is left to verify that this indeed solves \eqref{eq:Bellman_generic_charger_example}. To do this, we need to show:
\begin{align}
\log\tfrac{k+3}{k+1}+h_{k+1}
&\geq\log3-2\rho+h_1,
&k&=1,\ldots,\ell-1,\\
\log\tfrac{k+3}{k+1}+h_{k+1}
&\leq\log3-2\rho+h_1,
&k&\geq\ell.
\end{align}
This gives:
\begin{align}
\rho&\lefteqn{\geq \tfrac{1}{2}\tfrac{\ell}{\ell-k}\log\tfrac{(k+1)(k+2)}{2}-\tfrac{1}{2}\tfrac{k}{\ell-k}\log\tfrac{(\ell+1)(\ell+2)}{2},}\nonumber\\*
&&k&=1,\ldots,\ell-1,
\label{eq:rho_lower_bounds_k}\\
\rho&\leq\tfrac{1}{2}\log\tfrac{(\ell+1)(\ell+2)}{2}-\tfrac{\ell}{2}\log\tfrac{k+3}{k+1},
&k&\geq\ell.
\label{eq:rho_upper_bounds_k}
\end{align}
The RHS in \eqref{eq:rho_lower_bounds_k} is increasing in $k$, whereas the RHS in \eqref{eq:rho_upper_bounds_k} is decreasing in $k$. Hence, this yields the following bounds on $\rho$:
\begin{equation}
\tfrac{1}{2}\log\tfrac{(\ell+1)(\ell+2)}{2}
-\tfrac{\ell}{2}\log\tfrac{\ell+2}{\ell}
\leq\rho\leq
\tfrac{1}{2}\log\tfrac{(\ell+1)(\ell+2)}{2}
-\tfrac{\ell}{2}\log\tfrac{\ell+3}{\ell+1}.
\label{eq:generic_charger_example_rho_bounds}
\end{equation}
For any $\rho$ that satisfies these inequalities, the charging sequence $\{e_t^{(\ell)}\}_{t\geq1}$ in~\eqref{eq:periodic_charging_sequence_l} is optimal.
By Lemma~\ref{lemma:Lagrange_multipliers},
this yields $\Gamma=\frac{2}{\ell}$, and capacity is given by
\begin{align}
C_{\emptyset}(\tfrac{2}{\ell})&=J_\emptyset(\rho)+\rho\Gamma\nonumber\\*
&=\frac{1}{\ell}\log\frac{(\ell+1)(\ell+2)}{2}
\label{eq:generic_charger_Cl}
\end{align}
where $J_\emptyset(\rho)=\lambda$ is given by~\eqref{eq:generic_charger_example_lambda}.
This yields~\eqref{eq:example_C0_integer}.

To obtain $C_{\emptyset}(\Gamma)$ for all $0\leq\Gamma\leq2$, let $\rho=\tfrac{1}{2}\log\tfrac{(\ell+1)(\ell+2)}{2}
-\tfrac{\ell}{2}\log\tfrac{\ell+3}{\ell+1}$, the upper bound in~\eqref{eq:generic_charger_example_rho_bounds}, for some integer $\ell\geq1$.
It can be verified from \eqref{eq:generic_charger_example_rho_bounds} that in this case, both sequences $\{e_t^{(\ell)}\}_{t\geq1}$ and $\{e_t^{(\ell+1)}\}_{t\geq1}$ are solutions to the Bellman equation \eqref{eq:Bellman_generic_charger_example}.
Hence, time-sharing between these two sequences yields an optimal solution.
To apply the result of Lemma~\ref{lemma:Lagrange_multipliers}, this corresponds to $\underline{\Gamma}=\frac{2}{\ell+1}$ and $\overline{\Gamma}=\frac{2}{\ell}$.
For any $0\leq\alpha\leq1$, this yields 
\[
\Gamma=\alpha\frac{2}{\ell+1}+(1-{\alpha})\frac{2}{\ell},
\]
\[
J_\emptyset(\rho)=\log\tfrac{\ell+3}{\ell+1},
\]
\begin{align*}
C_\emptyset(\Gamma)
&=J_\emptyset(\rho)+\rho(\alpha\underline{\Gamma}+(1-\alpha)\overline{\Gamma})\\
&=\alpha\frac{1}{\ell+1}\log\frac{(\ell+2)(\ell+3)}{2}\\*
&\qquad+(1-\alpha)\frac{1}{\ell}\log\frac{(\ell+1)(\ell+2)}{2}\\
&=\alpha C_\emptyset(\tfrac{2}{\ell+1})+(1-\alpha)C_\emptyset(\tfrac{2}{\ell}).
\end{align*}
This yields~\eqref{eq:example_C0_Gamma}.
For $\Gamma=0$, the capacity is trivially $C_\emptyset(0)=0$, and for $\Gamma>2$ we can achieve the unconstrained capacity $\log3$ by letting $e_t=2$ for all $t$.
\qed

\subsection{Proof of Proposition~\ref{prop:example_CX}: Charger Adjacent to Transmitter}
\label{subsec:ternary_channel_binary_energy_CX}

\begin{figure*}[!b]
\vspace*{4pt}
\hrulefill
\normalsize
\setcounter{MYtempeqncnt}{\value{equation}}
\setcounter{equation}{106}
\begin{align}
p^\star(x|b=0)=p^\star(x|b=2)
&=\frac{1}{{1+2^{2\rho+1}+\sqrt{1+2^{2\rho+2}}}}\left(
{2^{2\rho+1}},\ 
{-1+\sqrt{1+2^{2\rho+2}}},\ {2}\right),
\label{eq:example_CX_largerho_pb02}\\
p^\star(x|b=1)&=\frac{1}{1+\sqrt{1+2^{2\rho+2}}}\left(
-1+\sqrt{1+2^{2\rho+2}},\ 2,\ 0\right).
\label{eq:example_CX_largerho_pb1}
\end{align}
\setcounter{equation}{\value{MYtempeqncnt}}
\end{figure*}

We compute the capacity using the result of Section \ref{subsec:charger_adj_Tx_noiseless}.
To this end, we will need to find a solution for \eqref{eq:DP_Bellman_CX_general}, which in this case is given by~\eqref{eq:example_CX_Bellman} at the top of the page.
\addtocounter{equation}{1}

The maximization over $p(x)$ can be easily solved using Lemma~\ref{lemma:max_entropy_plus_linear} in Appendix \ref{sec:proof_of_lemma_max_entropy_plus_linear}.
Without loss of generality we choose $h(0)=0$, leaving 3 equations for $\lambda$, $h(1)$, $h(2)$:
\begin{align}
\lambda&=\max\big\{0,\ \log(1+2^{h(1)}+2^{h(2)})-2\rho\big\},
	\label{eq:DP_Bellman_CX_h0}\\
\lambda+h(1)&=\max\big\{\log(1+2^{h(1)}),\nonumber\\*
	&\hspace{4.5em} \log(1+2^{h(1)}+2^{h(2)})-2\rho\big\},
	\label{eq:DP_Bellman_CX_h1}\\
\lambda+h(2)&=\log(1+2^{h(1)}+2^{h(2)}),
	\label{eq:DP_Bellman_CX_h2}
\end{align}
where in \eqref{eq:DP_Bellman_CX_h2} it is evident that $e^\star(2)=0$ regardless of $\rho$.

The solution can be separated into 3 different regions of the parameter $\rho$. Each region determines the maximum in \eqref{eq:DP_Bellman_CX_h0} and \eqref{eq:DP_Bellman_CX_h1}, and therefore also $e^\star(b)$. The solution $\lambda$, $h(b)$ follows.
\begin{enumerate}
\item $0\leq\rho<\frac{1}{2}\log\left(\frac{1+2^{h(1)}+2^{h(2)}}{1+2^{h(1)}}\right)$\\[\medskipamount]
The optimal solution is
$e^\star(b)=(2,\ 2,\ 0)$,
$h(b)=(0,\ 0,\ 2\rho)$,
and $\lambda=\log(2+2^{2\rho})-2\rho$.
It can be seen that the equivalent range of $\rho$ is $0\leq \rho<1/2$.
%

\item $\frac{1}{2}\log\left(\frac{1+2^{h(1)}+2^{h(2)}}{1+2^{h(1)}}\right)\leq\rho<\frac{1}{2}\log(1+2^{h(1)}+2^{h(2)})$\\[\medskipamount]
In this case we get $e^\star(b)=(2,\ 0,\ 0)$,
$h(b)=(0,\ \log(-1+\sqrt{1+2^{2\rho+2}})-1,\ 2\rho)$,
and $\lambda=2\log(1+\sqrt{1+2^{2\rho+2}})-2\rho-2$.
Substituting in the range of $\rho$ under consideration, it becomes simply $\rho\geq1/2$.
%

\item $\frac{1}{2}\log(1+2^{h(1)}+2^{h(2)})\leq\rho$\\[\medskipamount]
There is no feasible solution for values of $\rho$ in this range.
\end{enumerate}
In conclusion, we get the following expressions for $J(\rho)$:
\begin{equation}
J(\rho)=\begin{cases}
\log(2+2^{2\rho})-2\rho&,0\leq\rho<1/2\\
2\log(1+\sqrt{1+2^{2\rho+2}})-2\rho-2&,1/2\leq\rho
\end{cases}
\label{eq:J_rho_final}
\end{equation}

Next, we wish to find $C_X(\Gamma)$.
We do so for each range of values of $\rho$ separately:
\begin{enumerate}
\item $0\leq\rho<1/2$\\[\medskipamount]
Since the optimal strategy is to always charge when the battery is not full (i.e. when $b<2$),
the optimal input distribution for each $b$ is the same, and is given by (see Appendix~\ref{sec:proof_of_lemma_max_entropy_plus_linear}):
\begin{align*}
p^\star(x)=\left(\tfrac{2^{2\rho}}{2+2^{2\rho}},\tfrac{1}{2+2^{2\rho}}
,\tfrac{1}{2+2^{2\rho}}\right).
\end{align*}
Since the $X_t$'s are i.i.d., we can readily compute the rate:
\begin{align}
C_X(\rho)
&=H(X)\nonumber\\
&=\log(2+2^{2\rho})-2\rho\cdot\frac{2^{2\rho}}{2+2^{2\rho}}.
\label{eq:CX_rho_smallrho}
\end{align}
Now, from~\eqref{eq:J_rho_final} and $J_X(\rho)={C}_X(\rho)-\rho\Gamma(\rho)$:
\begin{equation}
\Gamma(\rho)=\frac{4}{2+2^{2\rho}}.
\label{eq:Gamma_rho_smallrho}
\end{equation}
First, observe that $0\leq\rho<1/2$ translates to $1<\Gamma\leq4/3$.
Next, solving for $\rho$ 
and substituting in \eqref{eq:CX_rho_smallrho}
yields:
\begin{equation}
{C}_X(\Gamma)=\tfrac{\Gamma}{2}+H_2\left(\tfrac{\Gamma}{2}\right).
\label{eq:CX_Gamma_smallrho}
\end{equation}

\item $1/2<\rho$\\[\medskipamount]
The optimal input distribution in this case is given by~\eqref{eq:example_CX_largerho_pb02} and~\eqref{eq:example_CX_largerho_pb1} at the bottom of the page.\addtocounter{equation}{2}
We will use $p^\star(x|b)$ and the battery evolution equation~\eqref{eq:EH_battery} to find the stationary distribution $p^\star(b)$ of the Markov chain $B_t$.
We have a set of linear equations for $p^\star(b)$, $b=0,1,2$:
\begin{align*}
p^\star(b)&=\sum_{b'\in\mathcal{B}}p^\star(b')
	\sum_{x\in\mathcal{X}}p^\star(x|b')
	1_{\{b=\min\{b'+e^\star(b'),2\}-x\}}\\
&=\sum_{b'\in\mathcal{B}}p^\star(b')
	p(x=\min\{b'+e^\star(b'),2\}-b\ |b').
\end{align*}
Solving for $p^\star(b=0)$ yields:
\begin{align*}
p^\star(b=0)
&=\frac{1}{\sqrt{1+2^{2\rho+2}}}.
\end{align*}
Since $E_t$ is a deterministic function of $B_t$, we can readily compute
\begin{align}
\Gamma(\rho)
&=\lim_{n\to\infty}\frac{1}{n}\sum_{t=1}^{n}\mathbb{E}[E_t]\nonumber\\*
&=\sum_{b\in\mathcal{B}}p^\star(b)e^\star(b)\nonumber\\
&=2\cdot p^\star(b=0)\nonumber\\
&=\frac{2}{\sqrt{1+2^{2\rho+2}}}.
\label{eq:Gamma_rho_largerho}
\end{align}
The resulting range of $\Gamma$ values for $\rho>1/2$ is $0<\Gamma<2/3$.
Solving \eqref{eq:Gamma_rho_largerho} for $\rho$ 
and substituting in~\eqref{eq:J_rho_final} together with ${C}_X(\rho)=J_X(\rho)+\rho\Gamma(\rho)$ we obtain:
\begin{equation}
{C}_X(\Gamma)=
\big(1+\tfrac{\Gamma}{2}\big)\log\tfrac{2+\Gamma}{2\Gamma}
	-\big(1-\tfrac{\Gamma}{2}\big)\log\tfrac{2-\Gamma}{2\Gamma}.
\label{eq:CX_Gamma_largerho}
\end{equation}

\item $\rho=1/2$\\[\medskipamount]
In the two previous cases we derived ${C}_X(\Gamma)$ for $0<\Gamma<2/3$ and $1<\Gamma\leq 4/3$.
The remaining range, $2/3\leq\Gamma\leq1$, can be attained with $\rho=1/2$.
To see this, note that the optimal $e^\star(b)$ for either one of the previous cases is optimal here.
Therefore, time-sharing is optimal here, and we apply the result of Lemma~\ref{lemma:Lagrange_multipliers} with $\underline{\Gamma}=2/3$ and $\overline{\Gamma}=1$.
From \eqref{eq:J_rho_final} we have $J_X(\rho)=1$.
Hence, for $\Gamma=\alpha\underline{\Gamma}+(1-\alpha)\overline{\Gamma}$ with $0\leq\alpha\leq 1$:
\begin{align}
{C}_X(\Gamma)
&=J_X(\rho)+\rho\Gamma\nonumber\\
&=\frac{1}{2}\Gamma+1.
\label{eq:CX_Gamma_timesharing}
\end{align}
\end{enumerate}
Finally, we obviously have ${C}_X(0)=0$, and for any $\Gamma\geq4/3$ we can achieve ${C}_X(\Gamma)=\log 3$, which is the capacity without an input constraint, by using the optimal policy for $\rho=0$.
Collecting~\eqref{eq:CX_Gamma_smallrho}, \eqref{eq:CX_Gamma_largerho}, and~\eqref{eq:CX_Gamma_timesharing}, we obtain~\eqref{eq:example_CX_Gamma}.\qed

\subsection{Proof of Proposition~\ref{prop:example_CM}: Fully Cognitive Charger}
\label{subsec:CM_example_derivation}

Before applying the results of Section~\ref{subsec:charger_knows_M_noiseless}, we make use of a unique structure of this channel to simplify the MDP formulation.
We return to the original capacity expression in \eqref{eq:CM}.
First, we claim
that the optimal input distribution is of the form $p(x^n,e^n)=p(x^n)1(e^n|x^n)$, i.e. $e^n$ is a deterministic function of $x^n$.
In other words, we claim that capacity can equivalently be written as follows:
\begin{equation}
C_M(\Gamma)=\lim_{n\to\infty}\frac{1}{n}
\max_{\substack{p(x^n),\ f:\mathcal{X}^n\to\mathcal{E}^n:\\
	(X^n,f(X^n))\in\mathcal{A}_n(\Gamma)\text{ a.s.}}}
H(X^n).
\label{eq:CM_original}
\end{equation}
\begin{proof}
We will show for every $n\geq1$:
\begin{equation}
\max_{\substack{p(x^n,e^n):\\ 
	(X^n,E^n)\in\mathcal{A}_n(\Gamma)\text{ a.s.}}}
	H(X^n)
=\max_{\substack{p(x^n),\ f:\mathcal{X}^n\to\mathcal{E}^n:\\
	(X^n,f(X^n))\in\mathcal{A}_n(\Gamma)\text{ a.s.}}}
	H(X^n).
\end{equation}

Clearly,
\[
\max_{\substack{p(x^n,e^n):\\ 
	(X^n,E^n)\in\mathcal{A}_n(\Gamma)\text{ a.s.}}}
	H(X^n)
\geq
\max_{\substack{p(x^n),\ f:\mathcal{X}^n\to\mathcal{E}^n:\\
	(X^n,f(X^n))\in\mathcal{A}_n(\Gamma)\text{ a.s.}}}
	H(X^n).
\]
To show the other direction, fix $p(x^n,e^n)$ such that $(X^n,E^n)\in\mathcal{A}_n(\Gamma)$ a.s.,
and let $p(x^n)=\sum_{e^n}p(x^n,e^n)$ be the marginal.
Then, for any $x^n$ with positive probability $p(x^n)>0$, there exists at least one $e^n$ such that $(x^n,e^n)\in\mathcal{A}_n(\Gamma)$.
Otherwise, $p(x^n,e^n)=0$ for every $e^n$ which will imply $p(x^n)=0$.
Define a function $f$ by mapping every such $x^n$ to some $e^n$ for which $(x^n,e^n)\in\mathcal{A}_n(\Gamma)$.
The mapping $f$ for $x^n$ that have zero probability can be set to any arbitrary $e^n$.
Using $f$, define a joint probability distribution
$\tilde{p}(x^n,e^n)=p(x^n)1_{\{e^n=f(x^n)\}}$.
Clearly, the marginal is
\[
\tilde{p}(x^n)=\sum_{e^n}\tilde{p}(x^n,e^n)=p(x^n).
\]
Since $H(X^n)$ depends only on the marginal $p(x^n)$, we have shown that any distribution $p(x^n,e^n)$ can be replaced by a distribution $p(x^n)1_{\{e^n=f(x^n)\}}$ without changing the objective.
This concludes the proof.
\end{proof}

From now on, we will only be concerned with functions $f$ that satisfy $(x^n,f(x^n))\in\mathcal{A}'_n$ for all $x^n\in\mathcal{X}^n$.
This is because for any pair $(f,p(x^n))$ that satisfies $(X^n,f(X^n))\in\mathcal{A}_n(\Gamma)$ a.s., the only way a particular $x^n$ will have $(x^n,f(x^n))\not\in\mathcal{A}'_n$ is if $p(x^n)=0$. Since this will not affect the probability $\Pr\left(\sum_{t=1}^{n}E_t\leq n\Gamma\right)$, we can equivalently choose $f(x^n)=(2,\ldots,2)$ for any such~$x^n$.

Denote the $\ell_1\text{-norm}$ of an energy sequence $\|f(x^n)\|_1=\sum_{t=1}^{n}e_t$, where $e^n=f(x^n)$.
We call a function $f^\star:\mathcal{X}^n\to\mathcal{E}^n$ \emph{minimal}, if, for any other function $f:\mathcal{X}^n\to\mathcal{E}^n$, we have
\[
\|f^\star(x^n)\|_1\leq \|f(x^n)\|_1
\qquad\forall x^n\in\mathcal{X}^n,
\]
In other words, a minimal $f^\star$ will give an energy sequence with the smallest possible cost for all $x^n\in\mathcal{X}^n$, among all functions that satisfy the energy constraints.
Note that there may be more than one such minimal function.

We make the following claim:
\begin{lemma}
\label{lemma:minimal_function_capacity}
Let $f^{(n)}:\mathcal{X}^n\to\mathcal{E}^n$, $n\geq1$, be a sequence of minimal functions that satisfy the energy constraints for every $x^n$. Then $C_M(\Gamma)=C^\star_M(\Gamma)$, where
\begin{equation}
{C}^\star_M(\Gamma)=\lim_{n\to\infty}\frac{1}{n}
\max_{\substack{p(x^n):\\ 
	\|f^{(n)}(X^n)\|_1\leq n\Gamma\text{ a.s.}}}
H(X^n).
\label{eq:CM_minimal}
\end{equation}
\end{lemma}
\begin{proof}
Clearly $C^\star_M(\Gamma)\leq C_M(\Gamma)$, since the optimization domain in~\eqref{eq:CM_minimal} is a subset of that in~\eqref{eq:CM_original}.
For the other direction, fix $n$ and let $f$ and $p(x^n)$ be the optimal pair in~\eqref{eq:CM_original}. By changing $f$ to a minimal $f^{(n)}$, the energy constraints are still satisfied while the total energy cost $\sum_{t=1}^{n}E_t$ can only decrease.
Therefore $C_M(\Gamma)\leq C_M^\star(\Gamma)$.
\end{proof}

In what follows we will describe a way to find one such minimal $f$.
It turns out that, in this case, we can find a minimal $f$ which depends on $x^n$ in a causal manner, namely $e_t=f_t(x^t)$.
In fact, this minimal $f$ will depend only on $x_t$ and the current battery state $b_t$ (which is a function of $b_{t-1}$, $x_{t-1}$, and $e_{t-1}$).
This will be used to simplify the MDP formulation. Specifically, it will allow us to consider a finite action space, instead of the entire probability simplex in $\mathbb{R}^{|\mathcal{B}|}$.

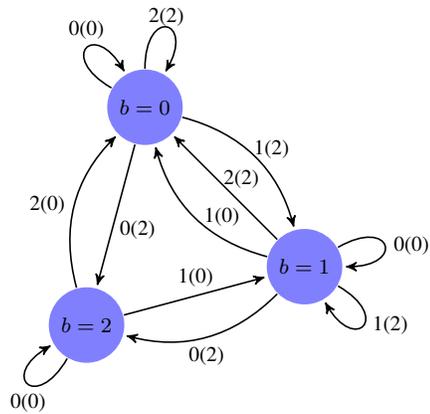
\begin{figure}
\centering
\def \nodedistance {3cm}
\begin{tikzpicture}[->,>=stealth',shorten >=1pt,semithick,font=\footnotesize]
\tikzstyle{every state}=[draw=none,fill=blue!50]
\node[state] (b2) {$b=2$};
\node[state] (b1) at (15:\nodedistance) {$b=1$};
\node[state] (b0) at (75:\nodedistance) {$b=0$};

\path
(b0) edge [in=60,out=90,loop] node [above=-3] {2(2)} ()
	 edge [in=120,out=150,loop] node [above] {0(0)} ()
	 edge [bend left] node [above right=-4] {1(2)} (b1)
	 edge node [below right=-2] {0(2)} (b2)
(b1) edge [in=-60,out=-30,loop] node [right] {1(2)} ()
	 edge [in=0,out=30,loop] node [right] {0(0)} ()
     edge node [above right=-4] {2(2)} (b0)
     edge [bend left] node [right] {1(0)} (b0)
     edge [bend left] node [below] {0(2)} (b2)
(b2) edge [in=-150,out=-120,loop] node [below] {0(0)} (b2)
	 edge [bend left] node [left] {2(0)} (b0)
	 edge node [above] {1(0)} (b1);
\end{tikzpicture}
\caption{Graph for the construction of codeword pairs $(x^n,e^n)$ for the noiseless channel with $\mathcal{X}=\{0,1,2\}$ and $\mathcal{E}=\{0,2\}$.
The nodes represent the battery state $b_t$, and the edges are labeled $x_t (e_t)$.
Any pair $(x^n,e^n)$ that satisfies the energy constraints can be produced by a length $n$ walk on the graph.
}
\label{fig:CM_original_graph}
\end{figure}
\begin{figure}
\centering
\def \nodedistance {3cm}
\begin{tikzpicture}[->,>=stealth',shorten >=1pt,semithick,font=\footnotesize]
\tikzstyle{every state}=[draw=none,fill=blue!50]
\node[state] (b2) {$b=2$};
\node[state] (b1) at (15:\nodedistance) {$b=1$};
\node[state] (b0) at (75:\nodedistance) {$b=0$};

\path
(b0) edge [in=60,out=90,loop] node [above] {2(2)} ()
	 edge [in=120,out=150,loop] node [left] {0(0)} ()
	 edge [bend left] node [above right=-4] {1(2)} (b1)
(b1) edge [in=-30,out=0,loop] node [below] {0(0)} ()
     edge node [below left=-5] {2(2)} (b0)
     edge [bend left] node [below left=-1] {1(0)} (b0)
(b2) edge [in=-150,out=-120,loop] node [below] {0(0)} (b2)
	 edge [bend left] node [left] {2(0)} (b0)
	 edge [bend right] node [below] {1(0)} (b1);
\end{tikzpicture}
\caption{Pruned graphs for the noiseless channel with $\mathcal{X}=\{0,1,2\}$ and $\mathcal{E}=\{0,2\}$.
The nodes represent the battery state $b_t$, and the edges are labeled $x_t (e_t)$.
Any input sequence $x^n$ and an energy sequence $e^n=f(x^n)$, where $f$ is minimal, can be produced by walking $n$ steps on the graph.
}
\label{fig:CM_pruned_graph}
\end{figure}
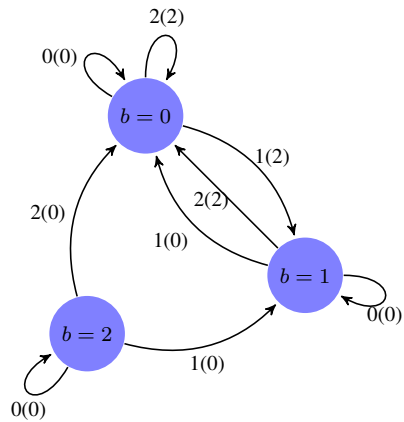

See Fig.~\ref{fig:CM_original_graph}.
Any pair $(x^n,e^n)$ can be produced by a length $n$ walk on this graph.
Notice that we have excluded edges coming out of state $b=2$ with cost $e=2$, since clearly it is never optimal to charge the battery if it is already fully charged.
We will show by a series of pruning steps, i.e. removing certain edges of the graph, that we can reduce it to the graph in Fig.~\ref{fig:CM_pruned_graph} while maintaining the fact that $e^n=f(x^n)$ with $f$ minimal.

For $x^n\in\mathcal{X}^n$ and $b\in\{0,1,2\}$,
define $V_b(x^n)$ as the minimal cost $\sum_{t=1}^{n}e_t$ of an energy sequence $e^n$ such that $(x^n,e^n)\in\mathcal{A}'_n$ and the initial battery level is $b_1=b$:
\[
V_b(x^n)=\min_{\substack{e^n:\\
	 (x^n,e^n)\in\mathcal{A}'_n\\ b_1=b}}
\sum_{t=1}^{n}e_t.
\]
The minimal cost $V_b(x^n)$ can be computed recursively using Bellman's principle of optimality~\cite[Ch.~1.3]{Bertsekas2001vol1}:
\[
V_b(x_t^n)=\min_{\substack{e\in\mathcal{E}:\\
	x_t\leq\min\{b+e,\bar{B}\}}}
\big\{e+V_{b_{t+1}}(x_{t+1}^n)\big\},
\qquad b=0,1,2,
\]
where $b_{t+1}=\min\{b+e,\bar{B}\}-x_t$, a deterministic function of $b$, $e$, and $x_t$.
This means that given the current battery level $b_t$, the desired input symbol $x_t$, and the optimal cost-to-go $\sum_{i=t+1}^{n}e_i$ for every possible value of $b_{t+1}$,  
we need to choose $e_t$ such that the total cost $\sum_{i=t}^{n}e_i$ is minimized.
This gives the following equations for $V_b(x_t^n)$ (we adopt the notation $x_t^n=x_tx_{t+1}^n$):
\begin{align}
x_t=0:
&&	V_0(0x_{t+1}^n)&=\min\big\{V_0(x_{t+1}^n),\ 
		V_2(x_{t+1}^n)+2\big\}
		\label{eq:V_b0_x0}\\
&&	V_1(0x_{t+1}^n)&=\min\big\{V_1(x_{t+1}^n),\ 
		V_2(x_{t+1}^n)+2\big\}
		\label{eq:V_b1_x0}\\
&&	V_2(0x_{t+1}^n)&=V_2(x_{t+1}^n)
		\label{eq:V_b2_x0}\\[\medskipamount]
x_t=1:
&&	V_0(1x_{t+1}^n)&=V_1(x_{t+1}^n)+2
		\label{eq:V_b0_x1}\\
&&	V_1(1x_{t+1}^n)&=\min\big\{V_0(x_{t+1}^n),\ 
		V_1(x_{t+1}^n)+2\big\}
		\label{eq:V_b1_x1}\\
&&	V_2(1x_{t+1}^n)&=V_1(x_{t+1}^n)
		\label{eq:V_b2_x1}\\[\medskipamount]
x_t=2:
&&	V_0(2x_{t+1}^n)&=V_0(x_{t+1}^n)+2
		\label{eq:V_b0_x2}\\
&&	V_1(2x_{t+1}^n)&=V_0(x_{t+1}^n)+2
		\label{eq:V_b1_x2}\\
&&	V_2(2x_{t+1}^n)&=V_0(x_{t+1}^n)
		\label{eq:V_b2_x2}
\end{align}

We start with eq.~\eqref{eq:V_b0_x0}.
We will show that $V_0(x_{t+1}^n)\leq V_2(x_{t+1}^n)+2$ for any value of $x_{t+1}^n$, which will allow us to remove an edge from the graph (specifically the edge labeled 0(2) going from $b=0$ to $b=2$).
First, from~\eqref{eq:V_b0_x1} and~\eqref{eq:V_b2_x1},
$V_0(1x_{t+2}^n)=V_2(1x_{t+2}^n)+2$.
Second, from~\eqref{eq:V_b0_x2} and~\eqref{eq:V_b2_x2},
$V_0(2x_{t+2}^n)=V_2(2x_{t+2}^n)+2$.
Finally, from~\eqref{eq:V_b0_x0} and~\eqref{eq:V_b2_x0},
\[V_0(0x_{t+2}^n)\leq V_2(x_{t+2}^n)+2=V_2(0x_{t+2}^n)+2.\]
We conclude that $V_0(x_{t+1}^n)\leq V_2(x_{t+1}^n)+2$ for all $x_{t+1}$, hence we can replace~\eqref{eq:V_b0_x0} with
\begin{equation}
V_0(0x_{t+1}^n)=V_0(x_{t+1}^n).
\label{eq:V_b0_x0_new}
\end{equation}
We continue with eq.~\eqref{eq:V_b1_x0}.
By applying similar arguments, we can show that 
$V_1(x_{t+1}^n)\leq V_2(x_{t+1}^n)+2$ for all $x_{t+1}$, therefore~\eqref{eq:V_b1_x0} can be replaced with
\begin{equation}
V_1(0x_{t+1}^n)=V_1(x_{t+1}^n),
\label{eq:V_b1_x0_new}
\end{equation}
and the edge labeled 0(2) going from $b=1$ to $b=2$ can be removed from the graph.
We are left with eq.~\eqref{eq:V_b1_x1}; in this case we show by induction that $V_1(x_{t}^n)\leq V_0(x_{t}^n)\leq V_1(x_{t}^n)+2$ for all $t=1,\ldots,n$.
Starting with $t=n+1$, we have $V_0(\emptyset)=V_1(\emptyset)=0$.
Assume 
\[V_1(x_{t+1}^{n})\leq V_0(x_{t+1}^{n})\leq V_1(x_{t+1}^{n})+2.\]
From~\eqref{eq:V_b0_x2} and~\eqref{eq:V_b1_x2},
\[V_1(2x_{t+1}^n)=V_0(2x_{t+1}^n)\leq V_1(2x_{t+1}^n)+2.\]
From~\eqref{eq:V_b0_x0_new} and~\eqref{eq:V_b1_x0_new},
\[V_1(0x_{t+1}^n)\leq V_0(0x_{t+1}^n)\leq V_1(0x_{t+1}^n)+2.\]
Finally, since the induction assumption implies
\[
V_0(x_{t+1}^n)\leq V_1(x_{t+1}^n)+2\leq V_0(x_{t+1}^n)+2,
\]
and by~\eqref{eq:V_b1_x1}, $V_1(1x_{t+1}^n)=V_0(x_{t+1}^n)$. Along with~\eqref{eq:V_b0_x1}, we get
\[
V_1(1x_{t+1}^n)\leq V_0(1x_{t+1}^n)\leq V_1(1x_{t+1}^n)+2.
\]
We conclude that eq.~\eqref{eq:V_b1_x1} can be replaced with
\begin{equation}
V_1(1x_{t+1}^n)=V_0(x_{t+1}^n),
\label{eq:V_b1_x1_new}
\end{equation}
and the loop edge labeled 1(2) going from $b=1$ to itself can be removed.

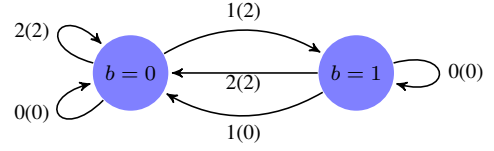
\begin{figure}
\centering
\def \nodedistance {3cm}
\begin{tikzpicture}[->,>=stealth',shorten >=1pt,semithick,font=\footnotesize]
\tikzstyle{every state}=[draw=none,fill=blue!50]
\node[state] (b0) at (0,0) {$b=0$};
\node[state] (b1) at (\nodedistance,0) {$b=1$};

\path
(b0) edge [in=135,out=165,loop] node [left] {2(2)} ()
	 edge [in=195,out=225,loop] node [left] {0(0)} ()
	 edge [bend left] node [above] {1(2)} (b1)
(b1) edge [in=-15,out=15,loop] node [right] {0(0)} ()
     edge node [below=-3] {2(2)} (b0)
     edge [bend left] node [below] {1(0)} (b0);
\end{tikzpicture}
\caption{Final two-state graph for construction of $x^n$ and minimal $e^n$.
The nodes represent the battery state $b_t$, and the edges are labeled $x_t(e_t)$.
Since each node has exactly 3 outgoing edges, one for each $x\in\mathcal{X}$, this defines a function $e^\star$ s.t. $e_t=e^\star(b_t,x_t)$.
Moreover, the next state $b_{t+1}$ is determined by $b_t$ and $x_t$, hence this defines the function $b_{t+1}=f_B(b_t,x_t)$.
}
\label{fig:CM_codeword_two_state_graph}
\end{figure}

\begin{figure*}[!b]
\vspace*{4pt}
\hrulefill
\normalsize
\setcounter{MYtempeqncnt}{\value{equation}}
\setcounter{equation}{127}
\begin{align}
\lambda+h(0)&=\max_{(p_0,p_1,p_2)}\big\{H_3(p_0,p_1,p_2)
	-2\rho(p_1+p_2)+(p_0+p_2)h(0)+p_1h(1)\big\}
	\label{eq:example_CM_Bellman_b0}\\
\lambda+h(1)&=\max_{(p_0,p_1,p_2)}\big\{H_3(p_0,p_1,p_2)
	-2\rho p_2+(p_1+p_2)h(0)+p_0h(1)\big\}
	\label{eq:example_CM_Bellman_b1}
\end{align}
\setcounter{equation}{\value{MYtempeqncnt}}
\end{figure*}

We are left with the graph in Fig.~\ref{fig:CM_pruned_graph}.
We have essentially shown that it is enough to consider codewords pairs $(x^n,e^n)$ produced by this graph without loss of optimality.
Note that each node has exactly 3 outgoing nodes, one for each $x\in\mathcal{X}$.
Therefore, it defines a \emph{minimal} function $e_t=e^\star(b_t,x_t)$.
By Lemma~\ref{lemma:minimal_function_capacity}:
\[
C_M(\Gamma)=\lim_{n\to\infty}\frac{1}{n}\
\max_{\substack{p(x^n):\\ \sum_{t=1}^{n}e^\star(B_t,X_t)\leq n\Gamma}}
H(X^n).
\]

Returning to our MDP formulation of Section~\ref{subsec:charger_knows_M_noiseless}, specifically Table~\ref{tab:charger_knows_M_general_DP}, we see that the action, which is a probability distribution $p(x_t,e_t|b_t)$, can be written as
\[
p(x_t,e_t|b_t)=p(x_t|b_t)1(e_t|b_t,x_t),
\]
where $1(e_t|b_t,x_t)=1_{\{e_t=e^\star(b_t,x_t)\}}$ is fixed.
Hence the action is simply $u_t(x_t|b_t)=p(x_t|b_t)$.
Note that the action space is the set of all stochastic matrices $p(x|b)$, since the energy constraints are satisfied completely by $e^\star(b,x)$ for every $b,x$.

With this simplified action, we will now show that we can in fact reduce the state space to a finite set of points.
Recall that we assume $B_1=2$, therefore the initial state is $\mathbf{s}_1=(0,0, 1)$.
We will show by induction that the state at any time $t$ must be a singleton, i.e. 
$s_t(b_t)=1_{\{b_t=b'\}}$ for some $b'\in\mathcal{B}$.
Assuming the above for time $t$, we have $s_{t+1}$ from~\eqref{eq:CM_DP_state_dynamics}:
\begin{align*}
s_{t+1}(b_{t+1})
&=\frac{\sum_{e_t,b_t}s_t(b_t)u_t(x_t|b_t)1(e_t|b_t,x_t)1(b_{t+1}|b_t,x_t,e_t)}
	{\sum_{e_t,b_t}s_t(b_t)u_t(x_t|b_t)1(e_t|b_t,x_t)}\\*
&=\frac{u_t(x_t|b')1(b_{t+1}|b',x_t,e^\star(b',x_t))}{u_t(x_t|b')}\\
&=1(b_{t+1}|b',x_t,e^\star(b',x_t))\\
&=1_{\{b_{t+1}=b''\}},
\end{align*}
where $b''=\min\{b'+e^\star(b',x_t),2\}-x_t$.
We conclude that the state space can be reduced to the finite set
\[
\mathcal{S}=\big\{(1,0,0),\ (0,1,0),\ (0,0,1)\big\}.
\]
Now, since $\mathbf{s}_t$ is a degenerate probability distribution of $b_t$, we can replace the state with $b_t$, and the state space with $\mathcal{B}=\{0,1,2\}$.
Moreover, the action can be reduced to a probability distribution over $\mathcal{X}$: $u_t(x_t)=p(x_t)$.

The state transition function $b_{t+1}=f_B(b_t,x_t)$ depends only on the state and the disturbance, and is given by the state transition graph in Fig.~\ref{fig:CM_pruned_graph}.
Observe that the node $b=2$ has no incoming edges, i.e. it is a transient state. 
The system starts in this state,
however, unless $X_t=0$ w.p. 1 (which will give zero rate), it will stay in this state for a finite amount of time. We can therefore remove this state without decreasing capacity, and the remaining state diagram is depicted in Fig.~\ref{fig:CM_codeword_two_state_graph}.

The simplified MDP formulation is summarized in Table \ref{tab:CM_example_DP}.
We are now ready to apply Theorem~\ref{thm:Bellman}, obtaining the following Bellman equation for $b=0,1$:
\begin{equation}
\lambda+h(b)=\max_{p(x)}\big\{
H(X)-\rho\mathbb{E}[e^\star(b,X)]
+h(f_B(b,X))
\big\}.
\label{eq:example_CM_Belmman}
\end{equation}
Let $(p_0,p_1,p_2)$ denote a probability distribution over $\mathcal{X}=\{0,1,2\}$ and 
$H_3(p_0,p_1,p_2)\triangleq-\sum_{i=0}^{2}p_i\log p_i$
denote the ternary entropy function.
Then \eqref{eq:example_CM_Belmman} yields \eqref{eq:example_CM_Bellman_b0} and~\eqref{eq:example_CM_Bellman_b1} at the bottom of the page.\addtocounter{equation}{2}

\begin{table}
\renewcommand{\arraystretch}{1.5}
\caption{Simplified MDP Formulation of Capacity with Fully Cognitive Charger, for input and energy alphabets $\mathcal{X}=\{0,1,2\}$ and $\mathcal{E}=\{0,2\}$}
\label{tab:CM_example_DP}
\centering
\begin{tabular}{|l|p{4.5cm}|}
\hline
state
&$b_t$, the state of the transmitter's battery\\
\hline
state space& $\mathcal{S}=\{0,1\}$\\
\hline
action&$\mathbf{u}_t=[u_t(x):x\in\mathcal{X}]$, a probability distribution of the input\\
\hline
action space&$\mathcal{U}$, the probability simplex in $\mathbb{R}^{|\mathcal{X}|}$\\
\hline
reward&$g(b,\mathbf{u})=H(X)-\rho\mathbb{E}[e^\star(b,X)]$,
where $e^\star(b,x)$ is given by Fig.~\ref{fig:CM_codeword_two_state_graph}\\
\hline
disturbance&$x_t$, the channel input\\
\hline
disturbance distribution&$p(x_t)$, determined by $\mathbf{u}_t$\\
\hline
state dynamics&
$b_{t+1}=f_B(b_t,x_t)$ according to Fig.~\ref{fig:CM_codeword_two_state_graph}\\
\hline
\end{tabular}
\end{table}

Choosing $h(0)=0$ and applying Lemma~\ref{lemma:max_entropy_plus_linear} in Appendix~\ref{sec:proof_of_lemma_max_entropy_plus_linear}:
\begin{align*}
\lambda&=\log\big(1+2^{h(1)-2\rho}+2^{-2\rho}\big),\\
\lambda+h(1)&=\log\big(2^{h(1)}+1+2^{-2\rho}\big),
\end{align*}
which gives
\begin{align}
J_M(\rho)=\lambda&=\log\big(2+2^{-2\rho}+2^{-2\rho}\sqrt{5+2^{2\rho+2}}\big)-1,
	\label{eq:Jrho_CM}\\
h(1)&=\log\big(-1+\sqrt{5+2^{2\rho+2}}\big)-1,
\end{align}
and the maximizing probability distributions $(p_0^\star,p_1^\star,p_2^\star)$ are
\begin{align}
p^\star(x|b=0)&=
	\frac{\left(2^{1+2\rho},\ 
		-1+\sqrt{5+2^{2+2\rho}},\ 
		2\right)}
	{1+2^{1+2\rho}+\sqrt{5+2^{2+2\rho}}}
	,\label{eq:P_X_B0}\\
p^\star(x|b=1)&=\frac{\left(-1+\sqrt{5+2^{2+2\rho}},\ 
		2,\ 2^{1-2\rho}\right)}
	{1+2^{1-2\rho}+\sqrt{5+2^{2+2\rho}}}.\label{eq:P_X_B1}
\end{align}

Next, we will compute $\underline{\Gamma}=\liminf_{n\to\infty}\frac{1}{n}\sum_{t=1}^{n}\mathbb{E}[E_t]$ and $\overline{\Gamma}=\limsup_{n\to\infty}\frac{1}{n}\sum_{t=1}^{n}\mathbb{E}[E_t]$.
We will do so by first finding the steady state distribution of $B_t$ using the state transitions depicted in Fig.~\ref{fig:CM_codeword_two_state_graph}:
\begin{align*}
p^\star(b=1)&=p^\star(x=1|b=0)\cdot (1-p^\star(b=1))\nonumber\\*
&\qquad +p^\star(x=0|b=1)\cdot p^\star(b=1).
\end{align*}
Solving for $p^\star(b=1)$, substituting~\eqref{eq:P_X_B0} and~\eqref{eq:P_X_B1}, and simplifying:
\begin{equation}
p^\star(b=1)=\frac{1}{2}\left(1-\frac{1}{\sqrt{5+2^{2\rho+2}}}\right).
\label{eq:PB1_steady_state}
\end{equation}
Next, according to Fig.~\ref{fig:CM_codeword_two_state_graph}, we have $\mathbb{E}[E_t|B_t=0]=2\big(1-p^\star(x=0|b=0)\big)$ and $\mathbb{E}[E_t|B_t=1]=2p^\star(x=2|b=1)$.
Substituting~\eqref{eq:P_X_B0}--\eqref{eq:PB1_steady_state} and simplifying yields $\underline{\Gamma}=\overline{\Gamma}=\Gamma(\rho)$, where:
\begin{align}
\Gamma(\rho)
&=\sum_{b=0,1}p^\star(b)\mathbb{E}[E_t|B_t=b]\nonumber\\
&=\frac{1}{2^{4\rho}-1}\left(
	\frac{2^{2\rho+1}(2^{2\rho}+2)}{\sqrt{5+2^{2\rho+2}}}
	-2\right).
\label{eq:Gamma_rho_CM}
\end{align}
Observe that for $\rho\geq0$ we have $0<\Gamma(\rho)\leq10/9$.

To find the capacity as a function of $\Gamma$, it remains to solve for $\rho(\Gamma)$ and substitute in~\eqref{eq:Jrho_CM}.
Letting $\zeta=\sqrt{5+2^{2\rho+2}}$, eq.~\eqref{eq:Gamma_rho_CM} becomes:
\[
\Gamma=\frac{2(\zeta^2+2\zeta+5)}{\zeta(\zeta-1)(\zeta+3)}.
\]
Rearranging, we get the cubic equation~\eqref{eq:zeta_cubic_equation} for $\zeta$.
For $0<\Gamma\leq10/9$, there is exactly one real root to this equation.
Substituting~\eqref{eq:Jrho_CM} in $C_M(\Gamma)=J_M(\rho)+\rho\Gamma(\rho)$, we get~\eqref{eq:example_CM_Gamma}.
Finally, for the remaining values of $\Gamma$: Clearly $C_M(0)=0$.
For $\Gamma>10/9$, we apply the optimal solution obtained from $\rho=0$.
This gives $C_M(10/9)=\log 3$, which is the unconstrained capacity.\qed

\section{Maximization of Sum of Entropy and Linear Function}
\label{sec:proof_of_lemma_max_entropy_plus_linear}

\begin{lemma}
\label{lemma:max_entropy_plus_linear}
Let $X\in\mathcal{X}$ be a discrete finite-alphabet random variable, and $a(x)$ be a real function. Then
\[
\max_{p(x)}\{H(X)-\mathbb{E}[a(X)]\}=\log\Xi,
\]
where $\Xi=\sum_{x\in\mathcal{X}}2^{-a(x)}$
and the optimal distribution is 
$p(x)=2^{-a(x)}/\Xi$.
\end{lemma}
\begin{proof}
Applying KKT conditions, the Lagrangian is:
\[
L=\sum_{x\in\mathcal{X}}p(x)[-\log p(x)-a(x)+\lambda(x)-\nu].
\]
Taking derivative w.r.t. $p(x)$:
\[
\frac{\partial L}{\partial p(x)}=
-\log p(x)-a(x)+\lambda(x)-\nu-\log e=0,
\]
\[
p(x)=2^{-a(x)+\lambda(x)-{\nu}-\log e}.
\]
Since $p(x)>0$, by complementary slackness $\lambda(x)=0$.
Since $\sum_{x\in\mathcal{X}}p(x)=1$, we get
$p(x)=2^{-a(x)}/\Xi$, where $\Xi=\sum_{x\in\mathcal{X}}2^{-a(x)}$.
Substituting the optimal $p(x)$ in the objective, we get the result of Lemma~\ref{lemma:max_entropy_plus_linear}.
\end{proof}

\bibliographystyle{IEEEtran}
\bibliography{remotely_powered_comm}

\begin{IEEEbiographynophoto}{Dor Shaviv}
(S'13) received his B.Sc. (summa cum laude) degrees in electrical engineering and physics from the Technion---Israel Institute of Technology, Haifa, Israel, in 2007.
He received his M.Sc. degree in electrical engineering from the same institution in 2012.
During 2007--2013 he worked as an R\&D engineer in the Israel Defense Forces.
He is currently a Ph.D. candidate in the Electrical Engineering Department at Stanford University, and is a recipient of a Robert Bosch Stanford Graduate Fellowship.
\end{IEEEbiographynophoto}

\begin{IEEEbiographynophoto}{Ayfer \"Ozg\"ur}
(M'06) received her B.Sc. degrees in electrical engineering and physics from Middle East Technical University, Turkey, in 2001 and the M.Sc. degree in communications from the same university in 2004. From 2001 to 2004, she worked as hardware engineer for the Defense Industries Development Institute in Turkey. She received her Ph.D. degree in 2009 from the Information Processing Group at EPFL, Switzerland. In 2010 and 2011, she was a post-doctoral scholar with the Algorithmic Research in Network Information Group at EPFL. She is currently an Assistant Professor in the Electrical Engineering Department at Stanford University. Her research interests include network communications, wireless systems, and information and coding theory. Dr.~\"Ozg\"ur received the EPFL Best Ph.D. Thesis Award in 2010 and a NSF CAREER award in 2013.
\end{IEEEbiographynophoto}

\begin{IEEEbiographynophoto}{Haim H. Permuter}
(M'08-SM'13) received his B.Sc.\@ (summa cum laude) and M.Sc.\@
(summa cum laude) degrees in Electrical and Computer Engineering
from the Ben-Gurion University, Israel, in 1997 and 2003,
respectively, and the Ph.D. degree in Electrical Engineering from
Stanford University, California in 2008.

Between 1997 and 2004, he was an officer at a research and
development unit of the Israeli Defense Forces. Since 2009 he is
with the department of Electrical and Computer Engineering at
Ben-Gurion University where he is currently an associate professor.

Prof. Permuter is a recipient of several awards, among them the
Fullbright Fellowship, the Stanford Graduate Fellowship (SGF), Allon
Fellowship, and the U.S.-Israel Binational Science Foundation
Bergmann Memorial Award. Haim is currently serving on the editorial
boards of the IEEE Transactions on Information Theory.
\end{IEEEbiographynophoto}

\end{document}